\documentclass[runningheads]{llncs}
\usepackage{geometry}
\usepackage{amsthm}
\usepackage{amsmath}

\usepackage{booktabs} 
\usepackage{multirow} 
\usepackage{adjustbox} 
\usepackage{rotating}
\usepackage{caption} 
\usepackage{subcaption} 
\usepackage{appendix} 
\usepackage{graphicx}
\usepackage{pgfplots}
\usepackage[all]{nowidow}
\usepackage[utf8]{inputenc}
\usepackage{tikz}
\usetikzlibrary{er,positioning,bayesnet}
\usepackage{multicol}
\usepackage{algpseudocode,algorithm,algorithmicx}
\usepackage{hyperref}
\usepackage[inline]{enumitem} 
\usepackage{rotating}
\usepackage[utf8]{inputenc}
\usepackage{diagbox} 
\usepackage{mdframed} 

\definecolor{blue}{HTML}{1F77B4}
\definecolor{orange}{HTML}{FF7F0E}
\definecolor{green}{HTML}{2CA02C}

\pgfplotsset{compat=1.14}

\usepackage[sectionbib]{natbib} 
\usepackage{amsfonts} 
\usepackage{bm} 
\usepackage{bbm} 
\newcommand{\argmin}[1]{\underset{#1}{\operatorname{arg}\operatorname{min}}\;}
\newcommand{\argmax}[1]{\underset{#1}{\operatorname{arg}\operatorname{max}}\;}
\newcommand\norm[1]{\left\lVert#1\right\rVert}
\newcommand{\abs}[1]{\left\vert{#1}\right\vert}

\usepackage{amssymb} 
\newtheorem*{assumption*}{\assumptionnumber}
\providecommand{\assumptionnumber}{}
\makeatletter
\newenvironment{assumption}[1]
 {%
  \renewcommand{\assumptionnumber}{Assumption #1}%
  \begin{assumption*}%
  \protected@edef\@currentlabel{$#1$}%
 }
 {%
  \end{assumption*}
 }
\makeatother

\usepackage[normalem]{ulem}

\usepackage{setspace}
\begin{document}
%
\title{Measuring tail risk at high-frequency: An $L_1$-regularized extreme value regression approach with unit-root predictors}
%
\author{Julien Hambuckers$^{\dagger}$, Li Sun$^{\dagger}$\thanks{Corresponding author: ir.li.sun@gmail.com} and Luca Trapin$^{\ddagger}$}

%
%
\institute{$\dagger$ University of Liège - HEC Li\`{e}ge, Belgium
\\
$\ddagger$ University of Bologna, Italy
}
\maketitle              
\begin{abstract}


We study tail risk dynamics in high-frequency financial markets and their connection with trading activity and market uncertainty. We introduce a dynamic extreme value regression model accommodating both stationary and local unit-root predictors to appropriately capture the time-varying behaviour of the distribution of high-frequency extreme losses. To characterize trading activity and market uncertainty, we consider several volatility and liquidity predictors, and propose a two-step adaptive $L_1$-regularized maximum likelihood estimator to select the most appropriate ones. We establish the oracle property of the proposed estimator for selecting both stationary and local unit-root predictors, and show its good finite sample properties in an extensive simulation study. Studying the high-frequency extreme losses of nine large liquid U.S. stocks using 42 liquidity and volatility predictors, we find the severity of extreme losses to be well predicted by low levels of price impact in period of high volatility of liquidity and volatility.

\end{abstract}
\section{Introduction}


Measuring tail risk at high-frequency has become of utmost importance to market players and regulators \citep{weller2018}. While much efforts have been devoted to the measurement of tail risk at low-frequency \citep{nieto2016frontiers}, few attempts have been made to measure risk at high-frequency, see \cite{giot2005market}, \cite{dionne2009intraday} and \cite{chavez2012modelling}. Moreover, although these models can be very accurate, they explain the tail risk evolution in a ``reduced form'' manner, i.e., using autoregressive terms exploiting the persistence of the time series. They thus fail to provide a deeper structural understanding of the factors driving tail risk. As much as understanding the macroeconomic determinants of tail risk is a relevant problem at low-frequency \citep{massacci2017tail}, it is important to understand how market uncertainty and trading activity impacts tail risk at high-frequency. 

From a market microstructure perspective, though the intensification of high-frequency trading has improved trading costs and liquidity \citep{hendershott2011}, it is also suspected to be responsible for more frequent extreme price movements over short periods of time \citep{brogaard2018}. Such extreme fluctuations are often the result of an aggressive directional market making activity initiated when the market is already under stress. \cite{brogaard2018} find that market wide extreme shocks are likely to trigger the risk controls of high-frequency liquidity providers that thus withdraw from the market to reduce their risk exposure. Similarly, \cite{kirilenko2017flash} find that during the market turbulence induced by the 2010 Flash Crash, many high-frequency liquidity providers withdrew from the market, thus exacerbating the price fall. Studying how market uncertainty and trading activity affect extreme losses can thus provide a deeper understanding of the evolution of tail risk at high-frequency, and this paper proposes appropriate econometric techniques to do so.

We consider a dynamic extreme value regression framework \citep{chavez2016extreme,massacci2017tail,schwaab2021modeling} where the distribution of extreme losses is assumed to be well approximated by a generalized Pareto distribution (GPD) with time-varying parameters driven by exogenous preditors and autoregressive terms. To assess the impact of market uncertainty and trading activity on extreme losses, we consider several volatility predictors, proxing for market uncertainty, and liquidity predictors, characterizing trading activity. Despite extreme value regression techniques have been widely applied in finance \citep{chavez2016extreme, hambuckers2018understanding,bee2019realized}, our investigation presents new challenges: (i) as the financial literature proposes several volatility and liquidity measures, we face a variable selection problem aimed at identifying predictors capturing the most relevant aspects of trading activity affecting extremes as well as improving the predictive accuracy of tail risk; (ii) volatility and liquidity measures observed at high-frequency exhibit strong persistence and seasonalities, thus violating the classical stationary assumptions required for inference with the maximum likelihood estimator (MLE). To overcome these issues, we develop a two-step adaptive $L_1$-regularized maximum likelihood estimator (ALMLE) that allows performing variable selection with both stationary and local unit-root predictors \citep{lee2021LASSO}, and establish its oracle property.

We investigate the impact of 42 liquidity and volatility indicators on the distribution of high-frequency extreme losses of nine large liquid U.S. stocks observed from 2006 to 2014. We find that the severity of tail risk, as measured by the shape parameter of the GPD, is well predicted by low price impact \citep{goyenko2009liquidity} during periods of high volatility of volatility and high volatility of liquidity. This finding is coherent with the evidence in \cite{brogaard2018} that market markers liquidity supply is outstripped by liquidity demand after large uncertainty shocks, and their rush to leave the market to lower their risk exposures amplify extreme price movements. Our two-step ALMLE is necessary to reveal this pattern as the standard MLE finds almost all predictors to be significant. To validate our estimating strategy, we provide an out-of-sample VaR forecast analysis and find that the estimated model performs well in the out-of-sample.
The remainder of the paper is organized as follows: Section~\ref{sec:model_specifcation} presents the time-varying GPD model accommodating stationary and local unit-root predictors as well as autoregressive components; Section~\ref{sec:MLE_Inference} presents the MLE and shows its asymptotic non-normality when local unit-root predictors are included in the model;
Section~\ref{sec:LASSO} introduces the two-step ALMLE and prove the oracle property of this estimator in selecting both stationary and local unit-root predictors; Section~\ref{sec:simulation} provides an extensive simulation study comparing the performance of the two-step ALMLE to those of the MLE, showing the superiority of the former in finite samples. Section~\ref{sec:empirical_study} discusses the results of the empirical study whereas Section~\ref{sec:conclusion} concludes. Additional results and mathematical proofs are relegated to the Appendix.
\section{Extreme value regression}
\label{sec:model_specifcation}
We denote the logarithmic loss and return time series of a financial asset by $\{ l_t\}_{t=1}^T$ and $\{ r_t\}_{t=1}^T$, respectively, with $l_t = - r_t$, and denote $\bm{z}_t$ a vector of exogenous predictors observed at time $t$.
\begin{assumption}{M.1}
\label{assump:prob_space}
 $\{l_t\}_{t=1}^T$ and $\{\bm{z}_t\}_{t=1}^T$ are on a complete probability space $(\Omega, \mathcal{F}, P)$. At each time $t \in\{ 1,2,\ldots,T\}$, we have an information set $\mathcal{F}_{t-1}$ available which is the $\sigma$-algebra generated by $\{\bm{z}_{t-1}, l_{t-1}, \bm{z}_{t-2},$ $l_{t-2}, \ldots \}$. 
\end{assumption}

Let assume $\{ l_t\}_{t=1}^T$ is independent and identically distributed (i.i.d.) with a cumulative distribution function (c.d.f.) $F(\cdot)$. Probabilistic results from extreme value theory show that if there exist real sequences $a_T > 0$ and $\beta_T$ such that $\lim_{T\rightarrow \infty} F^T\left( a_T\,x + \beta_T\right)$ converges to a non-degenerate distribution $G(\cdot)$, then $F(\cdot)$ belongs to the max-domain of attraction of $G(\cdot)$, i.e. $F\in \mathcal{D}(G)$, and $G(\cdot)$ must be the generalized extreme value (GEV)  distribution (see Theorem 3.1.1. of \cite{coles2001introduction}).

Let $\{ y_t\}_{t=1}^T$ be a censored sequence of excess losses above a high threshold $u$, such that the excess loss $y_t = l_t - u$, if $l_t > u$, and $y_t = 0$ otherwise. Define the conditional distribution of excess losses,
$$
F_{|u}(y) := P\left\{ l_t - u \leq y \middle| l_t > u\right\} = P\left\{ y_t \leq y \middle| y_t > 0\right\}, \quad 0< y_t\leq L^F - u,
$$
with $L^F := \sup\{x: F(x) < 1\}$ the right end point of $F(\cdot)$. \cite{pickands1975statistical} and \cite{balkema1974residual} show that if $F(\cdot)\in \mathcal{D}(G)$ then the limiting distribution of $F_{|u}(y)$ is a GPD, i.e.
\begin{equation}
\label{thm:POT}
  \lim_{u\rightarrow +\infty} \, \sup_{0 < y < +\infty} |F_{|u}(y) - \text{GPD}(y;k,\sigma_u)| = 0,
\end{equation}
where $\text{GPD}(\cdot; k,\sigma)$ denotes the GPD with shape parameter $k\in\mathbb{R}$ and scale parameter $\sigma>0$, 
\begin{equation}\label{eq:gpd_dist}
    \text{GPD}(y;k,\sigma) = 1 - \left(1 + k\frac{y}{\sigma}\right)^{-1/k}.
\end{equation}
Eq.~\eqref{thm:POT} suggests that $F_{|u}(y)$ with $u$ large enough can be approximated by a $\text{GPD}(\cdot; k,\sigma_{u})$, where the scale parameter $\sigma_u$ depends on $u$. The peaks-over-threshold (POT) approach assumes this relationship holds exactly above a fixed threshold $u$ and uses the exceedances of such threshold to estimate the GPD parameters $\sigma$ and $k$ (see section 4.3 of \cite{coles2001introduction}).

\subsection{Time-varying peaks-over-threshold (POT) approach}
The classical POT approach assumes that $\{l_t\}$ is i.i.d. However, financial data typically exhibit dependence features such as time-varying heteroscedasticity and extremal clustering that violate this assumption. To capture these aspects, we adopt a dynamic POT approach. Let $\{ y_t\}_{t=1}^T$ be a censored sequence of excess losses over a threshold time series $\{u_t\}_{t=1}^T$, we model the excess loss distribution conditional on the information set $\mathcal{F}_{t-1}$, $F_{t|u_t}(y):=P\{ y_t\leq y | y_t>0, \mathcal{F}_{t-1}\}$, using a GPD with time-varying parameters $k_t$ and $\sigma_t$. See, e.g., \cite{chavez2014extreme,massacci2017tail,bee2019realized}. 
Consider the vector-valued time series of $p\in \mathbb{N}$ explanatory variables $\{\bm{z}_t := [z_{1,t},\ldots,z_{p,t}]'\}^T_{t=1}$. Given the information set $\mathcal{F}_{t-1}$, we consider the following specification for $\{ (k_t, \sigma_{t})\}$,
\begin{align}
\label{Model:k_sigma_models}
& \log\left(\frac{ k_t }{ 0.5 - k_t } \right)  =	\beta_{1,0} + \sum_{j=1}^{p} \beta_{1,j}\,z_{j,t-1} ,
\\
\label{Model:k_sigma_models_2ndEqu}
& \log(\sigma_t) 	=	\beta_{2,0} + \sum_{j=1}^{p}\beta_{2,j}\,z_{j,t-1}  + \beta_{2,p+1}\,\log(\sigma_{t-1}).
\end{align}
We impose that $0<k_t<0.5$ and $\sigma_t>0$ (see ~\cite{hosking1987parameter}) to ensure a finite conditional variance of $y_t$ and numerical stability in the estimation. 
As the scale parameter $\sigma_t$ can be associated with the variance of the underlying distribution $F_{t}(\cdot)$, we accommodate an autoregressive term in $\log(\sigma_t)$ in the spirit of GARCH models \citep{engle2001garch}.
We allow for both stationary and unit-root explanatory variables in \eqref{Model:k_sigma_models} and \eqref{Model:k_sigma_models_2ndEqu}, such that persistent predictors can be accommodated. 

\section{Maximum likelihood estimation}
\label{sec:MLE_Inference}
Let $\bm{\beta}: = [\beta_{1,0}, \beta_{1,1}, \ldots, \beta_{1,p}, \beta_{2,0}, \beta_{2,1}, \ldots, \beta_{2,p+1}]'$ denote the vector of the model coefficients in \eqref{Model:k_sigma_models}-\eqref{Model:k_sigma_models_2ndEqu}, and define the coefficient space $\Theta$ of $\bm{\beta}$ as a subspace of $\mathbb{R}^{2p + 2}\times(-1,1)$ accomodating all permissible coefficient vectors $\bm{\beta}$. We present the MLE of the model coefficients in \eqref{Model:k_sigma_models}-\eqref{Model:k_sigma_models_2ndEqu} and show it is consistent but asymptotically non-normal when local unit-root explanatory variables are included in the model. 

\subsection{Maximum likelihood estimator}
\label{subsec:mle_intro}

\begin{assumption}{M.2}
\label{assump:correct_model_speci}
We assume that for a given $\{u_t\}$, the conditional c.d.f. $F_{t}(\cdot)$ of $l_t$ given $\mathcal{F}_{t-1}$ exists for $t\in \{1,2,\ldots, T\}$ and $y_t:=l_t-u_t>0$ follows a time-varying GPD, i.e.
\begin{equation}
F_{t|u_t}(y_t) = \text{GPD}(y_t;k_t,\sigma_{t}) =
1 - \left( 1 + k_t\frac{y_t}{\sigma_t} \right)^{-\frac{1}{k_t}},
\label{eq:conditional_GPD}
\end{equation}
where $\{k_t\}$ and $\{\sigma_t\}$ are specified by \eqref{Model:k_sigma_models}-\eqref{Model:k_sigma_models_2ndEqu} with the true coefficient vector $\bm{\beta}^o\in \Theta\subset \mathbb{R}^{2p + 2}\times(-1,1)$. Moreover, $\{u_t\}$ returns a constant unconditional exceedance rate, i.e., $P\{y_t>0\}=\tau$ for all $t\in\{1,\ldots,T\}$ with a constant $\tau$ close to zero.
\end{assumption}

\begin{assumption}{M.3}
\label{assump:linear_processes_Zt}
Among the explanatory variables in Model~\eqref{Model:k_sigma_models}-\eqref{Model:k_sigma_models_2ndEqu}, we assume that $\{ z_{i,t}, i = 1, \ldots, p_0\} \in I(0)$ and $\{ z_{j,t}, j = p_0+1, \ldots, p\} \in I(1)$ with $\epsilon_{j,t} := z_{j,t} - z_{j,t-1} $ and $\{ \epsilon_{j,t}, j = p_0+1, \ldots, p\} \in I(0)$. We denote by $I(0)$ and $I(1)$ the set of stationary and unit-root predictors, respectively.
\end{assumption}

Under Assumption~\ref{assump:correct_model_speci}, the conditional probability density function (p.d.f.) of $y_t |\{y_t>0, \mathcal{F}_{t-1}\}$ is
\begin{equation}
f_t(y_t) = 
\frac{1}{\sigma_t}\left( 1 + k_t\frac{y_t}{\sigma_t} \right)^{-\frac{1}{k_t}-1},
\label{eq:conditional_GP_pdf}
\end{equation}
and the log-likelihood function $L(\cdot)$ of $\{y_t| y_t>0, \mathcal{F}_{t-1}\}$ can be defined as \citep{schwaab2021modeling},
\begin{equation}
\begin{aligned}
\mathcal{L}(\bm{\beta}; \{ y_t \}, \{ \bm{z}_{t-1} \})	&  = \sum_{t=1}^{T}\mathbbm{1}\{y_t>0\} \, \log(f_t(y_t))
				\\
				& = \sum_{t=1}^{T} \mathbbm{1}\{y_t>0\}\,\left( -\log(\sigma_t) - \left( \frac{1}{k_t} +1 \right)\, \log(1 + k_t\frac{y_t}{\sigma_t}) \right),
\end{aligned}
\label{eq:likelihood_function}
\end{equation}
where
\begin{equation}
\left\{
\begin{aligned}
& k_t(\bm{\beta})  = 0.5 \left( 1 + \exp\left(-(\beta_{1,0} + \sum_{j=1}^{p} \beta_{1,j}\,z_{j,t-1} ) \right) \right)^{-1},
\\
& \sigma_t(\bm{\beta}) 	=	\exp\left(\beta_{2,0} + \sum_{j=1}^{p}\beta_{2,j}\,z_{j,t-1} +  \beta_{2,p+1}\,\log(\sigma_{t-1}) \right),
\end{aligned}
\right.
\label{Model:k_sigma_models_rewrite}
\end{equation}
for $t = 1,2,\ldots,T$, with $\mathbbm{1}\{\cdot\}$ the indicator function taking value one if the input is true and zero otherwise. 

We consider standardized predictors $\{\bm{z}^*_t\}$ in the estimation to get stochastically bounded variables, i.e., for each $t\in\{1,\ldots,T\}$, we standardize $\bm{z}_t$ as follows,
\begin{equation}
    \bm{z}_t^*:=  [\; z_{1,t},\ldots, z_{p_0,t}, \frac{z_{p_0+1,t}}{\sqrt{T}}, \ldots, \frac{z_{p,t}}{\sqrt{T}} \;]'.
\end{equation}
Replacing $\{\bm{z}_t\}$ with $\{\bm{z}^*_t\}$ into the likelihood function in \eqref{eq:likelihood_function} and maximizing we obtain 
\begin{equation}
\widehat{\bm{\beta}}^{\text{mle}} = \argmax{\bm{\beta}\in \Theta } \mathcal{L}(\bm{\beta}; \{ y_t \}, \{ \bm{z}^*_{t-1} \}),	
\label{eq:ML_estimator}
\end{equation}
where $\Theta \subset \mathbb{R}^{2p + 2}\times(-1,1)$. We denote the corresponding vector of true coefficients $\bm{\beta}^{o*}$.

\subsubsection{Remark.}
Assumption \ref{assump:correct_model_speci} assumes a constant unconditional probability for the exceedance $\mathbbm{1}\{y_t>0\}$ for $t\in\{1,\ldots,T\}$, which is more general than assuming a constant conditional probability for $\mathbbm{1}\{y_t>0|\mathcal{F}_{t-1}\}$. This causes us no extra burden to obtain the limiting behaviour of the MLE because $\mathbbm{1}\{y_t>0\}$ is bounded and not a function of the model coefficients.
Assumption~\ref{assump:linear_processes_Zt} allows for both stationary and unit-root predictors among $\bm{z}_t$.

\subsection{Asymptotic properties of the MLE}
\label{subsec:limiting_mle}
\cite{smith1985maximum} establishes the asymptotic properties of the MLE of a GPD with constant $k$ and $\sigma$ in an i.i.d. setting. We extend \cite{smith1985maximum} establishing the consistency and limiting distribution of the MLE of the dynamic GPD with stationary and unit-root predictors in \eqref{eq:ML_estimator}. In what follows, we list the assumptions required to derive the asymptotic behaviour of the MLE, and establish the consistency and limiting distribution of $\widehat{\bm{\beta}}$.

\begin{assumption}{M.4}
\label{assump:magnitude_of_coefs} 
We assume that,
\begin{equation}
\left\{
\begin{aligned}
    & \beta_{s,i}^{o*}:= \beta_{s,i}^o = O(1) , \quad\text{for}\quad i = 0, 1, \ldots, p_0, \;\text{and}\; s = 1,2 ;
    \\
    &
    \beta_{s,j}^{o*}:=\sqrt{T}\beta_{s,j}^o = O(1), 
    \quad\text{for}\quad j = p_0+1,\ldots, p_0 + p \;\text{and}\; s = 1,2 ;
    \\
   &
    \beta_{2,p+1}^{o*}:=\beta_{2,p+1}^o \in (-1,1),
\end{aligned} 
\right.
\end{equation}
and $\bm{\beta}^{o*}:= [\beta_{1,0}^{o*}, \ldots, \beta_{1,p}^{o*}, \beta_{2,0}^{o*}, \ldots, \beta_{2,p+1}^{o*}]'\in\mathbb{R}^{2p+2}\times(-1,1)$.
\end{assumption}

\begin{assumption}{M.5}
\label{assump:limiting_Zt}
$\{\bm{\epsilon}_t := \left[ z_{1,t}, \ldots, z_{p_0,t}, \epsilon_{p_0+1,t}, \ldots, \epsilon_{p,t} \right]'\}_{t=1}^T$ is assumed $\text{i.i.d.}\,(\bm{0}, \Sigma^{(0)}) $ with mean $\bm{0}$ and positive definite covariance matrix $\Sigma^{(0)}$. With $
    \left\{\bm{z}_t^* 
    := 
    [\; z_{1,t},\ldots, z_{p_0,t}, \frac{z_{p_0+1,t}}{\sqrt{T}}, \ldots, \frac{z_{p,t}}{\sqrt{T}} \;]'\right\}
$, we assume that as $T\rightarrow\infty$, we have that
\begin{equation*}
\begin{aligned}
    & (1) 
    \left\{ 
    \begin{aligned}
    & 
    \frac{1}{\sqrt{T}}\sum_{t=1}^{T} z_{i,t}^* = O_p(1), \quad 
    \frac{1}{T}\sum_{t=1}^{T}  (z^*_{i,t})^2 = O_p(1), \quad i = 1,2,\ldots,p_0; 
    \\
    &
    \frac{1}{T}\sum_{t=1}^{T} z^*_{j,t} = O_p(1), \quad 
    \frac{1}{T}\sum_{t=1}^{T} (z^*_{j,t})^2 = O_p(1), \quad j = (p_0+1), \ldots,p;
    \\
    & 
    \frac{1}{T} \sum_{t=1}^{T} \bm{z}_t^* \bm{z}_t^{*'} \text{ is positive definite in probability one} ;
    \end{aligned}
    \right.
    \\
    & \text{and there exists a positive definite matrix $\Sigma:=[\Sigma_{i,j}]_{i,j = 1,\dots, p}$ such that }
    \\
    & (2)
    \left\{
    \begin{aligned}
    & 
    \lim_{T\rightarrow\infty}\frac{1}{T}\sum_{t=1}^{T} z^*_{i,t} = 0,
    \quad 
    \lim_{T\rightarrow\infty}\frac{1}{T}\sum_{t=1}^{T} (z^*_{i,t})^2 = 1, 
    \quad
    \frac{1}{\sqrt{T}}\sum_{t=1}^{t_0} z^*_{i,t} \overset{D}{\sim} \,W_i( \frac{t_0}{T}), \quad i = 1,2,\ldots,p_0; 
    \\
    &
    \frac{1}{T}\sum_{t=1}^{T} z^*_{j,t}  \overset{D}{\sim} \int^1_0 \Sigma_{j,j}^{1/2}\, W_j(t)\,dt,
    \quad
    \frac{1}{T}\sum_{t=1}^{T} (z^*_{j,t})^2 \overset{D}{\sim} \int^1_0 \Sigma_{j,j} W_j^2( t)\,dt, 
    \qquad j = (p_0+1), \ldots, p; 
    \\
    &
    \frac{1}{T}\sum_{t=1}^{T} \bm{z}_t^* \bm{z}_t^{*'} \overset{D}{\sim} \int^1_0 \left(\Sigma^{1/2}\,\bm{W_{z*}}(t)\right)\,\left(\Sigma^{1/2}\,\bm{W_{z*}}(t)\right)'\,dt,
    \\
    & 
    \frac{1}{T} \sum_{t=1}^{T}  \bm{z}_t^* \bm{z}_t^{*'} \text{ is positive definite in probability one};
    \end{aligned}
    \right.
\end{aligned}    
\end{equation*}    
where $1\leq t_0 \leq T$, and $W_i(\cdot), W_j(\cdot)$ are independent Brownian motions. And denote \lq\, $\overset{D}{\sim}$ \rq \,for convergence in distribution and $\bm{W_{z*}}(t) : = [\frac{d W_1(t)}{\sqrt{dt}}, \ldots, \frac{d W_{p_0}(t)}{\sqrt{dt}}, W_{p_0+1}(t), \ldots, W_{p}(t)]' $. 
\end{assumption}

\begin{assumption}{M.6}
\label{assump:narrow_Theta}
$\Theta$ is a compact subspace of $\mathbb{R}^{2p+2}\times(-1,1)$ containing the true coefficient vector $\bm{\beta}^{o*}$ such that $ H_{\mathcal{L}}(\cdot)$ is positive definite in $\Theta$ almost surely. 
\end{assumption}


\begin{assumption}{M.7}
\label{assump:limiting_score_Z}
Given Assumptions~\ref{assump:linear_processes_Zt} and \ref{assump:limiting_Zt}, we further assume that as $T \rightarrow \infty$, it holds that 
\begin{equation}
  \frac{1}{\sqrt{T}} \sum^T_{t=1}\, \bm{\psi}_t(\bm{\beta}^{o*})
  \overset{\mathcal{D}}{\sim}
  S_{\psi},
\end{equation}
where $S_{\psi}$ is a non-degenerate distribution:
\begin{equation}
   S_{\psi} = \Sigma^{1/2}_{\psi} \int_0^1\; 
   \begin{bmatrix}
    dW_{g_k}(t)
   \\
    \bm{W_{z*}}(t) dW_{g_k}(t)
    \\
    dW_{g_{\sigma}}(t)
    \\
    \bm{W_{z*}}(t) dW_{g_{\sigma}}(t)
    \\
    d W_{ar}(t)\,dW_{g_{\sigma}}(t)
    \end{bmatrix}
\end{equation}
where $\Sigma_{\psi}$ is a positive definite $(2p+3)\times(2p+3)$ matrix and $\{W_{g_k}(t)\}, \{W_{g_{\sigma}}(t)\}, \{W_{ar}(t)\}$ are Brownian motions.
\end{assumption}

\begin{assumption}{M.8}
\label{assump:weakconverge_HL}
$\frac{H_L(\bm{\beta}^{o*}; \{l_t\}, \{\bm{z}^*_t\})}{T}$ at $\bm{\beta}^{o*}$ is assumed to weakly converge to a stochastic integral $\Omega_{H}$, i.e.,
\begin{equation}
    \label{def:HL_weakconverge}
     \frac{H_{\mathcal{L}}( \bm{\beta}^{o*}; \{l_t\}, \{\bm{z}_t^*\}}{T} \overset{D}{\sim} \Omega_{H},
     \qquad \text{when } T\rightarrow\infty,
\end{equation}
where $\Omega_{H}^{-1}$ exists upon the limiting behaviours of ${\bm{z}^*_t}$ in Assumption~\ref{assump:limiting_Zt}. 
\end{assumption}

\subsubsection{Remark.}
Assumption~\ref{assump:magnitude_of_coefs} imposes the orders of magnitude of $\bm{\beta}^{o}$ to ensure that the unit-root explanatory variables $z_{p_0+1,t}, \ldots, z_{p,t} $ have coefficients of local-to-zero rate being $\frac{1}{2}$, see \cite{phillips2013predictive} and \cite{lee2021LASSO}.
Assumption~\ref{assump:limiting_Zt} ensures that the partial sums of $\{\bm{z}_t^*\}$ and $\{\bm{z}_t^*\,\bm{z}_t^{*'}\}$ converge at specific rates. Assumption~\ref{assump:limiting_Zt} is also used by \cite{saikkonen1993continuous, saikkonen1995problems,lee2021LASSO}, and was shown to hold for time series with a moderate degree of temporal dependence and heteroscedasticity of $\{\bm{\epsilon}_t\}$. See, e.g.,Theorem 18.2 of~\cite{billingsley2013convergence}, ~\cite{phillips1986multiple, phillips1991optimal}.
Assumption~\ref{assump:narrow_Theta} restricts the permissible parameter space $\Theta$ for the ML estimation, especially maintaining the positive definiteness of $H_{\mathcal{L}}(\bm{\beta})$ in analogy to Assumption~9 of~\cite{smith1985maximum} for settling the uniqueness of the estimator.
Assumption~\ref{assump:limiting_score_Z} assumes the limiting distribution of the likelihood gradient function at $\bm{\beta}^{o*}$, see Lemma A.1.1 of \cite{lee2016predictive}.
Assumption~\ref{assump:weakconverge_HL} assumes the existence of the limiting distribution of the likelihood Hessian matrix at $\bm{\beta}^{o*}$, see Lemma A.1.2 of \cite{lee2016predictive}. 

\begin{theorem} [MLE consistency]\ \\
\label{thm:MLE_consistency}
Under Assumptions~\ref{assump:prob_space}, \ref{assump:correct_model_speci}, \ref{assump:linear_processes_Zt}, \ref{assump:limiting_Zt}(1), \ref{assump:magnitude_of_coefs} and \ref{assump:narrow_Theta}, and for any $\epsilon >0$,
\begin{equation}
 \lim_{T\rightarrow\infty} P\left\{\norm{\widehat{\bm{\beta}}^{\text{mle}} - \bm{\beta}^{o*} }>\epsilon \right\} = 0,   
\end{equation} 
\end{theorem}
\begin{proof}
See Appendix~\ref{append:MLE}.
\end{proof}

\begin{theorem} [MLE asymptotics]\ \\
\label{thm:MLE_asym}
Under Assumptions~\ref{assump:prob_space} to \ref{assump:weakconverge_HL}, we have
\begin{equation}
\sqrt{T}\left(\widehat{\bm{\beta}}^{\text{mle}}-\bm{\beta}^{o*}\right) \overset{D}{\sim} \Omega_{H}^{-1} \; S_{\psi},
\qquad \text{as } T\rightarrow\infty.
\end{equation} 
\end{theorem}
\begin{proof}
See Appendix~\ref{append:MLE}.
\end{proof}

\section{Adaptive $L_1$-regularized maximum likelihood estimation}
\label{sec:LASSO}

Variable selection facilitates interpretation of a regression model and solves the trade-off issue between bias and efficiency so as to achieve predictive accuracy, see \cite{james2013introduction}.
Although variable selection performed via inferential tests based on the asymptotic normality of the MLE might seem a viable solution, it is not appropriate in our setting because of the following three issues: (i) the inability to control type I error for multiple predictor selection; (ii) severe size distortion for selecting unit-root predictors because of the non-normal limiting distribution; (iii) low power in selecting predictors for the shape parameter due to high standard errors of coefficients, see simulations in Section \ref{sec:simulation}. 

To circumvent these issues, we adopt $L_1$-regularized MLE for automatic variable selection \citep{tibshirani1996regression}. 
Due to the constraining nature of $L_1$-regularization, this estimator sets some coefficients exactly to zero so as to perform variable selection. \cite{zou2006adaptive} explore the advantages of using weighted $L_1$-regularization on model coefficients and proposed the adaptive LASSO. With proper adaptive weights, the adaptive LASSO exhibits the oracle property, which produces an asymptotic efficient estimator of variable selection consistency as if the true underlying model were given from the outset.
\cite{medeiros2016} prove the oracle property for the adaptive LASSO in high-dimensional time series with non-Gaussian and heteroscedastic errors as well as with highly correlated regressors. \cite{kock2016consistent} show that the adaptive LASSO is oracle efficient in stationary and non-stationary autoregressions. 
\cite{lee2021LASSO} prove the oracle property of the adaptive LASSO with stationary and local unit-root predictors, and propose a novel post-selection adaptive LASSO for selecting mixed-root predictors i.e. stationary, local unit root, and cointegrated predictors.

Drawing on this literature, we extend the adaptive LASSO to the MLE in \eqref{eq:ML_estimator} to estimate and select stationary and local unit-root predictors in \eqref{Model:k_sigma_models}-\eqref{Model:k_sigma_models_2ndEqu}.
A general form of adaptive $L_1$-regularized maximum likelihood estimator (ALMLE) can be drawn directly from \cite{zou2006adaptive} and is formulated as follows:
\begin{equation}
	\widehat{\bm{\beta}}^{\text{al}}  = \argmin{\bm{\beta}\in \Theta^* } - \mathcal{L}(\bm{\beta}; \{ y_t \}, \{ \bm{z}_t^* \}) + \lambda_{k, T}\sum_{i=1}^{p} w_{k,i}|\beta_{1,i}| + \lambda_{\sigma, T}\sum_{j=1}^{p+1} w_{\sigma,j}|\beta_{2,j}|,
	\label{eq:ML_adaL1_estimator}
\end{equation}
where $\mathcal{L}(\bm{\beta}; \{ y_t \}, \{ \bm{z}_t^* \})$ is the log-likelihood function specified in~\eqref{eq:likelihood_function}; $\lambda_{k, T}, \lambda_{\sigma, T}>0$ are tuning parameters; and $w_{k,i}, w_{\sigma,j}$ are adaptive weights for penalizing coefficients differently. 
We consider two tuning parameters instead of one to be less restrictive on tuning parameter selection, thereby stabilizing the variable selection for both the shape and scale models in \eqref{Model:k_sigma_models}-\eqref{Model:k_sigma_models_2ndEqu}. 
To set the tuning parameters, we start off with large enough values of $\lambda_{k, T}$ and $\lambda_{\sigma, T}$ such that no predictors are selected by $\widehat{\bm{\beta}}^{\text{al}}$, and denote these two values as $\lambda_{k, T, \max}$ and $\lambda_{\sigma, T, \max}$, respectively. 
We then search for the optimal tuning parameters using an information criterion (IC) over equally-spaced grids of $n_{\lambda_k}$ and $n_{\lambda_{\sigma}}$ nodes\footnote{We use $n_{\lambda_k}=50$ and $n_{\lambda_{\sigma}}=30$ across this paper unless stated otherwise. We also have tried $n_{\lambda_k}=100$ and $n_{\lambda_{\sigma}}=100$ to check the sufficiency of $n_{\lambda_k}=50$ and $n_{\lambda_{\sigma}}=30$, and found that differences in the results are small.} 
defined on the intervals  $[\lambda_{k, T, \max},10^{-6}]$ and $[\lambda_{\sigma, T, \max},10^{-6}]$. Formally, the grids for the shape and scale parameters are defined as $S_{\lambda_{k,T}}:=\{\exp(\log(\lambda_{k, T, \max}) - j\,\frac{\log(\lambda_{k, T, \max}) - \log(10^{-6})}{n_{\lambda_k} - 1} ), j=0,1,\ldots,(n_{\lambda_k}-1)\}$ and $S_{\lambda_{\sigma,T}}:=\{\exp(\log(\lambda_{\sigma, T, \max}) - j\,\frac{\log(\lambda_{\sigma, T, \max}) - \log(10^{-6})}{n_{\lambda_{\sigma}} - 1} ), j=0,1,\ldots,(n_{\lambda_{\sigma}}-1)\}$, respectively. We consider different information criteria, namely the Bayesian Information Criterion (BIC), the Hannan–Quinn information criterion (HQ) and the Akaike Information Criterion (AIC), and thus select the optimal tuning parameters $(\widehat{\lambda}_{k,T}, \widehat{\lambda}_{\sigma,T})$  according to the following rules,
\begin{align}
\text{AIC:}\qquad
& 
\resizebox{.9\hsize}{!}{$
	(\widehat{\lambda}_{k,T}, \widehat{\lambda}_{\sigma,T}) = \argmin{\lambda_{k,T}\in S_{\lambda_{k,T}}\,,\, \lambda_{\sigma,T}\in S_{\lambda_{\sigma,T}}}\, -2\log(\mathcal{L}(\widehat{\bm{\beta}}^{\text{al}}(\lambda_{k, T}, \lambda_{\sigma, T}); \{ y_t \}, \{ \bm{z}_t^* \})) + 2\left( \sum_{i=1,\ldots,p} \mathbbm{1}\{\widehat{\beta}^{\text{al}}_{1,i}\neq 0\} + \sum_{j=1,\ldots,(p+1)} \mathbbm{1}\{\widehat{\beta}^{\text{al}}_{2,j}\neq 0\} \right)
	$}
\\
\text{HQ:}\qquad
&
\resizebox{.9\hsize}{!}{$
	(\widehat{\lambda}_{k,T}, \widehat{\lambda}_{\sigma,T}) = \argmin{\lambda_{k,T}\in S_{\lambda_{k,T}}\,,\, \lambda_{\sigma,T}\in S_{\lambda_{\sigma,T}}}\, -2\log(\mathcal{L}(\widehat{\bm{\beta}}^{\text{al}}(\lambda_{k, T}, \lambda_{\sigma, T}); \{ y_t \}, \{ \bm{z}_t^* \})) + 2\log(\log(T))\left( \sum_{i=1,\ldots,p} \mathbbm{1}\{\widehat{\beta}^{\text{al}}_{1,i}\neq 0\} + \sum_{j=1,\ldots,(p+1)} \mathbbm{1}\{\widehat{\beta}^{\text{al}}_{2,j}\neq 0\} \right)
	$}
\\
\text{BIC:}\qquad
&
\resizebox{.9\hsize}{!}{$
	(\widehat{\lambda}_{k,T}, \widehat{\lambda}_{\sigma,T}) = \argmin{\lambda_{k,T}\in S_{\lambda_{k,T}}\,,\, \lambda_{\sigma,T}\in S_{\lambda_{\sigma,T}}}\, -2\log(\mathcal{L}(\widehat{\bm{\beta}}^{\text{al}}(\lambda_{k, T}, \lambda_{\sigma, T}); \{ y_t \}, \{ \bm{z}_t^* \})) + \log(T)\left( \sum_{i=1,\ldots,p} \mathbbm{1}\{\widehat{\beta}^{\text{al}}_{1,i}\neq 0\} + \sum_{j=1,\ldots,(p+1)} \mathbbm{1}\{\widehat{\beta}^{\text{al}}_{2,j}\neq 0\} \right)
	$}
\end{align}

The sequential strong rules of~\cite{tibshirani2012strong} is typically employed for computing LASSO-type problems. However, when $k_t$ presents persistent dynamics the sequential strong rules for $\widehat{\bm{\beta}}^{\text{al}}$ fails to screen among truly active and inactive predictors due to estimation bias when the tuning parameters are not small enough, and, as a byproduct, favors the boundary solution $k_t=0.5$. To reach variable selection consistency, it is necessary to enforce the optimizer to stay away from the boundary of the parameter space.
Theorem~\ref{thm:ALMLE_nece_cond} illustrates the restriction on the permissible coefficient space $\Theta$ in order to achieve the model selection consistency of $\widehat{\bm{\beta}}^{\text{al}}$, i.e.,
\begin{equation}
\lim_{T\rightarrow\infty} P\left\{ \mathcal{A}_T^{\text{al}} = \mathcal{A}\right\} = 1,
\end{equation} 
where $ \mathcal{A}_T^{\text{al}} :=  \mathcal{A}_{k,T}^{\text{al}} \cup \mathcal{A}_{\sigma,T}^{\text{al}}$ with $ \mathcal{A}_{k,T}^{\text{al}} := \left\{ (1,i): i\geq 1,  \widehat{\beta}_{1,i}^{\text{al}} \neq 0\right\}$ and $\mathcal{A}_{\sigma,T}^{\text{al}} := \left\{ (2,j): j\geq 1, \widehat{\beta}_{2,j}^{\text{al}} \neq 0\right\}$, and $ \mathcal{A} :=  \mathcal{A}_{k} \cup \mathcal{A}_{\sigma}$ with $ \mathcal{A}_{k} := \left\{ (1,i): i\geq 1,  \beta_{1,i}^{o*} \neq 0\right\}$ and $\mathcal{A}_{\sigma} := \left\{ (2,j): j\geq 1, \beta_{2,j}^{o*} \neq 0\right\}$. 

\begin{theorem}
\label{thm:ALMLE_nece_cond}
Under the assumptions in Theorem~\ref{thm:MLE_asym}, if there is no $\widehat{\bm{\beta}}^{\text{al}}(\lambda_{k, T}, \lambda_{\sigma, T})$ with $\lambda_{k, T}, \lambda_{\sigma, T}\in O(T^{\frac{1}{2}})$ such that
\begin{equation}
		\det\left( \frac{\partial^2 \mathcal{L}(\bm{\beta}) }{\partial [\bm{\beta}_{\mathcal{A}_k}', \bm{\beta}_{\mathcal{A}_{\sigma}}']' \partial [\bm{\beta}_{\mathcal{A}_k}', \bm{\beta}_{\mathcal{A}_{\sigma}}'] } \middle|_{\bm{\beta} = \widehat{\bm{\beta}}^{\text{al}}(\lambda_{k, T}, \lambda_{\sigma, T})} \right) 
		\neq 0,
 \label{eq:Thm3_DETcondition}
\end{equation}
then $\lim_{T\rightarrow\infty} P\left\{ \mathcal{A}_T^{\text{al}} = \mathcal{A}\right\} \neq 1$, where $\text{det}(\cdot)$ is the matrix determinant operator; $w_{k,i}$ and $w_{\sigma,j}$ are set using the MLE in section~\ref{sec:MLE_Inference} such that $\sqrt{T}(\frac{1}{w_{k,i}} - \beta_{1,i}^{*o})=O_p(1) $ and $\sqrt{T}(\frac{1}{w_{\sigma,j}} - \beta_{2,j}^{*o})=O_p(1) $, for $i=1,\ldots, p$, $j=1,\ldots, p+1$.
\end{theorem}
\begin{proof}
See Appendix~\ref{append:ALMLE}.
\end{proof}
Theorem~\ref{thm:ALMLE_nece_cond} shows that if not all the truly active predictors are able to enter the regression model with $\lambda_{k, T}, \lambda_{\sigma, T}\in O(T^{\frac{1}{2}})$, then truly inactive predictors start to be selected for compensating for the missing ones since $\lambda_{k, T}, \lambda_{\sigma, T}\in O(T^{\frac{1}{2}})$ and thereby fail $ \widehat{\bm{\beta}}^{\text{al}}$ in the variable selection. The necessary condition in Theorem~\ref{thm:ALMLE_nece_cond} tends to be broken when the underlying  $\{k_t(\bm{\beta}^{o*})\}$ involves local unit-root predictors.
To solve this issue we propose a two-step ALMLE and prove its oracle property. 

\subsection{Two-Step ALMLE}
\label{sec:tuning_paras}
From the previous discussion, we know that ALMLE can be improved if we ensure the estimation to stay away from $\{k_t(\cdot)=0.5\}$ for every $\lambda_{k, T}$. Therefore, we propose a two-step ALMLE, denoted as $\widehat{\bm{\beta}}^{\text{tal}}$, to avoid the local minimizer issue of $\widehat{\bm{\beta}}^{\text{al}}$ by selecting predictors for the shape at the first step and running the ALMLE in \eqref{eq:ML_adaL1_estimator} at the second step with the selected $\widehat{\lambda}_{k, T}$ in the first step. Specifically, the two-step ALMLE $\widehat{\bm{\beta}}^{\text{tal}}$ is obtained using the following procedure:
\begin{description}
\item{\textbf{Step 1}}: Select the optimal tuning parameter $\widehat{\lambda}_{k, T}\in S_{\lambda_{k,T}}$ using an IC as follows,
    \begin{align*}
    \text{AIC:}\qquad
    & 
    \widehat{\lambda}_{k,T} = \argmin{\lambda_{k,T}\in S_{\lambda_{k,T}}}\, -2\log(\mathcal{L}(\widehat{\bm{\beta}}^{\text{k,al}}(\lambda_{k, T}); \{ y_t \}, \{ \bm{z}_t^* \})) + 2\,\sum_{i=1,\ldots,p} \mathbbm{1}\{\widehat{\beta}^{\text{k,al}}_{1,i}\neq 0\}
    \\
    \text{HQ:}\qquad
    &
    \widehat{\lambda}_{k,T} = \argmin{\lambda_{k,T}\in S_{\lambda_{k,T}}}\, -2\log(\mathcal{L}(\widehat{\bm{\beta}}^{\text{k,al}}(\lambda_{k, T}); \{ y_t \}, \{ \bm{z}_t^* \})) + 2\,\log(\log(T)) \sum_{i=1,\ldots,p} \mathbbm{1}\{\widehat{\beta}^{\text{k,al}}_{1,i}\neq 0\} 
    \\
    \text{BIC:}\qquad
    &
    \widehat{\lambda}_{k,T} = \argmin{\lambda_{k,T}\in S_{\lambda_{k,T}}}\, -2\log(\mathcal{L}(\widehat{\bm{\beta}}^{\text{k,al}}(\lambda_{k, T}); \{ y_t \}, \{ \bm{z}_t^* \})) + \log(T) \sum_{i=1,\ldots,p} \mathbbm{1}\{\widehat{\beta}^{\text{k,al}}_{1,i}\neq 0\} \,,
    \end{align*}
where $\widehat{\bm{\beta}}^{\text{k,al}}(\lambda_{k, T}):= [\widehat{\beta}^{\text{k,al}}_{1,0},\ldots,\widehat{\beta}^{\text{k,al}}_{1,p},\widehat{\beta}^{\text{k,al}}_{2,0}, 0,\ldots,0]'$ restricts $\widehat{\beta}^{\text{k,al}}_{2,1},\dots,\widehat{\beta}^{\text{k,al}}_{2,p+1}$ to zero and define
\begin{equation}
\label{eq:betahat_S1_TALMLE}
    \left[\widehat{\beta}^{\text{k,al}}_{1,0},\ldots,\widehat{\beta}^{\text{k,al}}_{1,p}, \widehat{\beta}^{\text{k,al}}_{2,0}\right]
     =
    \argmin{\beta_{1,0},\ldots,\beta_{1,p},\beta_{2,0}} - \mathcal{L}(\bm{\beta}; \{ y_t \}, \{ \bm{z}_t^* \}) +  \lambda_{k, T}\sum_{i=1}^{p} \widetilde{w}_{k,i} |\beta_{1,i}|.
\end{equation}

\item{\textbf{Step 2}}: Select the optimal tuning parameter $\widehat{\lambda}_{\sigma, T} \in S_{\lambda_{\sigma,T}}$ using the IC and $\widehat{\lambda}_{k,T}$ from Step 1 as follows,
    \begin{align*}
    \text{AIC:}\qquad
    & 
    \resizebox{.9\hsize}{!}{$
    \widehat{\lambda}_{\sigma,T} = \argmin{ \lambda_{\sigma,T}\in S_{\lambda_{\sigma,T}}}\, -2\log(\mathcal{L}(\widehat{\bm{\beta}}^{\text{tal}}( \lambda_{\sigma, T}); \{ y_t \}, \{ \bm{z}_t^* \})) + 2\left( \sum_{i=1,\ldots,p} \mathbbm{1}\{\widehat{\beta}^{\text{tal}}_{1,i}\neq 0\} + \sum_{j=1,\ldots,(p+1)} \mathbbm{1}\{\widehat{\beta}^{\text{tal}}_{2,j}\neq 0\} \right)
    $}
    \\
    \text{HQ:}\qquad
    &
    \resizebox{.9\hsize}{!}{$
    \widehat{\lambda}_{\sigma,T} = \argmin{ \lambda_{\sigma,T}\in S_{\lambda_{\sigma,T}}}\, -2\log(\mathcal{L}(\widehat{\bm{\beta}}^{\text{tal}}(\lambda_{\sigma, T}); \{ y_t \}, \{ \bm{z}_t^* \})) + 2\log(\log(T))\left( \sum_{i=1,\ldots,p} \mathbbm{1}\{\widehat{\beta}^{\text{tal}}_{1,i}\neq 0\} + \sum_{j=1,\ldots,(p+1)} \mathbbm{1}\{\widehat{\beta}^{\text{tal}}_{2,j}\neq 0\} \right)
    $}
    \\
    \text{BIC:}\qquad
    &
    \resizebox{.9\hsize}{!}{$
    \widehat{\lambda}_{\sigma,T} = \argmin{ \lambda_{\sigma,T}\in S_{\lambda_{\sigma,T}}}\, -2\log(\mathcal{L}(\widehat{\bm{\beta}}^{\text{al}}(\lambda_{\sigma, T}); \{ y_t \}, \{ \bm{z}_t^* \})) + \log(T)\left( \sum_{i=1,\ldots,p} \mathbbm{1}\{\widehat{\beta}^{\text{tal}}_{1,i}\neq 0\} + \sum_{j=1,\ldots,(p+1)} \mathbbm{1}\{\widehat{\beta}^{\text{tal}}_{2,j}\neq 0\} \right)\,,
    $}
    \end{align*}
where $\widehat{\bm{\beta}}^{\text{tal}}(\lambda_{\sigma, T}):= [\widehat{\boldsymbol{\beta}}^{\text{tal}'}_{1\cdot},\widehat{\boldsymbol{\beta}}^{\text{tal}'}_{2\cdot}]'=[\widehat{\beta}^{\text{tal}}_{1,0},\ldots,\widehat{\beta}^{\text{tal}}_{1,p},\widehat{\beta}^{\text{tal}}_{2,0},\ldots,\widehat{\beta}^{\text{tal}}_{2,(p+1)}]'$, with $\widehat{\beta}^{\text{tal}}_{1,i}=\widehat{\beta}^{\text{al}}_{1,i}=0, \, \forall (1,i) \not\in \mathcal{A}_T^{k,al}$ and
\begin{equation}
\label{eq:betahat_TALMLE}
\resizebox{.95\hsize}{!}{$
\begin{aligned}
    \left[\left[\widehat{\beta}^{\text{tal}}_{1,i}\right]_{(1,i) \in \{(1,0)\}\cup \mathcal{A}_T^{k,al}}, \widehat{\boldsymbol{\beta}}^{\text{tal}'}_{2\cdot} \right] 
     & = 
    \argmin{\left\{\beta_{1,i}|(1,i) \in \{(1,0)\}\cup \mathcal{A}_T^{k,al}\right\}, \boldsymbol{\beta}_{2\cdot}} - \mathcal{L}(\bm{\beta}; \{ y_t \}, \{ \bm{z}_t^* \})  + \widehat{\lambda}_{k, T} \sum_{i=1}^{p} \widetilde{w}_{k,i}|\beta_{1,i}| + \lambda_{\sigma, T}\sum_{j=1}^{p+1} \widetilde{w}_{\sigma,j} |\beta_{2,j}|,
\end{aligned}
$}
\end{equation}
where $\mathcal{A}_T^{k,al} := \left\{(1,i):\,i\geq 1, \widehat{\beta}_{1,i}^{k,al} \neq 0\right\}$.
\end{description}
 
The final two-step ALMLE $\widehat{\bm{\beta}}^{tal}$ is obtained using the optimal tuning parameters $\widehat{\lambda}_{k,T}$ and $\widehat{\lambda}_{\sigma,T}$.

We use two MLEs to set up $\widetilde{w}_{k,i}$ and $\widetilde{w}_{\sigma,j}$ as the two-step ALMLE involves two different likelihood functions in each step. Specifically, we set
\begin{equation}
\label{eq:adaptive_weights_twostep}
\left\{
    \begin{aligned}
    \widetilde{w}_{k,i}
    &= 
    \frac{1}{\widehat{\beta}^{\text{k,mle}}_{1,i}}\,\frac{1}{\widehat{\beta}^{\text{mle}}_{1,i}},\; i=1,\ldots,p,
    \\
    \widetilde{w}_{\sigma,j}
    &=
    \frac{1}{\widehat{\beta}^{\text{mle}}_{2,j}} ,\; j=1,\ldots,p+1.
    \end{aligned}
    \right.
\end{equation}
where $\widehat{\bm{\beta}}^{\text{mle}} := [\widehat{\beta}^{\text{mle}}_{1,0}, \ldots,\widehat{\beta}^{\text{mle}}_{1,p},\widehat{\beta}^{\text{mle}}_{2,0}, \ldots,\widehat{\beta}^{\text{mle}}_{2,p+1}]$ is the full-model MLE~\eqref{eq:ML_estimator} and  $ \widehat{\bm{\beta}}^{\text{k,mle}} := [\widehat{\beta}^{\text{k,mle}}_{1,0}, \ldots,\widehat{\beta}^{\text{k,mle}}_{1,p},\widehat{\beta}^{\text{k,mle}}_{2,0},$ $ 0,\ldots,0]$ is the partial-model MLE defined below
\begin{equation}
    \widehat{\bm{\beta}}^{\text{k,mle}} 
            = \argmin{\{\bm{\beta}\in \Theta | \beta_{2,j} = 0, j=1,\ldots,p+1 \}} 
    - \mathcal{L}(\bm{\beta}; \{ y_t \}, \{ \bm{z}_t^* \}).
    \label{eq:betahat_KMLE}
\end{equation}
In this way, we choose $\widetilde{w}_{k,i}$ and $\widetilde{w}_{\sigma,j}$ such that truly active predictors are ensured to be selected efficiently with $S_{\lambda_{k,T}}$ and $ S_{\lambda_{\sigma,T}}$ before the truly inactive ones in both Step 1 and Step 2. Therefore, we achieve the oracle property of $\widehat{\bm{\beta}}^{tal}$ as shown in Theorem~\ref{thm:ALMLE_tal_oracle}.

\begin{assumption}{L1}
\label{assump:NEW_tuning}
There exist $\lambda_{k,T} = O(T^{\frac{1}{2} - \gamma_1})$ and  $\lambda_{\sigma,T} =  O(T^{\frac{1}{2} - \gamma_2})$ with $0<\gamma_1<\frac{1}{2}$ and $0<\gamma_2<\frac{1}{2}$.
\end{assumption}
\begin{assumption}{L2}
\label{assump:kmle_tal}
We assume that there exists $\bm{\beta}^{\text{k,o}}:=[\beta^{\text{k,o}}_{1,0},\beta^{\text{k,o}}_{1,1},\ldots,\beta^{\text{k,o}}_{1,p},\beta^{\text{k,o}}_{2,0},0,\ldots,0 ]'\in \{\bm{\beta}\in \mathbb{R}^{2p+3} | \beta_{2,j} = 0, j=1,\ldots,p+1 \}$ such that for any $\epsilon>0$
\begin{equation}
   \lim_{T\rightarrow\infty} P\left\{ \left| \widehat{\beta}^{\text{k,mle}}_{1,i} - \beta^{\text{k,o}}_{1,i}  \right| > \epsilon \right\} = 0,\quad i=1,\ldots,p;
\end{equation}
and $\beta^{\text{k,o}}_{1,i}\neq 0$ for any $(1,i)\in\mathcal{A}_k$.
\end{assumption}

\begin{theorem}[Oracle Property of $\widehat{\bm{\beta}}^{tal}$]
\label{thm:ALMLE_tal_oracle}~\\
Under Assumptions~\ref{assump:NEW_tuning}, ~\ref{assump:kmle_tal} and the assumptions in Theorem~\ref{thm:MLE_asym}, we have that
\\
(a) Model selection consistency:
\\
\begin{equation}
    \lim_{T\rightarrow\infty}P\left\{ \mathcal{A}_{T}^{tal} = \mathcal{A}\right\} = 1,
\end{equation}
where $\mathcal{A}_{T}^{tal}:= \mathcal{A}_{k,T}^{tal}\cup\mathcal{A}_{\sigma,T}^{tal}$ with $\mathcal{A}_{k,T}^{tal}:=\left\{ (1,i): \widehat{\beta}_{1,i}^{tal} \neq 0, i=1,\ldots, p.\right\} $ and $\mathcal{A}_{\sigma,T}^{tal}:=\{ (2,j): \widehat{\beta}_{2,j}^{tal} \neq 0,$ $ j=1,\ldots,p+1.\} $.
\\
(a) Limiting distribution of $\widehat{\bm{\beta}}^{tal}$:
\\
\begin{equation}
    \begin{aligned}
    &
    \sqrt{T}\left(\widehat{\bm{\beta}}^{tal}_{\mathcal{A}} - \bm{\beta}^{o*}_{\mathcal{A}}\right) \overset{D}{\sim} \Omega^{-1}_{H_\mathcal{A}} \; S_{\psi_{\mathcal{A}}},
    \\
    &
    \sqrt{T}\left(\widehat{\bm{\beta}}^{tal}_{\mathcal{A}^c} - \bm{\beta}^{o*}_{\mathcal{A}^c}\right)\rightarrow 0,
    \end{aligned}
\end{equation}
as $T\rightarrow\infty$, where $S_{\psi_\mathcal{A}}$ and  $\Omega_{H_\mathcal{A}}$ are defined in Assumption~\ref{assump:limiting_score_Z} and \ref{assump:weakconverge_HL} under the model specification with only the truly active predictors involved and ordered according to $\mathcal{A}$.
\end{theorem}
\begin{proof}
See Appendix~\ref{append:ALMLE}. 
\end{proof}

The superiority of the proposed two-step ALMLE to the ALMLE~\eqref{eq:ML_adaL1_estimator} is not just in the oracle property when local unit-root predictors are included in the regression model but also in the computing cost. The ALMLE~\eqref{eq:ML_adaL1_estimator} is computed over a two-dimensional tuning parameter grid in order to select an optimal pair of $(\lambda_{k,T}, \lambda_{\sigma,T}) \in S_{\lambda_{k,T}}\times S_{\lambda_{\sigma,T}}$, while the two-step ALMLE is computed over two separate one-dimensional tuning parameter grids in order to select the optimal $\lambda_{k,T}\in S_{\lambda_{k,T}}$ first and $\lambda_{\sigma,T}\in S_{\lambda_{\sigma,T}}$ after. 

\section{Simulation study}
\label{sec:simulation}
We assess the finite sample properties of $ \widehat{\bm{\beta}}^{mle}$ and $\widehat{\bm{\beta}}^{tal}$ from the perspectives of their biases, mean square errors (MSEs) and model selection using four data generating processes (DGPs). These four DGPs are designed to reflect the characteristics of the high-frequency financial data used in Section \ref{sec:empirical_study}. First, DGPs are heteroscedastic and the conditional exceedance rates can change over time. Second, DGPs involve predictors which are functions of lagged loss rates characterizing the serial dependence structure in $\{(k_t,\sigma_t)\}$. Third, we consider either stationary or local unit-root predictors or both.

We simulate $\{ l_t\}$ from the following conditional distribution,
 \begin{equation}
    l_t \; =
    \left\{
    \begin{aligned}
    &  F^{-1}_{t\left(\frac{1}{k_t}\right)}(\tau_t) , \qquad \text{if}\quad \tau_t \leq F_{t\left(\frac{1}{k_t}\right)}(u)
				\\
    & u + F_{\text{GPD}(k_t, \sigma_t)}^{-1}\left(\frac{\tau_t - F_{t(\frac{1}{k_t})}(u)}{1-\tau_t}\right), \qquad \text{if}\quad \tau_t > F_{t\left(\frac{1}{k_t}\right)}(u)				,
    \end{aligned}
    \right.
    \label{DGP:threshold_tdist|GPD}
    \end{equation}
where $\{\tau_t\}$ is i.i.d. standard uniform distributed, $F_{t\left(\nu\right)}(\cdot)$  and $F^{-1}_{t\left(\nu\right)}(\cdot)$ denote the distribution and quantile functions of a Student's t distribution with $\nu$ degrees of freedom. The processes of $\{k_t\}$ and $\{ \sigma_t\}$ are specified according to the following specifications:
    \begin{equation}
    \left\{
        \begin{aligned}
        & \log\left(\frac{ k_t }{ 0.5 - k_t } \right)  =	\beta_{10} + \beta_{11}\, \log(|l_{t-1}| + 1 - r_m) + \sum_{j=1}^{14} \beta_{1,1+j}\,z_{j,t-1}\,,
        \\
        & \log(\sigma_t) 	=	\beta_{20} + \beta_{21}\,\log(\sigma_{t-1}) + \beta_{22}\,  \log(|l_{t-1}| + 1 - r_m) + \sum_{j=1}^{14}\beta_{2,2+j}\,z_{j,t-1}\,,
        \\
        & z_{i,t} 
        = \phi_i\,z_{i,t-1} + \epsilon_{i,t-1}, \qquad i = 1,2,\ldots, 14,
        \\
        & \bm{\phi} := [\phi_1, \ldots, \phi_{14}],
        \\
        & \{\bm{\epsilon}_t 
        := [\epsilon_{1,t},\ldots, \epsilon_{14,t}]'\} \overset{i.i.d.}{\sim} \mathcal{N}(0,I_{14\times 14})\,,
        \\
        & \bm{\beta^o_{1\cdot}} 
        := [\beta_{1,0}^o, \beta_{1,1}^o,\ldots, \beta_{1,15}^o]'
         \,,
        \\
        & \bm{\beta^o_{2\cdot}} 
        := [\beta_{2,0}^o, \beta_{2,1}^o,\ldots, \beta_{2,16}^o]'
        \,,
        \\
        & \bm{\beta^o} 
        := [\bm{\beta}^{o'}_{1\cdot}, \bm{\beta}^{o'}_{2\cdot}  ]'\,.
        \end{aligned}
        \right.
     \label{eq:DGP1-4}   
    \end{equation}
We set $u = F^{-1}_{t(3)}(0.8)$ and $ r_m = 0.05$, but use different $\bm{\phi}$ and $\bm{\beta^o}$ to obtain different degrees of serial dependence.

\begin{description}
    \item[\textbf{DGP 1.}] There are five truly active stationary predictors for both $\{k_t\}$ and $\{\sigma_t\}$, namely $\log(|l_{t-1}| + 1 - r_m)$, $z_{1,t-1},\ldots, z_{4,t-1}$. Among truly inactive predictors $z_{5,t-1},\ldots, z_{14,t-1}$, two of them are local unit-root, i.e. $z_{13,t-1}$ and $z_{14,t-1}$, and the others are stationary.

    \begin{equation}
    \left\{
        \begin{aligned}
        & \bm{\phi} = [0,0,0,0,0\ldots,0,1,1],
        \\
        & \bm{\beta^o_{1\cdot}} 
         = [-1, 0.3, -0.4, 0.2, 0.6, 0.6, 0, \ldots,0 ]'\,,
        \\
        & \bm{\beta^o_{2\cdot}} 
         = [-1, 0, 0.7, 0.4, 0.3, 0.5, 0.6, 0, \ldots,0 ]'\,,
        \end{aligned}
        \right.
     \label{eq:DGP3}   
    \end{equation}

\item[\textbf{DGP 2.}] As DGP 1 but with the difference that $\beta_{2,1}^o$ is changed to nonzero, and hence $\log(\sigma_{t-1})$ is now truly active. We set $\beta_{2,1}^o=0.7$ and keep the true values of the other coefficients unchanged.

    \begin{equation}
    \left\{
        \begin{aligned}
        & \bm{\phi} = [0,0,0,0,0\ldots,0,1,1],
        \\
        & \bm{\beta^o_{1\cdot}} 
         = [-1, 0.3, -0.4, 0.2, 0.6, 0.6, 0, \ldots,0 ]'\,,
        \\
        & \bm{\beta^o_{2\cdot}} 
         = [-1, 0.7, 0.7, 0.4, 0.3, 0.5, 0.6, 0, \ldots,0 ]'\,,
        \end{aligned}
        \right.
     \label{eq:DGP4}   
    \end{equation}

\item[\textbf{DGP 3.}] As DGP 1 but with the difference that $\phi_{4} = 1$ and $(\beta_{1,5}^o,\beta_{2,6}^o) = (\frac{0.6}{\sqrt{T}}, \frac{0.6}{\sqrt{T}})$. 

    \begin{equation}
    \left\{
        \begin{aligned}
        & \bm{\phi} = [0,0,0,1,0\ldots,0,1,1],
        \\
        & \bm{\beta^o_{1\cdot}} 
         = [-1, 0.3, -0.4, 0.2, 0.6, \frac{0.6}{\sqrt{T}}, 0, \ldots,0 ]'\,,
        \\
        & \bm{\beta^o_{2\cdot}} 
         = [-1, 0, 0.7, 0.4, 0.3, 0.5, \frac{0.6}{\sqrt{T}}, 0, \ldots,0 ]'\,,
        \end{aligned}
        \right.
     \label{eq:DGP5}   
    \end{equation}

\item[\textbf{DGP 4.}] As DGP 3 but with the difference that $\log(\sigma_{t-1})$ is truly active. We set $\beta_{2,1}^o=0.7$ and and keep the true values of the other coefficients unchanged. 

    \begin{equation}
    \left\{
        \begin{aligned}
        & \bm{\phi} = [0,0,0,1,0\ldots,0,1,1],
        \\
        & \bm{\beta^o_{1\cdot}} 
         = [-1, 0.3, -0.4, 0.2, 0.6, \frac{0.6}{\sqrt{T}}, 0, \ldots,0 ]'\,,
        \\
        & \bm{\beta^o_{2\cdot}} 
         = [-1, 0.7, 0.7, 0.4, 0.3, 0.5, \frac{0.6}{\sqrt{T}}, 0, \ldots,0 ]'\,,
        \end{aligned}
        \right.
     \label{eq:DGP6}   
    \end{equation}

\end{description}

In each simulation, we obtain a sample $\{ l_t\}_{t=1}^T$ of $T$ observations, and extract the excess time series $\{ y_t=\max(l_t - u, 0)\}$ using the true threshold $u$. We standardize the predictors using their empirical standard deviations. We then fit the full model specification \eqref{eq:DGP1-4} to $\{y_t\}$ using  standardized predictors, estimating the model parameter by $\widehat{\bm{\beta}}^{mle}$ and $ \widehat{\bm{\beta}}^{tal}$. Bias and mean squared error (MSE) are then computed as 

\begin{equation}
    \text{Bias} 
    = \frac{1}{\#\bm{\beta}^o}\,\left(
    \sum_{i=1,2;j=0} \abs{\widehat{\beta}_{i,j} - \beta_{i,j}^o }
    +
    \sum_{i=1,2;j\geq 1} s_{i,j}\abs{\widehat{\beta}_{i,j}\, \widehat{s}_{i,j} - \beta_{i,j}^o }
     \right)
    \label{eq:Bias}
\end{equation}
\begin{equation}
    \text{MSE} 
     = \frac{1}{\#\bm{\beta}^o}\,\left(
     \sum_{i=1,2;j=0} \left(\widehat{\beta}_{i,j}- \beta_{i,j}^o \right)^2 
     + 
     \sum_{i=1,2;j\geq 1} s_{i,j}^2\left(\widehat{\beta}_{i,j}\, \widehat{s}_{i,j} - \beta_{i,j}^o \right)^2 
      \right)
    \label{eq:MSE}
\end{equation}
where $\#\bm{\beta}^o$ denotes the number of parameters in $\bm{\beta}^o$, $\widehat{s}_{i,j}$ denotes the empirical standard deviation of the $(i,j)$-th predictor, and $s_{i,j}=1$ for $I(0)$ predictors and $s_{i,j}=\sqrt{T}$ for $I(1)$ predictors. 

{Table~\ref{tab:DGP3-6_estimates_bias_MSE} presents the average absolute bias and average MSE of the coefficient estimates obtained over 100 replications. These results show that $\widehat{\bm{\beta}}^{mle}$ and $\widehat{\bm{\beta}}^{tal}$ have decreasing biases and MSEs when $T$ increases, coherently with the theoretical results presented in Sections \ref{subsec:limiting_mle} and \ref{sec:LASSO}. 
Moreover $\widehat{\bm{\beta}}^{tal}$ under BIC always has the lowest bias and MSE across the DGPs, supporting the use of $\widehat{\bm{\beta}}^{tal}$ with BIC in the empirical section. Boxplots for the bias in Figure \ref{fig:boxplots_bias_DGPs} support these conclusions.

Table~\ref{tab:DGP3-6_variable_selection} presents the variable selection results for both $\widehat{\bm{\beta}}^{tal}$ and $\widehat{\bm{\beta}}^{mle}$. Note that for the latter, we perform variable selection based on the significance of the t-statistics associated to the candidate predictors. To measure the ability to select the correct predictors, we assess the average selection rates of truly active and inactive predictors for both the shape and scale parameters. Moreover, we compute the correct classification rate (CCR) of each estimator, i.e. the proportion of selected truly active and unselected truly inactive predictors on the total candidate predictors. 
Results in Table \ref{tab:DGP3-6_variable_selection} show that variable selection improves as $T$ increases for each estimator. For $\widehat{\bm{\beta}}^{mle}$ the average selection rates of truly inactive stationary predictors approach the significance level $\alpha=0.05$, whereas the average selection rates of truly inactive local unit-root predictors are much higher than $\alpha=0.05$, for both $k$ and $\sigma$, and regardless of the DGP. These results are coherent with the asymptotic results derived in Section \ref{subsec:limiting_mle}, and echo the size distortion concerns of using t-tests to select non-stationary predictors discussed in Section \ref{sec:LASSO}. Remarkably, the average selection rates of truly active predictors for $\widehat{\bm{\beta}}^{mle}$ are much lower than those for $\widehat{\bm{\beta}}^{tal}$. Moreover, we see that the power of t-tests performed with $\widehat{\bm{\beta}}^{mle}_{1\cdot}$ is lower than the one for $\widehat{\bm{\beta}}^{mle}_{2\cdot}$ due to the uncertainty in the estimation of $\widehat{\bm{\beta}}^{mle}_{1\cdot}$. Finally, Table~\ref{tab:DGP3-6_variable_selection} shows that $\widehat{\bm{\beta}}^{tal}$ with BIC always has the highest CCR and produces the most accurate selection regardless the DGP, supporting the use of $\widehat{\bm{\beta}}^{tal}$ with BIC for the empirical application. 


\begin{table}[htbp]
  \centering
  \caption{Average absolute bias \eqref{eq:Bias} and MSE \eqref{eq:MSE} over 100 replications obtained with $\widehat{\bm{\beta}}^{mle}$ and $\widehat{\bm{\beta}}^{tal}$ using the optimal tuning parameters selected by AIC, HQ and BIC criteria.}
    \begin{tabular}{clp{1.3cm}p{1.3cm}p{1.3cm}p{1.3cm}p{1.3cm}p{1.3cm}}
    \toprule
    \toprule
          &       & \multicolumn{3}{c}{Bias} & \multicolumn{3}{c}{MSE} \\
\midrule    
    DGPs  & \diagbox[innerrightsep=25pt]{Estimators}{T} & 25,000 & 50,000 & 100,000 & 25,000 & 50,000 & 100,000 \\
    \midrule
    \multirow{4}[2]{*}{DGP 1} & $\widehat{\bm{\beta}}^{mle}$ & 0.031 & 0.015 & 0.008 & 0.105 & 0.031 & 0.016 \\
          & $\widehat{\bm{\beta}}^{tal}$ + AIC & 0.014 & 0.007 & 0.004 & 0.036 & 0.014 & 0.007 \\
          & $\widehat{\bm{\beta}}^{tal}$ + HQ & 0.010 & 0.006 & 0.004 & 0.030 & 0.012 & 0.006 \\
          & $\widehat{\bm{\beta}}^{tal}$ + BIC & 0.010 & 0.005 & 0.004 & 0.026 & 0.011 & 0.006 \\
    \midrule
    \multirow{4}[2]{*}{DGP 2} & $\widehat{\bm{\beta}}^{mle}$ & 0.027 & 0.017 & 0.010 & 0.091 & 0.031 & 0.014 \\									
          & $\widehat{\bm{\beta}}^{tal}$ + AIC & 0.015 & 0.011 & 0.006 & 0.047 & 0.020 & 0.009 \\									
          & $\widehat{\bm{\beta}}^{tal}$ + HQ & 0.012 & 0.009 & 0.004 & 0.034 & 0.012 & 0.007 \\									
          & $\widehat{\bm{\beta}}^{tal}$ + BIC & 0.009 & 0.007 & 0.003 & 0.030 & 0.011 & 0.005 \\									

    \midrule
    \multirow{4}[2]{*}{DGP 3} & $\widehat{\bm{\beta}}^{mle}$ & 0.047 & 0.023 & 0.009 & 0.197 & 0.051 & 0.023 \\
          & $\widehat{\bm{\beta}}^{tal}$ + AIC & 0.021 & 0.015 & 0.005 & 0.082 & 0.024 & 0.013 \\
          & $\widehat{\bm{\beta}}^{tal}$ + HQ & 0.018 & 0.013 & 0.003 & 0.067 & 0.021 & 0.011 \\
          & $\widehat{\bm{\beta}}^{tal}$ + BIC & 0.017 & 0.012 & 0.003 & 0.060 & 0.020 & 0.010 \\
    \midrule
    \multirow{4}[2]{*}{DGP 4} & $\widehat{\bm{\beta}}^{mle}$ & 0.029 & 0.025 & 0.019 & 0.128 & 0.165 & 0.263 \\
          & $\widehat{\bm{\beta}}^{tal}$ + AIC & 0.018 & 0.017 & 0.006 & 0.071 & 0.099 & 0.013 \\
          & $\widehat{\bm{\beta}}^{tal}$ + HQ & 0.015 & 0.014 & 0.006 & 0.065 & 0.093 & 0.012 \\
          & $\widehat{\bm{\beta}}^{tal}$ + BIC & 0.015 & 0.014 & 0.005 & 0.057 & 0.091 & 0.011 \\
    \bottomrule
    \bottomrule
    \multicolumn{1}{l}{} &       &       &       &       &       &       &  \\
    \end{tabular}%
  \label{tab:DGP3-6_estimates_bias_MSE}%
\end{table}%

\begin{figure}
\begin{subfigure}{.9\textwidth}
  \centering
  \includegraphics[width=.9\linewidth, height=5.8cm]{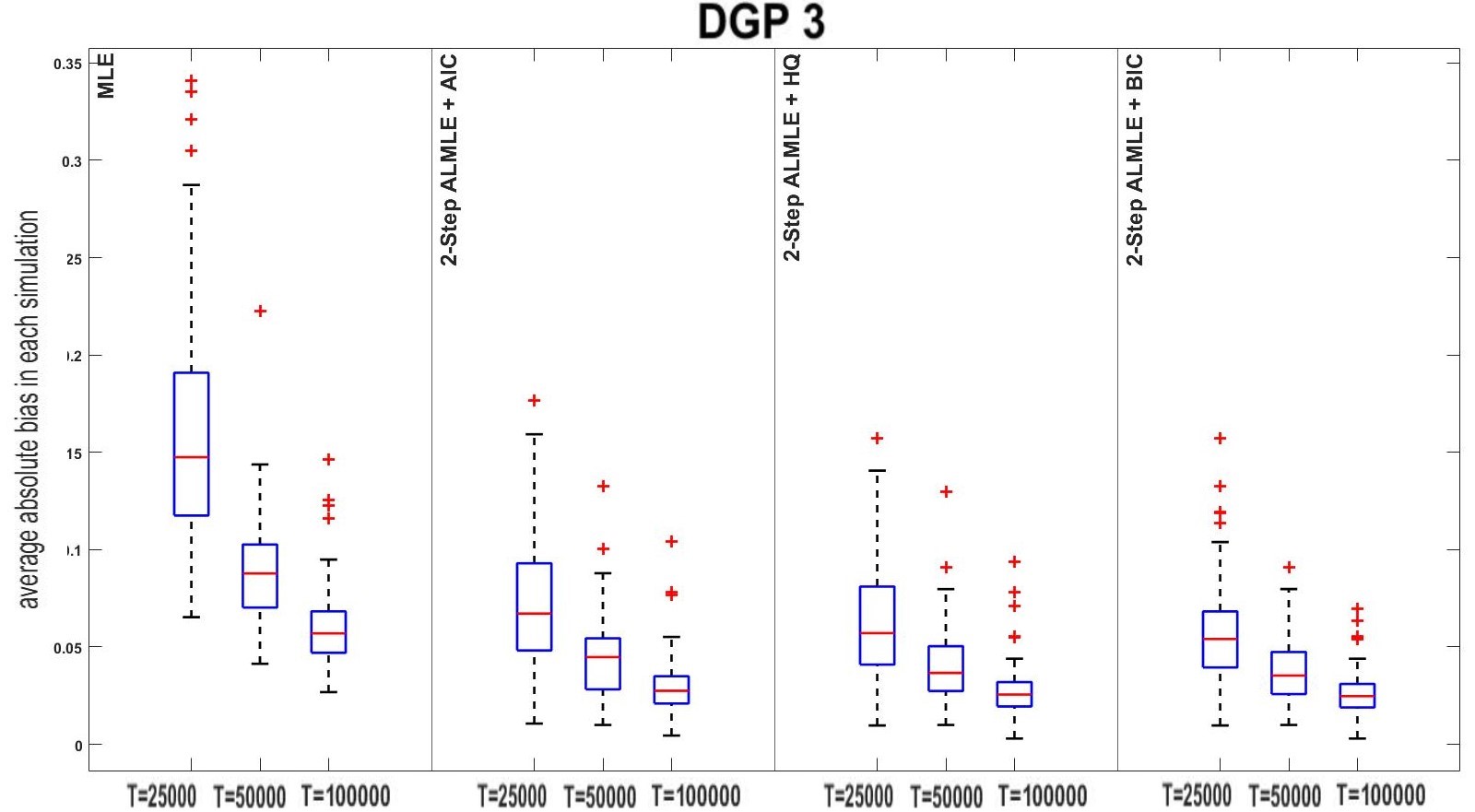}  
  \caption{}
  \label{fig:boxplot_DGP3}
\end{subfigure}
\newline
\begin{subfigure}{.9\textwidth}
  \centering
  \includegraphics[width=.9\linewidth, height=5.8cm]{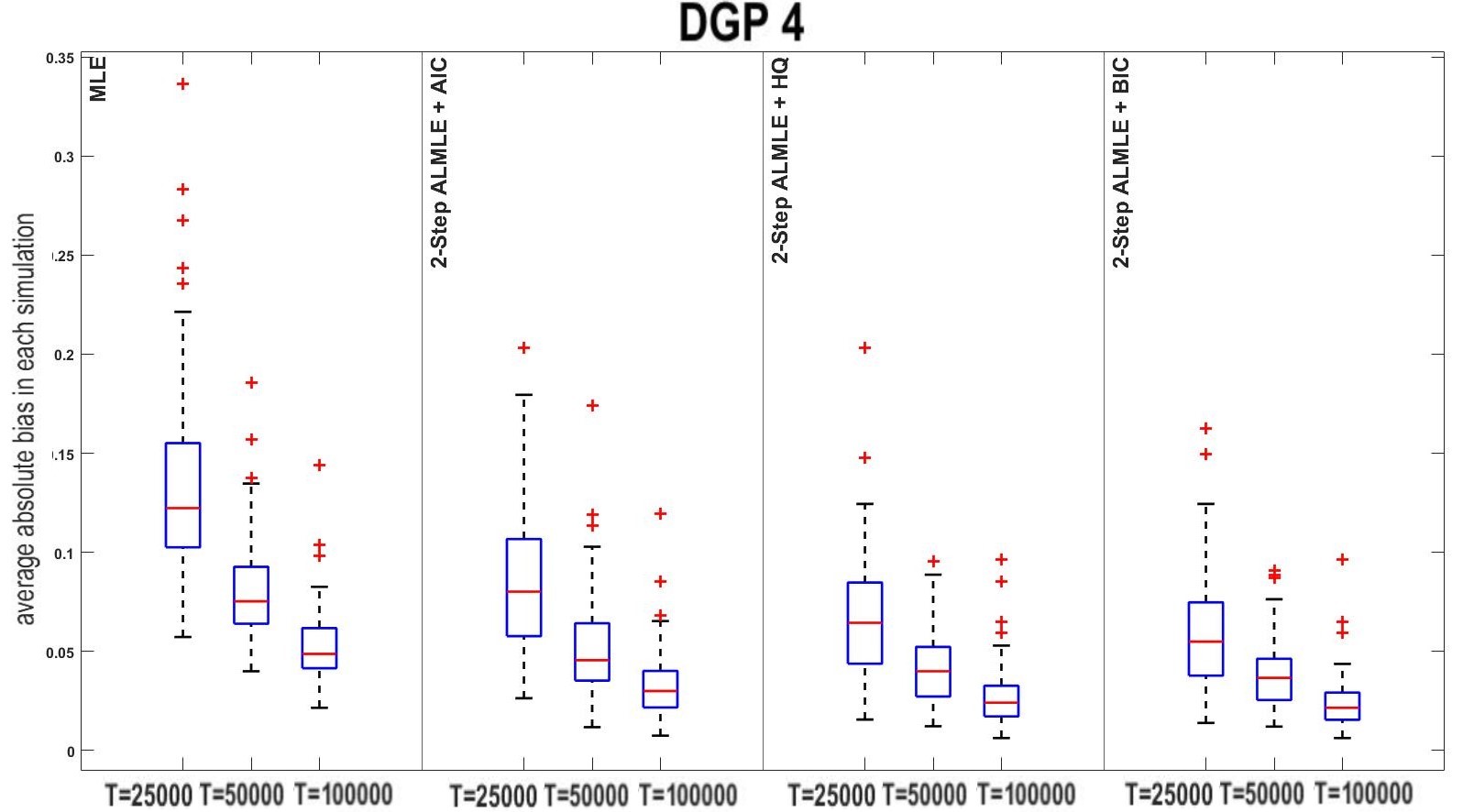}
  \caption{}
  \label{fig:boxplot_DGP4}
\end{subfigure}
\newline
\begin{subfigure}{.9\textwidth}
  \centering
  \includegraphics[width=.9\linewidth, height=5.8cm]{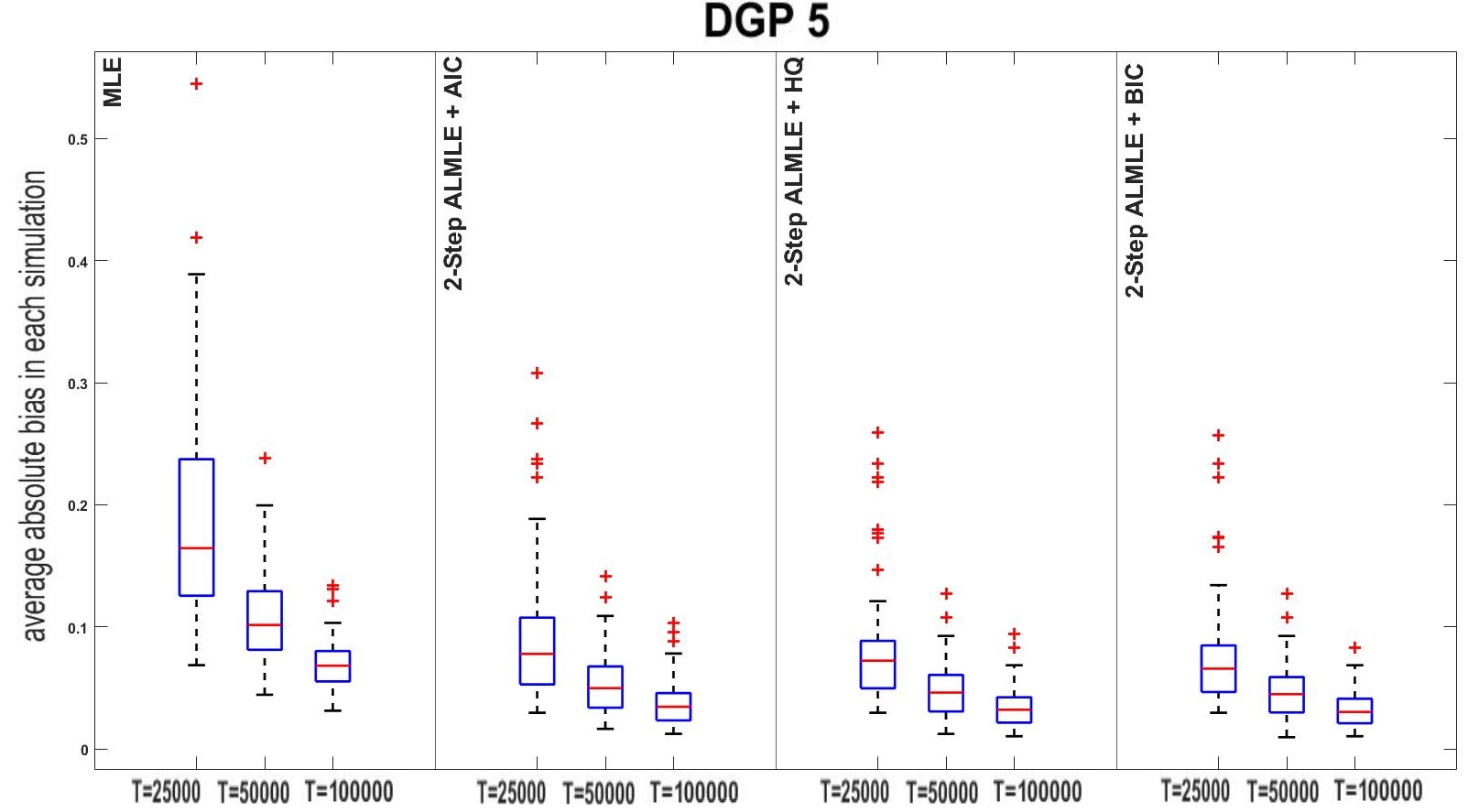}
  \caption{}
  \label{fig:boxplot_DGP5}
\end{subfigure}
\newline
\begin{subfigure}{.9\textwidth}
  \centering
  \includegraphics[width=.9\linewidth, height=5.8cm]{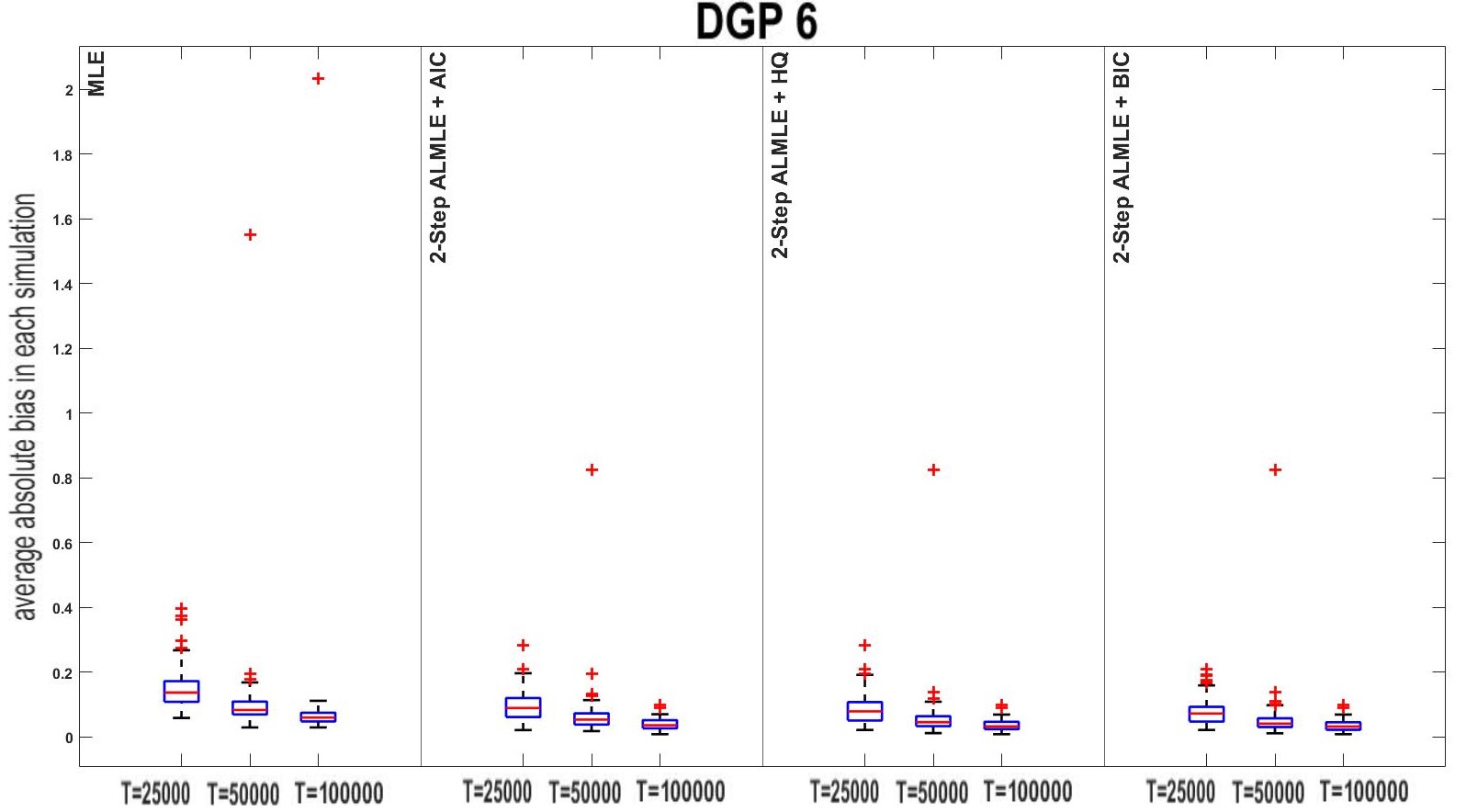}  
  \caption{}
  \label{fig:boxplot_DGP6}
\end{subfigure}
\caption{Boxplots of bias \eqref{eq:Bias} obtained from 100 replications  with $\widehat{\bm{\beta}}^{mle}$ and $\widehat{\bm{\beta}}^{tal}$ using the optimal tuning parameters selected by AIC, HQ and BIC criteria.}
\label{fig:boxplots_bias_DGPs}
\end{figure}

\begin{sidewaystable}[htbp]
  \caption{Average selection rate across 100 replication for truly active (t.p.) and truly inactive (f.p.) stationary (I(0)) and local unit-root (I(1)) predictors in the shape ($k$) and scale ($\sigma$) parameters, correct classification rate (CCR), and selection rates of $\log(\sigma_{t-1})$.}
  \label{tab:DGP3-6_variable_selection}
  \resizebox{0.9\textwidth}{!}{%
    \begin{tabular}{cc|l|l|cccccccccc}
    \toprule
    \toprule
          DGPs& T     & estimators & selection criteria & \multicolumn{1}{l}{  t.p.(k) of $I(0)s$\;\;\;\;} & \multicolumn{1}{l}{f.p.(k) of $I(0)s$\;\;\;\;} & \multicolumn{1}{l}{t.p.(k) of $I(1)s$\;\;\;\;} & \multicolumn{1}{l}{f.p.(k) of $I(1)s$\;\;\;\;} & \multicolumn{1}{l}{t.p.($\sigma$) of $I(0)s$\;\;\;\;} & \multicolumn{1}{l}{f.p.($\sigma$) of $I(0)s$\;\;\;\;} & \multicolumn{1}{l}{t.p.($\sigma$) of $I(1)s$\;\;\;\;} & \multicolumn{1}{l}{f.p.($\sigma$) of $I(1)s$\;\;\;\;} & \multicolumn{1}{l}{ CCR\;\;\;\;} & \multicolumn{1}{l}{ $\log(\sigma_{t-1})$} \\
    \midrule
    \multirow{12}[6]{*}{\textbf{DGP 3} } & \multirow{4}[2]{*}{25,000} & $\widehat{\bm{\beta}}^{mle}$ & t test ($\alpha=0.05$) & 0.546 & 0.063 & $-$ & 0.110 & 1.000 & 0.110 & $-$ & 0.215 & 0.858 & 0.060 \\
          &       & \multirow{3}[1]{*}{$\widehat{\bm{\beta}}^{tal}$} & AIC   & 0.998 & 0.160 & $-$ & 0.465 & 1.000 & 0.156 & $-$ & 0.220 & 0.869 & 0.150 \\
          &       &       & HQ    & 0.996 & 0.123 & $-$ & 0.420 & 1.000 & 0.024 & $-$ & 0.060 & 0.930 & 0.020 \\
          &       &       & BIC   & 0.984 & 0.076 & $-$ & 0.350 & 1.000 & 0.006 & $-$ & 0.010 & 0.953 & 0.000 \\
\cmidrule{2-14}          & \multirow{4}[2]{*}{50,000} & $\widehat{\bm{\beta}}^{mle}$ & t test ($\alpha=0.05$) & 0.718 & 0.050 & $-$ & 0.085 & 1.000 & 0.112 & $-$ & 0.185 & 0.892 & 0.050 \\
          &       & \multirow{3}[1]{*}{$\widehat{\bm{\beta}}^{tal}$} & AIC   & 0.998 & 0.100 & $-$ & 0.450 & 1.000 & 0.147 & $-$ & 0.160 & 0.892 & 0.110 \\
          &       &       & HQ    & 0.998 & 0.075 & $-$ & 0.410 & 1.000 & 0.032 & $-$ & 0.040 & 0.942 & 0.020 \\
          &       &       & BIC   & 0.998 & 0.053 & $-$ & 0.330 & 1.000 & 0.004 & $-$ & 0.015 & 0.963 & 0.000 \\
\cmidrule{2-14}          & \multirow{4}[2]{*}{100,000} & $\widehat{\bm{\beta}}^{mle}$ & t test ($\alpha=0.05$) & 0.788 & 0.045 & $-$ & 0.130 & 1.000 & 0.109 & $-$ & 0.220 & 0.900 & 0.040 \\
          &       & \multirow{3}[1]{*}{$\widehat{\bm{\beta}}^{tal}$} & AIC   & 0.988 & 0.045 & $-$ & 0.435 & 1.000 & 0.154 & $-$ & 0.200 & 0.901 & 0.160 \\
          &       &       & HQ    & 0.988 & 0.034 & $-$ & 0.385 & 1.000 & 0.030 & $-$ & 0.040 & 0.953 & 0.020 \\
          &       &       & BIC   & 0.988 & 0.024 & $-$ & 0.345 & 1.000 & 0.001 & $-$ & 0.010 & 0.969 & 0.000 \\
    \midrule
    \multirow{12}[6]{*}{\textbf{DGP 4}} & \multirow{4}[2]{*}{25,000} & $\widehat{\bm{\beta}}^{mle}$ & t test ($\alpha=0.05$) & 0.588 & 0.071 & $-$ & 0.100 & 1.000 & 0.101 & $-$ & 0.195 & 0.870 & 1.000 \\
          &       & \multirow{3}[1]{*}{$\widehat{\bm{\beta}}^{tal}$} & AIC   & 0.982 & 0.414 & $-$ & 0.665 & 1.000 & 0.176 & $-$ & 0.155 & 0.792 & 1.000 \\
          &       &       & HQ    & 0.966 & 0.213 & $-$ & 0.525 & 1.000 & 0.044 & $-$ & 0.040 & 0.892 & 1.000 \\
          &       &       & BIC   & 0.916 & 0.086 & $-$ & 0.450 & 1.000 & 0.004 & $-$ & 0.005 & 0.934 & 1.000 \\
\cmidrule{2-14}          & \multirow{4}[2]{*}{50,000} & $\widehat{\bm{\beta}}^{mle}$ & t test ($\alpha=0.05$) & 0.718 & 0.055 & $-$ & 0.110 & 1.000 & 0.085 & $-$ & 0.170 & 0.900 & 1.000 \\
          &       & \multirow{3}[1]{*}{$\widehat{\bm{\beta}}^{tal}$} & AIC   & 0.996 & 0.310 & $-$ & 0.630 & 1.000 & 0.145 & $-$ & 0.140 & 0.832 & 1.000 \\
          &       &       & HQ    & 0.992 & 0.178 & $-$ & 0.535 & 1.000 & 0.038 & $-$ & 0.025 & 0.907 & 1.000 \\
          &       &       & BIC   & 0.982 & 0.109 & $-$ & 0.500 & 1.000 & 0.001 & $-$ & 0.005 & 0.936 & 1.000 \\
\cmidrule{2-14}          & \multirow{4}[2]{*}{100,000} & $\widehat{\bm{\beta}}^{mle}$ & t test ($\alpha=0.05$) & 0.828 & 0.039 & $-$   & 0.090 & 1.000 & 0.094 & $-$   & 0.190 & 0.920 & 1.000 \\								
          &       & \multirow{3}[1]{*}{$\widehat{\bm{\beta}}^{tal}$} & AIC   & 1.000 & 0.248 & $-$   & 0.560 & 1.000 & 0.134 & $-$   & 0.100 & 0.859 & 1.000 \\								
          &       &       & HQ    & 1.000 & 0.111 & $-$   & 0.435 & 1.000 & 0.036 & $-$   & 0.040 & 0.931 & 1.000 \\								
          &       &       & BIC   & 0.998 & 0.060 & $-$   & 0.345 & 1.000 & 0.000 & $-$   & 0.000 & 0.962 & 1.000 \\								

    \midrule
    \multirow{12}[6]{*}{\textbf{DGP 5}} & \multirow{4}[2]{*}{25,000} & $\widehat{\bm{\beta}}^{mle}$ & t test ($\alpha=0.05$) & 0.458 & 0.065 & 0.280 & 0.140 & 1.000 & 0.120 & 1.000 & 0.225 & 0.832 & 0.080 \\
          &       & \multirow{3}[1]{*}{$\widehat{\bm{\beta}}^{tal}$} & AIC   & 0.998 & 0.161 & 0.930 & 0.410 & 1.000 & 0.152 & 1.000 & 0.175 & 0.874 & 0.190 \\
          &       &       & HQ    & 0.998 & 0.098 & 0.900 & 0.365 & 1.000 & 0.039 & 1.000 & 0.045 & 0.934 & 0.050 \\
          &       &       & BIC   & 0.998 & 0.083 & 0.890 & 0.350 & 1.000 & 0.006 & 1.000 & 0.005 & 0.950 & 0.000 \\
\cmidrule{2-14}          & \multirow{4}[2]{*}{50,000} & $\widehat{\bm{\beta}}^{mle}$ & t test ($\alpha=0.05$) & 0.620 & 0.058 & 0.390 & 0.135 & 1.000 & 0.111 & 1.000 & 0.215 & 0.862 & 0.060 \\
          &       & \multirow{3}[1]{*}{$\widehat{\bm{\beta}}^{tal}$} & AIC   & 0.993 & 0.085 & 0.990 & 0.355 & 1.000 & 0.151 & 1.000 & 0.175 & 0.899 & 0.150 \\
          &       &       & HQ    & 0.993 & 0.058 & 0.980 & 0.285 & 1.000 & 0.036 & 1.000 & 0.020 & 0.954 & 0.030 \\
          &       &       & BIC   & 0.993 & 0.044 & 0.970 & 0.260 & 1.000 & 0.003 & 1.000 & 0.005 & 0.969 & 0.000 \\
\cmidrule{2-14}          & \multirow{4}[2]{*}{100,000} & $\widehat{\bm{\beta}}^{mle}$ & t test ($\alpha=0.05$) & 0.755 & 0.036 & 0.510 & 0.095 & 1.000 & 0.092 & 1.000 & 0.180 & 0.899 & 0.040 \\
          &       & \multirow{3}[1]{*}{$\widehat{\bm{\beta}}^{tal}$} & AIC   & 0.993 & 0.033 & 1.000 & 0.300 & 1.000 & 0.124 & 1.000 & 0.170 & 0.924 & 0.140 \\
          &       &       & HQ    & 0.993 & 0.025 & 1.000 & 0.220 & 1.000 & 0.028 & 1.000 & 0.040 & 0.968 & 0.010 \\
          &       &       & BIC   & 0.993 & 0.020 & 1.000 & 0.195 & 1.000 & 0.000 & 1.000 & 0.005 & 0.981 & 0.000 \\
    \midrule
    \multirow{12}[5]{*}{\textbf{DGP 6}} & \multirow{4}[2]{*}{25,000} & $\widehat{\bm{\beta}}^{mle}$ & t test ($\alpha=0.05$) & 0.480 & 0.058 & 0.290 & 0.095 & 1.000 & 0.091 & 1.000 & 0.215 & 0.852 & 1.000 \\
          &       & \multirow{3}[1]{*}{$\widehat{\bm{\beta}}^{tal}$} & AIC   & 0.890 & 0.304 & 0.970 & 0.640 & 1.000 & 0.148 & 1.000 & 0.135 & 0.818 & 1.000 \\
          &       &       & HQ    & 0.873 & 0.219 & 0.970 & 0.535 & 1.000 & 0.039 & 1.000 & 0.020 & 0.880 & 1.000 \\
          &       &       & BIC   & 0.828 & 0.144 & 0.950 & 0.445 & 1.000 & 0.004 & 1.000 & 0.000 & 0.909 & 1.000 \\
\cmidrule{2-14}          & \multirow{4}[2]{*}{50,000} & $\widehat{\bm{\beta}}^{mle}$ & t test ($\alpha=0.05$) & 0.683 & 0.070 & 0.370 & 0.095 & 1.000 & 0.098 & 1.000 & 0.195 & 0.877 & 1.000 \\
          &       & \multirow{3}[1]{*}{$\widehat{\bm{\beta}}^{tal}$} & AIC   & 0.940 & 0.241 & 0.990 & 0.600 & 1.000 & 0.175 & 1.000 & 0.180 & 0.834 & 1.000 \\
          &       &       & HQ    & 0.930 & 0.140 & 0.990 & 0.500 & 1.000 & 0.034 & 1.000 & 0.045 & 0.911 & 1.000 \\
          &       &       & BIC   & 0.903 & 0.076 & 0.990 & 0.425 & 1.000 & 0.011 & 1.000 & 0.015 & 0.936 & 1.000 \\
\cmidrule{2-14}          & \multirow{4}[1]{*}{100,000} & $\widehat{\bm{\beta}}^{mle}$ & t test ($\alpha=0.05$) & 0.813 & 0.048 & 0.450 & 0.070 & 1.000 & 0.069 & 1.000 & 0.155 & 0.914 & 1.000 \\
          &       & \multirow{3}[0]{*}{$\widehat{\bm{\beta}}^{tal}$} & AIC   & 0.918 & 0.095 & 0.990 & 0.510 & 1.000 & 0.125 & 1.000 & 0.090 & 0.894 & 1.000 \\
          &       &       & HQ    & 0.918 & 0.036 & 0.990 & 0.385 & 1.000 & 0.033 & 1.000 & 0.025 & 0.945 & 1.000 \\
          &       &       & BIC   & 0.913 & 0.023 & 0.990 & 0.340 & 1.000 & 0.001 & 1.000 & 0.005 & 0.960 & 1.000 \\
          \bottomrule
          \bottomrule
    \end{tabular}%
    }
\end{sidewaystable}

\section{Empirical Study}
\label{sec:empirical_study}
We study the high-frequency excess loss distributions of nine large liquid U.S. stocks: American Express (AXP), Boeing (BA), General Electric (GE), Home Depot (HD), IBM, Johnson and Johnson (JNJ), JPMorgan Chase (JPM), Coca-Cola (KO), and ExxonMobil (XOM). Our data covers all transactions observed from January 2006 to December 2014. Market uncertainty and liquidity being elusive concepts, we study their impact on the excess loss distribution using as predictors several high-frequency volatility and liquidity indicators, and select the most appropriate ones with the two-step ALMLE developed in Section \ref{sec:LASSO}. We perform an in-sample analysis providing an economic interpretation for the impact of the selected predictors on the excess loss distribution, and an out-of-sample VaR forecast analysis to assess the goodness of fit of the predicted excess loss distribution.

\subsection{Variables description} 
\label{subsec:var_prepare}
The raw intraday data of the studied stocks contain transaction timestamps in milliseconds, transaction prices per share, and transaction volume in shares for each trade. We cleaned the raw data according to standard procedures in \cite{brownlees2006financial} and \cite{barndorff2009realized}. Since transaction data are irregularly-spaced, we need to define an equally-spaced grid at a fixed frequency to analyse losses with our model. We choose to analyse losses at the five minute frequency. Let $P_{t,i}$ be the transaction price of the $i$-th trade in the $t$-th five minute interval, and let $V_{t,i}$ be the corresponding quantity of traded shares, with $0\leq i \leq n_t$ where $n_t$ is the number of trades in the $t$-th five minute interval and $0<t\leq T$. We define 5-min prices, $P_t$, as the median transaction price in the $t$-th five minute interval, and compute 5-min losses as the negative $t$-th return, $R_{t}:= \log( P_{t} ) - \log(P_{t-1} )$. To obtain the time series of excess losses we consider a dynamic threshold accounting for the time-varying behavior of losses at high-frequency. Specifically, the threshold $u_t$ at time $t$ is defined as the $90\%$-quantile of the losses observed over the period $\left(t-1,t-h\right)$, with $h>1$ the moving window size. We consider 12 possible values of $h$ ranging from one week to twelve weeks.

Liquidity refers to the ability to trade large volume of a financial instrument with low price impact, cost and postponement. As liquidity can be decomposed into different dimensions (Harris et al., 1990), we consider several liquidity indicators as possible predictors. Similarly, to characterize market uncertainty we consider several indicators for the observed dispersion of transaction prices. Moreover, to disentangle the impact of trading activity at different frequencies, we build our set of candidate predictors considering both information within the $t$-th five minute interval and across neighbourhoods of the $t$-th five minute interval. Let $P_{t,BU}:=[P_{t,1},\ldots,P_{t,n_t}]'$ and $R_{t,BU} :=[R_{t,1},\ldots,R_{t,n_t}]'$ be the vectors of traded prices and trade returns observed within the $t$-th five minute interval, with $R_{t,i}:= \log(P_{t,i}) - \log(P_{t,i-1})$. Let $T_w$ be a neighborhood size, and define $P_{t,T_w}:= [P_t,\ldots, P_{t-T_w+1}]'$ the vector of 5-min prices within a neighborhood of size $T_w$ and $ R_{t,T_w}: = \log(P_{t,T_w}) - \log(P_{t-1,T_w})$ the corresponding vector of returns. Let $\text{dur}_{t,i}$ denote the execution duration of the $i$-th transaction in the $t$-th five minute interval, i.e. the time difference between the order executed time and order placed time. Table \ref{tab:liquidity_indicators} lists the liquidity predictors we consider in the analysis. They are classified according to their frequency, i.e. within or across the five minute interval, and by their nature of price impact or spread proxies \citep{goyenko2009liquidity} or volatility of liquidity measures. Table \ref{tab:volatility_indicators} lists the volatility predictors we consider in the analysis and are classified according to the frequency at which they are computed, i.e., within or across the five minute interval.
\begin{table}[htbp]
  \centering
  \caption{Liquidity measures proxying for price impact (PI), spread (S) and volatility of liquidity (Vol) computed using information within the five minute interval (W), across (A) 5-min observations in a neighbourhood of size $T_w$, or as a ratio (R) between the two frequencies. $\bigtriangleup$ denotes the first difference operator for vectors; cov(,) and cor(,) denote the covariance and the correlation between two input variables, respectively; var() denotes the variance of the input variable or vector.}
  \resizebox{0.98\textwidth}{!}{%
    \begin{tabular}{llll}
    \hline
    \hline
    \\
    \parbox{2.5cm}{\textbf{Frequency}} & \parbox{2.5cm}{\textbf{Proxy}} & \parbox{6cm}{\textbf{Liquidity Predictors}} & \parbox{5cm}{\textbf{Formula}} \\[10pt] 
    \hline\\
    W & PI & Transaction Volume &  $\text{TV}_t = \sum_{i=1}^{n_t} P_{t,i}\, V_{t,i}$ \\[10pt] 
    W & PI & Transaction Quantity &  $\text{TQ}_t = \sum_{i=1}^{n_t} \, V_{t,i}$ \\[10pt] 
    W & Vol & Micro Transaction Volume Volatility &  $\text{MTVV}_t = \sqrt{ \frac{1}{n_t}\sum_{j=1}^{n_t} \left( P_{t,j}\, V_{t,j} - \frac{1}{n_t}\sum_{i=1}^{n_t} P_{t,i}\, V_{t,i} \right)^2 } $ \\[10pt] 
    W & Vol & Micro Volatility of Trading Quantity in shares  &  $\text{MTQV}_t = \sqrt{ \frac{1}{n_t}\sum_{j=1}^{n_t} \left( V_{t,j} - \frac{1}{n_t}\sum_{i=1}^{n_t} V_{t,i} \right)^2 }$\\[10pt] 
    W & PI & Amihud Illiquidity Measure  & 
        $\text{AM}_t = \frac{1}{n_t}\sum_{i=1}^{n_t} \frac{ \abs{R_{t,i} } }{ P_{t,i}\,V_{t,i} } $
         \\[10pt] 
    W & PI & Extended Amihud Measures \citep{goyenko2009liquidity}  &  $   \text{EAM}_t = \frac{\max(P_{t,BU}) - \min(P_{t,BU})}{ \text{TV}_{t} }$\\[10pt] 
    W & PI & Transaction Duration & $\text{dur}_t = \frac{\sum_{i=1}^{n_t} \text{dur}_{t,i}}{n_t}$
         \\[10pt]
    A & S & Roll~\citep{roll1984simple}  &  $\text{Roll}_t = \text{cov}\left(\bigtriangleup P_{t,T_w}, \bigtriangleup P_{t-1.T_w} \right)$\\[10pt] 
    A & S & Modified Roll  &  $\text{RollMod}_t = \frac{\text{cov}\left(\bigtriangleup P_{t,T_w}, \bigtriangleup P_{t-1.T_w} \right)}{P_t^m}$\\[10pt] 
    A & S & Negative Roll & $\text{Roll}_t^- = \text{Roll}_t\, \mathbbm{1}\{\text{Roll}_t < 0\}$  \\[10pt] 
    A & S & Negative Modified Roll   & $\text{RollMod}_t^- = \text{RollMod}_t\, \mathbbm{1}\{\text{RollMod}_t < 0\}$  \\[10pt] 
    A & S & Return Autocorrelation \citep{grossman1988liquidity}  &  $\text{RAC}_t = \text{cor}\left(R_{t,T_w}, R_{t-1.T_w}\right)$  \\[10pt] 
    A & PI & Amihud Illiquidity Measure  & 
        $ \text{AMI}_t = \frac{1}{T_w}\sum_{j=0}^{T_w-1} \frac{ \abs{R_{t-j} } }{ \text{TV}_{t-j} } $
         \\[10pt] 
    A & Vol & Transaction Volume Volatility  & $ \text{TVV}_t = \sqrt{ \frac{1}{T_w}\sum_{j=0}^{T_w-1} \left( \text{TV}_{t-j} - \overline{\text{TV}}_{t,T_w} \right)^2 }$ \\[10pt]  
    A & Vol & Relative Transaction Volume Volatility  & $ \text{RTVV}_t = \frac{ \text{TVV}_t }{ \frac{1}{T_w} \sum_{j=0}^{T_w - 1} \text{TV}_{t-j} }$ \\[10pt] 
    A & Vol & Trading Quantity Volatility  &  $\text{TQV}_{t} = \sqrt{ \frac{1}{T_w}\sum_{j=0}^{T_w-1} \left( \text{TQ}_{t-j} - \frac{1}{T_w} \sum_{j=0}^{T_w - 1} \text{TQ}_{t-j} \right)^2 }$ \\[10pt] 
    A & Vol & Relative Trading Quantity Volatility &  $\text{RTQV}_{t} = \frac{\text{TQV}_{t} }{ \overline{\text{TQ}}_{t,T_w} }$ \\[10pt] 
    R & S & 
    Variance Ratio \citep{hasbrouck1988liquidity}  & $ \text{VR}_t = \frac{T_w\cdot \text{var}( \log(P_{t,BU}) - \log(P_t)) }{n_t\,\text{var}(R_{t,T_w})}$ \\[10pt]
    \hline
    \hline
    \end{tabular}%
    }
  \label{tab:liquidity_indicators}%
\end{table}%

\begin{table}[htbp]
  \centering
  \caption{Volatility measures computed using information within the five minute interval (W), across (A) 5-min observations in a neighbourhood of size $T_w$, or as a ratio (R) between the two frequencies.}
    \begin{tabular}{lll}
    \toprule
    \toprule
    \parbox{4cm}{\textbf{Frequency} } & \parbox{5cm}{\textbf{Volatility Predictors} } & \textbf{Formulas} \\[10pt]
    \midrule
    W & Micro Noise Return Volatility   & $\text{MNRV}_{t} = \sqrt{\frac{1}{n_t}\sum_{i=1}^{n_t} \left( \log(P_{t,i}) - \log(P_t) \right)^2}$  \\[10pt]
    W & Micro Realized Volatility   & $\text{MRV}_{t} = \sqrt{\sum_{i=1}^{n_t} \left( \log(P_{t,i}) - \log(P_{t,i-1}) \right)^2}$ \\[10pt]
    \midrule
    A & Realized Volatility   &  $\text{RV}_{t,T_w} = \sqrt{\sum_{j=0}^{T_w-1 } R_{t-j}^2}$\\[10pt]
    \midrule
    R & MNRV2RV & $\text{MNRV2RV} = \frac{\text{MNRV}_{t}}{\text{RV}_{t}}$ \\[10pt]
    R & MRV2RV &  $\text{MRV2RV} = \frac{\text{MRV}_{t}}{\text{RV}_{t}}$ \\
    \bottomrule
    \bottomrule
    \end{tabular}%
  \label{tab:volatility_indicators}%
\end{table}%

\subsection{In-sample estimates}
We divide each time series into an in-sample period covering the first 90\% of the observations and an out-of-sample period spanning the last 10\% of the sample. 
We model the excess losses $\{y_t\}$ of each stock with the time-varying GPD regression model in \eqref{Model:k_sigma_models}-\eqref{Model:k_sigma_models_2ndEqu}, using the variables defined in Tables~\ref{tab:liquidity_indicators} and~\ref{tab:volatility_indicators}, with $T_w\in\{2,6,12\}$, as possible predictors in both scale and shape parameters.
Coefficient estimates obtained with the two-step ALMLE are presented in Tables~\ref{tab:model_estimation_results_IS_10stocks_K} and ~\ref{tab:model_estimation_results_IS_10stocks_Sigma} for the shape and scale parameters, respectively. 

Results for the shape parameter in Table~\ref{tab:model_estimation_results_IS_10stocks_K} show that estimated coefficients have almost always the same sign across the stocks. 
As to liquidity predictors, we find that price impact proxies are selected for almost all the stocks, suggesting that they better capture liquidity effects on extreme losses. In particular, TV and TQ  display positive coefficients while AM and EAM display negative coefficients, entailing that larger extreme losses are associated with high levels of liquidity in the last five minutes. Although counter-intuitive at first, this result is very interesting when read together with the other selected variables. As to the volatility of liquidity, we notice that RTVV$(T_w=6)$ and RTQV$(T_w=6)$ are selected across most of the stocks and display large and positive coefficients, indicating that extreme losses tend to be larger during periods of high volatility of liquidity. Almost for every stock, we select the ratio MRV2RV($T_w=12$), essentially capturing the impact of the volatility of volatility or jump risk on extreme losses, and associate a positive coefficient to it, conveying the idea that extreme losses tend to be larger during periods of high uncertainty. Altogether these results are coherent with the findings in \cite{brogaard2018}, i.e. that market markers amplify extreme price movements while withdrawing from the market after large uncertainty shocks that caused their liquidity supply to be outstripped by liquidity demand.

Table~\ref{tab:model_estimation_results_IS_10stocks_Sigma} shows that more variables are selected for the scale parameter but their pattern is less stable across stocks. In general, we notice that the autoregressive component contributes to the dynamics, and that the realized volatility predictor computed within the five-minute interval is always selected and displays positive coefficient. This is coherent with the fact that the scale parameter captures the time-varying heteroscedasticity in the data.

For comparison purposes, we report estimated regression coefficients for the shape and scale parameters obtained with MLE in Tables \ref{tab:model_estimation_results_IS_10stocks_K_MLE}-\ref{tab:model_estimation_results_IS_10stocks_sigma_MLE}. All of the estimated coefficients are nonzero and we cannot compute the corresponding standard errors because the obtained Fisher information matrix of the MLE is not positive definitive. This makes variable interpretation very difficult if not impossible.

\begin{table}[htbp]
  \centering
  \caption{Empirical estimates of the regression coefficients for the shape parameter obtained with the two-step ALMLE.}
    \begin{tabular}{l|rrrrrrrrr|r}
    \toprule
    \toprule
     \diagbox{Predictors}{Stocks} & \multicolumn{1}{l}{AXP} & \multicolumn{1}{l}{BA} & \multicolumn{1}{l}{GE} & \multicolumn{1}{l}{HD} & \multicolumn{1}{l}{IBM} & \multicolumn{1}{l}{JNJ} & \multicolumn{1}{l}{JPM} & \multicolumn{1}{l}{KO} & \multicolumn{1}{l|}{XOM} & \multicolumn{1}{l}{selection per stock} \\
    \midrule
    TV    & 1.105 & 0.676 &   & 0.581 & 1.912 & 1.205 & 0.660 & 1.109 & 2.014 & 0.89 \\
    TQ    & 1.314 & 1.101 & 0.501 & 0.874 & 0.813 & 1.130 & 0.756 &   & 0.0004 & 0.89 \\
    AM    & -0.345 &   & 0.719 & -0.037 &  &  & -0.522 & -0.030 &   & 0.56 \\
    MTVV   &  &   &   &   &   & 0.067 &   &   &   & 0.11 \\
    EAM   & -1.637 & -1.760 &  & -0.913 & -1.789 & -1.681 & -1.396 & -1.650 & -0.201 & 0.89 \\
    MTQV   & 0.200 &  & 0.166 &  &  & 0.009 & 0.177 &  &  & 0.44 \\
    MRV   &  &  & 0.328 &  & -0.901 & -0.367 & -0.274 &  &  & 0.44 \\
    MNRV  &  & -0.794 & -0.242 & -0.491 &  &  &  & 0.020 &  & 0.44 \\
    dur   &  &  &  &  &  &  &  &  &  &  0.00 \\
    AMI ($T_w=2$) &  &  &  &  &  &  &  &  &  &  0.00 \\
    VR ($T_w=2$) &  &  &  &  &  &  &  &  &  &  0.00 \\
    RV ($T_w=2$) &  &  &  &  &  &  &  &  &  &  0.00 \\
    MNRV2RV ($T_w=2$) &  &  &  &  &  &  &  & -0.084 &  & 0.11 \\
    MRV2RV ($T_w=2$) &  &  &  &  &  &  &  & 0.095 &  & 0.11 \\
    AMI ($T_w=6$) &  &  &  &  &  &  &  & -0.003 &  & 0.11 \\
    $\text{Roll}$ ($T_w=6$) &  &  &  &  &  &  &  &  &  &  0.00 \\
    $\text{Roll}^-$ ($T_w=6$) &  &  &  &  &  &  &  &  &  &  0.00 \\
    $\text{RollMod}$ ($T_w=6$) &  &  &  &  &  &  &  &  &  &  0.00 \\
    $\text{RollMod}^-$ ($T_w=6$) &  &  &  &  &  &  &  &  &  &  0.00 \\
    TVV ($T_w=6$) &  &  &  &  &  &  &  &  &  &   0.00\\
    TQV ($T_w=6$) &  &  &  &  &  &  &  &  &  &   0.00\\
    RTVV ($T_w=6$) &  & 0.068 &  & 0.298 &  &  & 0.159 & 0.218 & 0.264 & 0.56 \\
    RTQV ($T_w=6$) &  & 0.069 &  & 0.300 &  &  & 0.159 & 0.218 & 0.270 & 0.56 \\
    RAC ($T_w=6$) &  &  &  &  &  &  &  &  &  &   0.00\\
    VR ($T_w=6$) &  &  &  &  &  &  &  &  &  &   0.00\\
    RV ($T_w=6$) &  &  & -0.203 &  &  &  &  &  &  & 0.11 \\
    MNRV2RV ($T_w=6$) &  &  & -0.415 &  &  &  &  & -0.185 &  & 0.22 \\
    MRV2RV ($T_w=6$) &  &  &  &  &  &  & 0.194 & -0.010 & 0.0003 & 0.33 \\
    AMI ($T_w=12$) &  &  & 0.524 &  &  &  &  & 0.648 &  & 0.22 \\
    $\text{Roll}$ ($T_w=12$) &  &  &  &  &  &  &  &  &  &   0.00\\
    $\text{Roll}^-$ ($T_w=12$) &  &  &  &  &  &  &  &  &  &   0.00\\
    $\text{RollMod}$ ($T_w=12$) &  &  &  &  &  &  &  &  &  &   0.00\\
    $\text{RollMod}^-$ ($T_w=12$) &  &  &  &  &  &  &  &  &  &  0.00 \\
    TVV ($T_w=12$) &  &  & -0.851 &  &  &  &  &  &  & 0.11 \\
    TQV ($T_w=12$) &  &  &  &  &  &  &  &  &  &   0.00\\
    RTVV ($T_w=12$) &  &  &  &  &  &  &  & -0.090 &  & 0.11 \\
    RTQV ($T_w=12$) &  &  &  &  &  &  &  & -0.090 &  & 0.11 \\
    RAC ($T_w=12$) &  &  &  &  &  &  &  &  &  &  0.00\\
    VR ($T_w=12$) & 0.415 &  &  &  &  &  & 0.244 &  & -0.00002 & 0.33 \\
    RV ($T_w=12$) &  &  &  &  &  &  &  & -0.436 &  & 0.11 \\
    MNRV2RV ($T_w=12$) &  &  &  &  &  &  &  & 0.297 &  & 0.11 \\
    MRV2RV ($T_w=12$) &  & 0.994 & 0.608 & 0.287 & 0.554 & 0.548 & 0.487 & 0.544 & 0.766 & 0.89 \\
    \midrule
    total number of selected variables & 6     & 7     & 10    & 8     & 5     & 7     & 11    & 17    & 8     &  \\
    \bottomrule
    \bottomrule
    \end{tabular}%
  \label{tab:model_estimation_results_IS_10stocks_K}%
\end{table}%

\begin{table}[htbp]
  \centering
  \caption{Empirical estimates of the regression coefficients for the scale parameter obtained with the two-step ALMLE.}
    \begin{tabular}{l|rrrrrrrrr|r}
    \toprule
    \toprule
    \diagbox{Predictors}{Stocks} & \multicolumn{1}{l}{AXP} & \multicolumn{1}{l}{BA} & \multicolumn{1}{l}{GE} & \multicolumn{1}{l}{HD} & \multicolumn{1}{l}{IBM} & \multicolumn{1}{l}{JNJ} & \multicolumn{1}{l}{JPM} & \multicolumn{1}{l}{KO} & \multicolumn{1}{l|}{XOM} & \multicolumn{1}{l}{selection per stock} \\
    \midrule
    TV    & -0.004 &  & -0.005 &  & 0.042 &  & -0.044 & 0.212 & 0.095 & 0.67 \\
    TQ    & 0.039 & 0.110 & 0.015 & 0.042 &  & 0.047 & 0.092 & -0.102 & -0.095 & 0.89 \\
    AM    &  &  & 0.024 & 0.050 &  &  &  &  &  & 0.22 \\
    MTVV   &  & 0.048 &  &  & -0.029 &  & -0.137 & -0.260 &  & 0.44 \\
    EAM   &  &  &  & -0.040 & -0.012 &  &  &  & -0.043 & 0.33 \\
    MTQV   &  & -0.123 & 0.016 &  & 0.021 &  & 0.135 & 0.168 & -0.003 & 0.67 \\
    MRV   & 0.049 &  & 0.032 &  &  &  & 0.057 &  &  & 0.33 \\
    MNRV  & 0.049 & 0.139 & 0.056 & 0.055 & 0.067 & 0.079 & 0.048 & 0.128 & 0.091 & 1.00 \\
    dur   &  &  &  &  &  &  &  &  &  & 0.00 \\
    AMI ($T_w=2$) &  &  &  &  & -0.018 &  &  &  &  & 0.11 \\
    VR ($T_w=2$) &  &  &  &  &  &  &  &  &  & 0.00 \\
    RV ($T_w=2$) &  &  &  &  & 0.017 &  & -0.021 &  &  & 0.22 \\
    MNRV2RV ($T_w=2$) &  &  &  &  &  &  &  &  &  & 0.00 \\
    MRV2RV ($T_w=2$) &  &  &  &  &  &  &  &  &  & 0.00 \\
    AMI ($T_w=6$) &  &  &  &  & -0.014 &  &  &  & -0.003 & 0.22 \\
    $\text{Roll}$ ($T_w=6$) &  &  &  &  &  & -0.054 &  &  &  & 0.11 \\
    $\text{Roll}^-$ ($T_w=6$) &  &  &  & -0.035 &  &  & -0.086 &  &  & 0.22 \\
    $\text{RollMod}$ ($T_w=6$) & -0.061 &  &  &  & -0.048 &  &  &  & -0.050 & 0.33 \\
    $\text{RollMod}^-$ ($T_w=6$) & 0.077 &  & 0.008 & 0.051 & 0.053 & 0.062 & 0.104 &  & 0.061 & 0.78 \\
    TVV ($T_w=6$) &  &  &  & 0.006 & 0.029 &  &  &  &  & 0.22 \\
    TQV ($T_w=6$) &  &  &  & -0.015 & -0.048 & -0.006 &  &  &  & 0.33 \\
    RTVV ($T_w=6$) &  &  &  &  & 0.015 &  &  &  &  & 0.11 \\
    RTQV ($T_w=6$) &  &  &  &  & 0.0003 &  &  &  &  & 0.11 \\
    RAC ($T_w=6$) &  &  &  &  &  &  &  &  &  & 0.00 \\
    VR ($T_w=6$) &  &  &  &  & -0.003 &  &  &  &  & 0.11 \\
    RV ($T_w=6$) &  &  & 0.017 & 0.038 &  &  &  &  &  & 0.22 \\
    MNRV2RV ($T_w=6$) & -0.002 & -0.056 &  &  &  & -0.037 &  &  & -0.082 & 0.44 \\
    MRV2RV ($T_w=6$) & -0.028 & 0.041 &  &  &  &  &  &  & 0.044 & 0.33 \\
    AMI ($T_w=12$) &  &  &  &  & 0.039 &  &  &  &  & 0.11 \\
    $\text{Roll}$ ($T_w=12$) &  &  & -0.005 &  &  &  & 0.027 &  &  & 0.22 \\
    $\text{Roll}^-$ ($T_w=12$) &  &  & 0.006 &  &  &  &  &  &  & 0.11 \\
    $\text{RollMod}$ ($T_w=12$) & -0.005 &  &  &  &  &  & -0.038 &  & -0.014 & 0.33 \\
    $\text{RollMod}^-$ ($T_w=12$) & 0.004 &  &  &  &  &  & 0.006 &  & 0.022 & 0.33 \\
    TVV ($T_w=12$) &  & -0.036 &  & -0.009 &  &  &  &  &  & 0.22 \\
    TQV ($T_w=12$) & -0.013 &  &  &  & -0.001 & -0.010 & -0.011 &  &  & 0.44 \\
    RTVV ($T_w=12$) &  &  &  &  &  &  &  &  & 0.003 & 0.11 \\
    RTQV ($T_w=12$) &  & 0.046 &  &  &  &  &  &  &  & 0.11 \\
    RAC ($T_w=12$) &  &  &  &  &  &  &  &  &  & 0.00 \\
    VR ($T_w=12$) &  &  &  &  & -0.004 &  &  &  & -0.032 & 0.22 \\
    RV ($T_w=12$) &  &  &  &  &  & 0.027 &  &  & 0.034 & 0.22 \\
    MNRV2RV ($T_w=12$) & 0.014 & 0.007 & 0.032 & 0.037 &  & 0.061 & 0.023 &  & 0.103 & 0.78 \\
    MRV2RV ($T_w=12$) & 0.022 &  &  &  & 0.021 & -0.005 & -0.00005 & 0.014 & -0.016 & 0.67 \\
    $\log(\sigma_{t-1})$ & 0.849 & 0.631 & 0.758 & 0.724 & 0.832 & 0.778 & 0.829 & 0.668 & 0.770 & 1.00  \\
    \midrule
    total number of selected variables & 14    & 10    & 12    & 12    & 20    & 11    & 16    & 7     & 18    &  \\
    \bottomrule
    \bottomrule
    \end{tabular}%
  \label{tab:model_estimation_results_IS_10stocks_Sigma}%
\end{table}%
 
\begin{table}[htbp]
  \centering
  \caption{Empirical estimates of the regression coefficients for the scale parameter obtained with the MLE.}
    \begin{tabular}{l|llllllllll}
    \toprule
    \toprule
    \diagbox{Predictors}{Stocks} & AXP  & BA   & GE   & HD   & IBM & JNJ  & JPM & KO   & \multicolumn{1}{l|}{XOM} & selection per stock \\
    \midrule
    TV    & \multicolumn{1}{r}{0.396} & \multicolumn{1}{r}{0.264} & \multicolumn{1}{r}{0.116} & \multicolumn{1}{r}{0.264} & \multicolumn{1}{r}{0.563} & \multicolumn{1}{r}{0.631} & \multicolumn{1}{r}{0.239} & \multicolumn{1}{r}{0.602} & \multicolumn{1}{r|}{0.239} & \multicolumn{1}{r}{1.00} \\
    TQ    & \multicolumn{1}{r}{0.469} & \multicolumn{1}{r}{0.313} & \multicolumn{1}{r}{0.173} & \multicolumn{1}{r}{0.326} & \multicolumn{1}{r}{0.457} & \multicolumn{1}{r}{0.619} & \multicolumn{1}{r}{0.259} & \multicolumn{1}{r}{0.519} & \multicolumn{1}{r|}{0.216} & \multicolumn{1}{r}{1.00} \\
    AM    & \multicolumn{1}{r}{-0.083} & \multicolumn{1}{r}{-0.167} & \multicolumn{1}{r}{0.303} & \multicolumn{1}{r}{-0.042} & \multicolumn{1}{r}{-0.287} & \multicolumn{1}{r}{-0.394} & \multicolumn{1}{r}{-0.067} & \multicolumn{1}{r}{-0.221} & \multicolumn{1}{r|}{-0.080} & \multicolumn{1}{r}{1.00} \\
    TVV   & \multicolumn{1}{r}{0.249} & \multicolumn{1}{r}{0.017} & \multicolumn{1}{r}{0.064} & \multicolumn{1}{r}{0.094} & \multicolumn{1}{r}{0.135} & \multicolumn{1}{r}{0.115} & \multicolumn{1}{r}{0.055} & \multicolumn{1}{r}{0.090} & \multicolumn{1}{r|}{0.059} & \multicolumn{1}{r}{1.00} \\
    EAM   & \multicolumn{1}{r}{-0.621} & \multicolumn{1}{r}{-0.406} & \multicolumn{1}{r}{-0.062} & \multicolumn{1}{r}{-0.451} & \multicolumn{1}{r}{-0.525} & \multicolumn{1}{r}{-1.095} & \multicolumn{1}{r}{-0.236} & \multicolumn{1}{r}{-0.966} & \multicolumn{1}{r|}{-0.350} & \multicolumn{1}{r}{1.00} \\
    TQV   & \multicolumn{1}{r}{0.195} & \multicolumn{1}{r}{0.020} & \multicolumn{1}{r}{0.076} & \multicolumn{1}{r}{0.100} & \multicolumn{1}{r}{0.151} & \multicolumn{1}{r}{0.112} & \multicolumn{1}{r}{0.053} & \multicolumn{1}{r}{0.087} & \multicolumn{1}{r|}{0.057} & \multicolumn{1}{r}{1.00} \\
    MRV   & \multicolumn{1}{r}{0.062} & \multicolumn{1}{r}{0.001} & \multicolumn{1}{r}{0.224} & \multicolumn{1}{r}{-0.003} & \multicolumn{1}{r}{0.136} & \multicolumn{1}{r}{0.188} & \multicolumn{1}{r}{0.051} & \multicolumn{1}{r}{0.086} & \multicolumn{1}{r|}{0.048} & \multicolumn{1}{r}{1.00} \\
    MNRV  & \multicolumn{1}{r}{-0.055} & \multicolumn{1}{r}{-0.081} & \multicolumn{1}{r}{0.072} & \multicolumn{1}{r}{-0.140} & \multicolumn{1}{r}{-0.021} & \multicolumn{1}{r}{0.028} & \multicolumn{1}{r}{0.009} & \multicolumn{1}{r}{-0.121} & \multicolumn{1}{r|}{-0.035} & \multicolumn{1}{r}{1.00} \\
    dur   & \multicolumn{1}{r}{-0.025} & \multicolumn{1}{r}{0.000} & \multicolumn{1}{r}{-0.026} & \multicolumn{1}{r}{0.016} & \multicolumn{1}{r}{0.000} & \multicolumn{1}{r}{0.033} & \multicolumn{1}{r}{0.020} & \multicolumn{1}{r}{0.021} & \multicolumn{1}{r|}{0.000} & \multicolumn{1}{r}{1.00} \\
    AMI ($T_w=2$) & \multicolumn{1}{r}{-0.130} & \multicolumn{1}{r}{-0.048} & \multicolumn{1}{r}{0.036} & \multicolumn{1}{r}{0.129} & \multicolumn{1}{r}{-0.386} & \multicolumn{1}{r}{-0.194} & \multicolumn{1}{r}{-0.008} & \multicolumn{1}{r}{0.098} & \multicolumn{1}{r|}{-0.239} & \multicolumn{1}{r}{1.00} \\
    VR ($T_w=2$) & \multicolumn{1}{r}{0.000} & \multicolumn{1}{r}{0.060} & \multicolumn{1}{r}{0.000} & \multicolumn{1}{r}{0.013} & \multicolumn{1}{r}{0.000} & \multicolumn{1}{r}{0.034} & \multicolumn{1}{r}{-0.001} & \multicolumn{1}{r}{0.016} & \multicolumn{1}{r|}{0.000} & \multicolumn{1}{r}{0.56} \\
    RV ($T_w=2$) & \multicolumn{1}{r}{-0.047} & \multicolumn{1}{r}{-0.048} & \multicolumn{1}{r}{0.036} & \multicolumn{1}{r}{-0.178} & \multicolumn{1}{r}{0.047} & \multicolumn{1}{r}{0.006} & \multicolumn{1}{r}{0.027} & \multicolumn{1}{r}{-0.005} & \multicolumn{1}{r|}{-0.070} & \multicolumn{1}{r}{1.00} \\
    MNRV2RV ($T_w=2$) & \multicolumn{1}{r}{0.000} & \multicolumn{1}{r}{0.048} & \multicolumn{1}{r}{0.000} & \multicolumn{1}{r}{0.012} & \multicolumn{1}{r}{0.000} & \multicolumn{1}{r}{0.084} & \multicolumn{1}{r}{-0.043} & \multicolumn{1}{r}{-0.066} & \multicolumn{1}{r|}{0.000} & \multicolumn{1}{r}{0.78} \\
    MRV2RV ($T_w=2$) & \multicolumn{1}{r}{0.000} & \multicolumn{1}{r}{0.035} & \multicolumn{1}{r}{0.000} & \multicolumn{1}{r}{0.010} & \multicolumn{1}{r}{0.000} & \multicolumn{1}{r}{0.007} & \multicolumn{1}{r}{-0.025} & \multicolumn{1}{r}{0.036} & \multicolumn{1}{r|}{0.000} & \multicolumn{1}{r}{0.78} \\
    AMI ($T_w=6$) & \multicolumn{1}{r}{0.005} & \multicolumn{1}{r}{-0.154} & \multicolumn{1}{r}{0.107} & \multicolumn{1}{r}{-0.059} & \multicolumn{1}{r}{-0.401} & \multicolumn{1}{r}{-0.188} & \multicolumn{1}{r}{0.000} & \multicolumn{1}{r}{0.111} & \multicolumn{1}{r|}{0.019} & \multicolumn{1}{r}{1.00} \\
    $\text{Roll}$ ($T_w=6$) & \multicolumn{1}{r}{0.043} & \multicolumn{1}{r}{0.035} & \multicolumn{1}{r}{-0.010} & \multicolumn{1}{r}{0.018} & \multicolumn{1}{r}{0.032} & \multicolumn{1}{r}{-0.045} & \multicolumn{1}{r}{-0.004} & \multicolumn{1}{r}{0.001} & \multicolumn{1}{r|}{0.001} & \multicolumn{1}{r}{1.00} \\
    $\text{Roll}^-$ ($T_w=6$) & \multicolumn{1}{r}{0.052} & \multicolumn{1}{r}{0.019} & \multicolumn{1}{r}{-0.016} & \multicolumn{1}{r}{0.045} & \multicolumn{1}{r}{0.011} & \multicolumn{1}{r}{-0.066} & \multicolumn{1}{r}{-0.011} & \multicolumn{1}{r}{0.001} & \multicolumn{1}{r|}{0.010} & \multicolumn{1}{r}{1.00} \\
    $\text{RollMod}$ ($T_w=6$) & \multicolumn{1}{r}{0.026} & \multicolumn{1}{r}{0.024} & \multicolumn{1}{r}{-0.022} & \multicolumn{1}{r}{-0.014} & \multicolumn{1}{r}{0.029} & \multicolumn{1}{r}{-0.051} & \multicolumn{1}{r}{-0.007} & \multicolumn{1}{r}{0.001} & \multicolumn{1}{r|}{0.005} & \multicolumn{1}{r}{1.00} \\
    $\text{RollMod}^-$ ($T_w=6$) & \multicolumn{1}{r}{0.024} & \multicolumn{1}{r}{0.013} & \multicolumn{1}{r}{-0.031} & \multicolumn{1}{r}{0.011} & \multicolumn{1}{r}{0.008} & \multicolumn{1}{r}{-0.076} & \multicolumn{1}{r}{-0.016} & \multicolumn{1}{r}{0.001} & \multicolumn{1}{r|}{0.006} & \multicolumn{1}{r}{1.00} \\
    TVV ($T_w=6$) & \multicolumn{1}{r}{0.015} & \multicolumn{1}{r}{0.026} & \multicolumn{1}{r}{-0.048} & \multicolumn{1}{r}{0.063} & \multicolumn{1}{r}{0.004} & \multicolumn{1}{r}{0.010} & \multicolumn{1}{r}{0.059} & \multicolumn{1}{r}{-0.021} & \multicolumn{1}{r|}{0.023} & \multicolumn{1}{r}{1.00} \\
    TQV ($T_w=6$) & \multicolumn{1}{r}{0.062} & \multicolumn{1}{r}{0.038} & \multicolumn{1}{r}{-0.004} & \multicolumn{1}{r}{0.095} & \multicolumn{1}{r}{-0.015} & \multicolumn{1}{r}{0.006} & \multicolumn{1}{r}{0.070} & \multicolumn{1}{r}{-0.049} & \multicolumn{1}{r|}{0.020} & \multicolumn{1}{r}{1.00} \\
    RTVV ($T_w=6$) & \multicolumn{1}{r}{0.132} & \multicolumn{1}{r}{0.081} & \multicolumn{1}{r}{0.080} & \multicolumn{1}{r}{0.182} & \multicolumn{1}{r}{0.080} & \multicolumn{1}{r}{0.152} & \multicolumn{1}{r}{0.093} & \multicolumn{1}{r}{0.168} & \multicolumn{1}{r|}{0.147} & \multicolumn{1}{r}{1.00} \\
    RTQV ($T_w=6$) & \multicolumn{1}{r}{0.131} & \multicolumn{1}{r}{0.081} & \multicolumn{1}{r}{0.080} & \multicolumn{1}{r}{0.182} & \multicolumn{1}{r}{0.080} & \multicolumn{1}{r}{0.152} & \multicolumn{1}{r}{0.092} & \multicolumn{1}{r}{0.168} & \multicolumn{1}{r|}{0.147} & \multicolumn{1}{r}{1.00} \\
    RAC ($T_w=6$) & \multicolumn{1}{r}{-0.042} & \multicolumn{1}{r}{-0.063} & \multicolumn{1}{r}{-0.029} & \multicolumn{1}{r}{-0.071} & \multicolumn{1}{r}{-0.020} & \multicolumn{1}{r}{-0.013} & \multicolumn{1}{r}{0.053} & \multicolumn{1}{r}{-0.078} & \multicolumn{1}{r|}{-0.101} & \multicolumn{1}{r}{1.00} \\
    VR ($T_w=6$) & \multicolumn{1}{r}{0.103} & \multicolumn{1}{r}{0.032} & \multicolumn{1}{r}{-0.010} & \multicolumn{1}{r}{0.115} & \multicolumn{1}{r}{-0.018} & \multicolumn{1}{r}{-0.020} & \multicolumn{1}{r}{0.033} & \multicolumn{1}{r}{-0.052} & \multicolumn{1}{r|}{0.043} & \multicolumn{1}{r}{1.00} \\
    RV ($T_w=6$) & \multicolumn{1}{r}{-0.116} & \multicolumn{1}{r}{-0.104} & \multicolumn{1}{r}{0.116} & \multicolumn{1}{r}{-0.176} & \multicolumn{1}{r}{-0.036} & \multicolumn{1}{r}{-0.099} & \multicolumn{1}{r}{-0.015} & \multicolumn{1}{r}{-0.114} & \multicolumn{1}{r|}{-0.128} & \multicolumn{1}{r}{1.00} \\
    MNRV2RV ($T_w=6$) & \multicolumn{1}{r}{0.011} & \multicolumn{1}{r}{0.090} & \multicolumn{1}{r}{-0.151} & \multicolumn{1}{r}{0.058} & \multicolumn{1}{r}{0.009} & \multicolumn{1}{r}{0.044} & \multicolumn{1}{r}{-0.004} & \multicolumn{1}{r}{-0.087} & \multicolumn{1}{r|}{-0.018} & \multicolumn{1}{r}{1.00} \\
    MRV2RV ($T_w=6$) & \multicolumn{1}{r}{0.097} & \multicolumn{1}{r}{0.165} & \multicolumn{1}{r}{0.031} & \multicolumn{1}{r}{0.188} & \multicolumn{1}{r}{0.231} & \multicolumn{1}{r}{0.107} & \multicolumn{1}{r}{0.068} & \multicolumn{1}{r}{0.142} & \multicolumn{1}{r|}{0.132} & \multicolumn{1}{r}{1.00} \\
    AMI ($T_w=12$) & \multicolumn{1}{r}{-0.032} & \multicolumn{1}{r}{-0.133} & \multicolumn{1}{r}{0.162} & \multicolumn{1}{r}{-0.036} & \multicolumn{1}{r}{-0.259} & \multicolumn{1}{r}{-0.030} & \multicolumn{1}{r}{-0.021} & \multicolumn{1}{r}{0.152} & \multicolumn{1}{r|}{-0.015} & \multicolumn{1}{r}{1.00} \\
    $\text{Roll}$ ($T_w=12$) & \multicolumn{1}{r}{-0.019} & \multicolumn{1}{r}{0.005} & \multicolumn{1}{r}{-0.023} & \multicolumn{1}{r}{0.000} & \multicolumn{1}{r}{0.047} & \multicolumn{1}{r}{0.009} & \multicolumn{1}{r}{0.011} & \multicolumn{1}{r}{0.002} & \multicolumn{1}{r|}{0.009} & \multicolumn{1}{r}{1.00} \\
    $\text{Roll}^-$ ($T_w=12$) & \multicolumn{1}{r}{0.074} & \multicolumn{1}{r}{0.059} & \multicolumn{1}{r}{-0.013} & \multicolumn{1}{r}{0.081} & \multicolumn{1}{r}{0.041} & \multicolumn{1}{r}{0.031} & \multicolumn{1}{r}{0.014} & \multicolumn{1}{r}{0.003} & \multicolumn{1}{r|}{0.040} & \multicolumn{1}{r}{1.00} \\
    $\text{RollMod}$ ($T_w=12$) & \multicolumn{1}{r}{-0.021} & \multicolumn{1}{r}{0.012} & \multicolumn{1}{r}{-0.026} & \multicolumn{1}{r}{-0.040} & \multicolumn{1}{r}{0.045} & \multicolumn{1}{r}{0.006} & \multicolumn{1}{r}{0.009} & \multicolumn{1}{r}{0.002} & \multicolumn{1}{r|}{0.013} & \multicolumn{1}{r}{1.00} \\
    $\text{RollMod}^-$ ($T_w=12$) & \multicolumn{1}{r}{0.026} & \multicolumn{1}{r}{0.051} & \multicolumn{1}{r}{-0.029} & \multicolumn{1}{r}{0.020} & \multicolumn{1}{r}{0.039} & \multicolumn{1}{r}{0.018} & \multicolumn{1}{r}{0.006} & \multicolumn{1}{r}{0.002} & \multicolumn{1}{r|}{0.031} & \multicolumn{1}{r}{1.00} \\
    TVV ($T_w=12$) & \multicolumn{1}{r}{-0.125} & \multicolumn{1}{r}{-0.049} & \multicolumn{1}{r}{-0.143} & \multicolumn{1}{r}{-0.051} & \multicolumn{1}{r}{-0.132} & \multicolumn{1}{r}{-0.141} & \multicolumn{1}{r}{-0.019} & \multicolumn{1}{r}{-0.285} & \multicolumn{1}{r|}{-0.069} & \multicolumn{1}{r}{1.00} \\
    TQV ($T_w=12$) & \multicolumn{1}{r}{-0.059} & \multicolumn{1}{r}{-0.031} & \multicolumn{1}{r}{-0.096} & \multicolumn{1}{r}{-0.021} & \multicolumn{1}{r}{-0.159} & \multicolumn{1}{r}{-0.150} & \multicolumn{1}{r}{-0.008} & \multicolumn{1}{r}{-0.308} & \multicolumn{1}{r|}{-0.070} & \multicolumn{1}{r}{1.00} \\
    RTVV ($T_w=12$) & \multicolumn{1}{r}{-0.008} & \multicolumn{1}{r}{0.035} & \multicolumn{1}{r}{0.010} & \multicolumn{1}{r}{-0.001} & \multicolumn{1}{r}{0.006} & \multicolumn{1}{r}{0.033} & \multicolumn{1}{r}{0.030} & \multicolumn{1}{r}{-0.028} & \multicolumn{1}{r|}{0.056} & \multicolumn{1}{r}{1.00} \\
    RTQV ($T_w=12$) & \multicolumn{1}{r}{-0.008} & \multicolumn{1}{r}{0.036} & \multicolumn{1}{r}{0.010} & \multicolumn{1}{r}{-0.002} & \multicolumn{1}{r}{0.005} & \multicolumn{1}{r}{0.032} & \multicolumn{1}{r}{0.030} & \multicolumn{1}{r}{-0.027} & \multicolumn{1}{r|}{0.056} & \multicolumn{1}{r}{1.00} \\
    RAC ($T_w=12$) & \multicolumn{1}{r}{-0.097} & \multicolumn{1}{r}{-0.007} & \multicolumn{1}{r}{-0.063} & \multicolumn{1}{r}{-0.124} & \multicolumn{1}{r}{0.036} & \multicolumn{1}{r}{0.132} & \multicolumn{1}{r}{0.019} & \multicolumn{1}{r}{-0.054} & \multicolumn{1}{r|}{0.031} & \multicolumn{1}{r}{1.00} \\
    VR ($T_w=12$) & \multicolumn{1}{r}{0.087} & \multicolumn{1}{r}{0.029} & \multicolumn{1}{r}{0.011} & \multicolumn{1}{r}{0.046} & \multicolumn{1}{r}{0.031} & \multicolumn{1}{r}{0.043} & \multicolumn{1}{r}{0.062} & \multicolumn{1}{r}{-0.003} & \multicolumn{1}{r|}{0.085} & \multicolumn{1}{r}{1.00} \\
    RV ($T_w=12$) & \multicolumn{1}{r}{-0.171} & \multicolumn{1}{r}{-0.167} & \multicolumn{1}{r}{0.028} & \multicolumn{1}{r}{-0.241} & \multicolumn{1}{r}{-0.118} & \multicolumn{1}{r}{-0.267} & \multicolumn{1}{r}{-0.074} & \multicolumn{1}{r}{-0.227} & \multicolumn{1}{r|}{-0.258} & \multicolumn{1}{r}{1.00} \\
    MNRV2RV ($T_w=12$) & \multicolumn{1}{r}{0.034} & \multicolumn{1}{r}{0.138} & \multicolumn{1}{r}{-0.050} & \multicolumn{1}{r}{0.043} & \multicolumn{1}{r}{0.161} & \multicolumn{1}{r}{0.217} & \multicolumn{1}{r}{0.050} & \multicolumn{1}{r}{0.136} & \multicolumn{1}{r|}{0.091} & \multicolumn{1}{r}{1.00} \\
    MRV2RV ($T_w=12$) & \multicolumn{1}{r}{0.163} & \multicolumn{1}{r}{0.228} & \multicolumn{1}{r}{0.170} & \multicolumn{1}{r}{0.190} & \multicolumn{1}{r}{0.376} & \multicolumn{1}{r}{0.309} & \multicolumn{1}{r}{0.131} & \multicolumn{1}{r}{0.376} & \multicolumn{1}{r|}{0.257} & \multicolumn{1}{r}{1.00} \\
    \midrule
    total number of selected variables & \multicolumn{1}{r}{39} & \multicolumn{1}{r}{42} & \multicolumn{1}{r}{39} & \multicolumn{1}{r}{42} & \multicolumn{1}{r}{41} & \multicolumn{1}{r}{42} & \multicolumn{1}{r}{42} & \multicolumn{1}{r}{42} & \multicolumn{1}{r|}{41} &  \\
    \midrule
    \midrule
    \end{tabular}%
  \label{tab:model_estimation_results_IS_10stocks_K_MLE}%
\end{table}%

\begin{table}[htbp]
  \centering
  \caption{Empirical estimates of the regression coefficients for the scale parameter obtained with the MLE.}
    \begin{tabular}{l|llllllllll}
    \toprule
    \toprule
    \diagbox{Predictors}{Stocks} & AXP  & BA  & GE   & HD  & IBM  & JNJ  & JPM  & KO   & \multicolumn{1}{l|}{XOM} & selection per stock \\
    \midrule
    TV    & \multicolumn{1}{r}{-0.033} & \multicolumn{1}{r}{-0.138} & \multicolumn{1}{r}{0.058} & \multicolumn{1}{r}{-0.064} & \multicolumn{1}{r}{0.114} & \multicolumn{1}{r}{-0.055} & \multicolumn{1}{r}{-0.075} & \multicolumn{1}{r}{0.314} & \multicolumn{1}{r|}{0.253} & \multicolumn{1}{r}{1.00} \\
    TQ    & \multicolumn{1}{r}{0.094} & \multicolumn{1}{r}{0.325} & \multicolumn{1}{r}{-0.033} & \multicolumn{1}{r}{0.100} & \multicolumn{1}{r}{0.012} & \multicolumn{1}{r}{0.090} & \multicolumn{1}{r}{0.166} & \multicolumn{1}{r}{-0.190} & \multicolumn{1}{r|}{-0.110} & \multicolumn{1}{r}{1.00} \\
    AM    & \multicolumn{1}{r}{-0.010} & \multicolumn{1}{r}{0.018} & \multicolumn{1}{r}{0.046} & \multicolumn{1}{r}{0.039} & \multicolumn{1}{r}{0.021} & \multicolumn{1}{r}{0.006} & \multicolumn{1}{r}{0.021} & \multicolumn{1}{r}{0.029} & \multicolumn{1}{r|}{0.027} & \multicolumn{1}{r}{1.00} \\
    TVV   & \multicolumn{1}{r}{-0.045} & \multicolumn{1}{r}{0.157} & \multicolumn{1}{r}{-0.127} & \multicolumn{1}{r}{-0.018} & \multicolumn{1}{r}{0.051} & \multicolumn{1}{r}{-0.035} & \multicolumn{1}{r}{-0.224} & \multicolumn{1}{r}{-0.378} & \multicolumn{1}{r|}{-0.043} & \multicolumn{1}{r}{1.00} \\
    EAM   & \multicolumn{1}{r}{-0.026} & \multicolumn{1}{r}{-0.039} & \multicolumn{1}{r}{-0.007} & \multicolumn{1}{r}{-0.055} & \multicolumn{1}{r}{-0.060} & \multicolumn{1}{r}{-0.018} & \multicolumn{1}{r}{-0.029} & \multicolumn{1}{r}{-0.008} & \multicolumn{1}{r|}{-0.056} & \multicolumn{1}{r}{1.00} \\
    TQV   & \multicolumn{1}{r}{0.045} & \multicolumn{1}{r}{-0.312} & \multicolumn{1}{r}{0.142} & \multicolumn{1}{r}{0.011} & \multicolumn{1}{r}{-0.241} & \multicolumn{1}{r}{0.048} & \multicolumn{1}{r}{0.219} & \multicolumn{1}{r}{0.279} & \multicolumn{1}{r|}{-0.148} & \multicolumn{1}{r}{1.00} \\
    MRV   & \multicolumn{1}{r}{0.028} & \multicolumn{1}{r}{0.024} & \multicolumn{1}{r}{0.021} & \multicolumn{1}{r}{0.018} & \multicolumn{1}{r}{-0.007} & \multicolumn{1}{r}{0.019} & \multicolumn{1}{r}{0.055} & \multicolumn{1}{r}{0.023} & \multicolumn{1}{r|}{-0.015} & \multicolumn{1}{r}{1.00} \\
    MNRV  & \multicolumn{1}{r}{0.099} & \multicolumn{1}{r}{0.137} & \multicolumn{1}{r}{0.060} & \multicolumn{1}{r}{0.070} & \multicolumn{1}{r}{0.199} & \multicolumn{1}{r}{0.060} & \multicolumn{1}{r}{0.061} & \multicolumn{1}{r}{0.094} & \multicolumn{1}{r|}{0.149} & \multicolumn{1}{r}{1.00} \\
    dur   & \multicolumn{1}{r}{-0.002} & \multicolumn{1}{r}{-0.031} & \multicolumn{1}{r}{0.000} & \multicolumn{1}{r}{-0.005} & \multicolumn{1}{r}{0.046} & \multicolumn{1}{r}{-0.007} & \multicolumn{1}{r}{-0.001} & \multicolumn{1}{r}{-0.005} & \multicolumn{1}{r|}{0.003} & \multicolumn{1}{r}{1.00} \\
    AMI ($T_w=2$) & \multicolumn{1}{r}{0.030} & \multicolumn{1}{r}{0.009} & \multicolumn{1}{r}{-0.012} & \multicolumn{1}{r}{0.000} & \multicolumn{1}{r}{0.045} & \multicolumn{1}{r}{0.019} & \multicolumn{1}{r}{0.005} & \multicolumn{1}{r}{-0.045} & \multicolumn{1}{r|}{0.043} & \multicolumn{1}{r}{1.00} \\
    VR ($T_w=2$) & \multicolumn{1}{r}{-0.005} & \multicolumn{1}{r}{-0.009} & \multicolumn{1}{r}{0.000} & \multicolumn{1}{r}{-0.007} & \multicolumn{1}{r}{0.014} & \multicolumn{1}{r}{0.001} & \multicolumn{1}{r}{0.020} & \multicolumn{1}{r}{-0.014} & \multicolumn{1}{r|}{-0.007} & \multicolumn{1}{r}{1.00} \\
    RV ($T_w=2$) & \multicolumn{1}{r}{-0.037} & \multicolumn{1}{r}{-0.050} & \multicolumn{1}{r}{0.011} & \multicolumn{1}{r}{-0.006} & \multicolumn{1}{r}{-0.068} & \multicolumn{1}{r}{-0.007} & \multicolumn{1}{r}{-0.071} & \multicolumn{1}{r}{0.031} & \multicolumn{1}{r|}{-0.015} & \multicolumn{1}{r}{1.00} \\
    MNRV2RV ($T_w=2$) & \multicolumn{1}{r}{0.010} & \multicolumn{1}{r}{0.001} & \multicolumn{1}{r}{0.001} & \multicolumn{1}{r}{-0.001} & \multicolumn{1}{r}{0.026} & \multicolumn{1}{r}{-0.010} & \multicolumn{1}{r}{-0.033} & \multicolumn{1}{r}{-0.015} & \multicolumn{1}{r|}{-0.007} & \multicolumn{1}{r}{1.00} \\
    MRV2RV ($T_w=2$) & \multicolumn{1}{r}{-0.024} & \multicolumn{1}{r}{0.009} & \multicolumn{1}{r}{0.001} & \multicolumn{1}{r}{0.008} & \multicolumn{1}{r}{0.023} & \multicolumn{1}{r}{0.015} & \multicolumn{1}{r}{0.017} & \multicolumn{1}{r}{0.039} & \multicolumn{1}{r|}{-0.007} & \multicolumn{1}{r}{1.00} \\
    AMI ($T_w=6$) & \multicolumn{1}{r}{0.009} & \multicolumn{1}{r}{-0.016} & \multicolumn{1}{r}{-0.009} & \multicolumn{1}{r}{-0.009} & \multicolumn{1}{r}{-0.047} & \multicolumn{1}{r}{-0.032} & \multicolumn{1}{r}{0.006} & \multicolumn{1}{r}{0.002} & \multicolumn{1}{r|}{-0.106} & \multicolumn{1}{r}{1.00} \\
    $\text{Roll}$ ($T_w=6$) & \multicolumn{1}{r}{0.069} & \multicolumn{1}{r}{0.014} & \multicolumn{1}{r}{-0.002} & \multicolumn{1}{r}{0.040} & \multicolumn{1}{r}{0.019} & \multicolumn{1}{r}{-0.087} & \multicolumn{1}{r}{0.020} & \multicolumn{1}{r}{0.050} & \multicolumn{1}{r|}{-0.001} & \multicolumn{1}{r}{1.00} \\
    $\text{Roll}^-$ ($T_w=6$) & \multicolumn{1}{r}{-0.062} & \multicolumn{1}{r}{-0.013} & \multicolumn{1}{r}{-0.046} & \multicolumn{1}{r}{-0.090} & \multicolumn{1}{r}{-0.024} & \multicolumn{1}{r}{-0.022} & \multicolumn{1}{r}{-0.066} & \multicolumn{1}{r}{0.070} & \multicolumn{1}{r|}{-0.009} & \multicolumn{1}{r}{1.00} \\
    $\text{RollMod}$ ($T_w=6$) & \multicolumn{1}{r}{-0.114} & \multicolumn{1}{r}{0.022} & \multicolumn{1}{r}{-0.072} & \multicolumn{1}{r}{-0.027} & \multicolumn{1}{r}{-0.037} & \multicolumn{1}{r}{0.019} & \multicolumn{1}{r}{-0.058} & \multicolumn{1}{r}{0.040} & \multicolumn{1}{r|}{-0.041} & \multicolumn{1}{r}{1.00} \\
    $\text{RollMod}^-$ ($T_w=6$) & \multicolumn{1}{r}{0.127} & \multicolumn{1}{r}{-0.016} & \multicolumn{1}{r}{0.139} & \multicolumn{1}{r}{0.096} & \multicolumn{1}{r}{0.049} & \multicolumn{1}{r}{0.113} & \multicolumn{1}{r}{0.125} & \multicolumn{1}{r}{0.058} & \multicolumn{1}{r|}{0.068} & \multicolumn{1}{r}{1.00} \\
    TVV ($T_w=6$) & \multicolumn{1}{r}{0.016} & \multicolumn{1}{r}{0.045} & \multicolumn{1}{r}{-0.025} & \multicolumn{1}{r}{0.070} & \multicolumn{1}{r}{-0.033} & \multicolumn{1}{r}{0.055} & \multicolumn{1}{r}{-0.003} & \multicolumn{1}{r}{0.031} & \multicolumn{1}{r|}{-0.011} & \multicolumn{1}{r}{1.00} \\
    TQV ($T_w=6$) & \multicolumn{1}{r}{-0.031} & \multicolumn{1}{r}{-0.041} & \multicolumn{1}{r}{0.024} & \multicolumn{1}{r}{-0.078} & \multicolumn{1}{r}{0.016} & \multicolumn{1}{r}{-0.074} & \multicolumn{1}{r}{-0.010} & \multicolumn{1}{r}{-0.027} & \multicolumn{1}{r|}{-0.002} & \multicolumn{1}{r}{1.00} \\
    RTVV ($T_w=6$) & \multicolumn{1}{r}{-0.011} & \multicolumn{1}{r}{-0.010} & \multicolumn{1}{r}{-0.005} & \multicolumn{1}{r}{-0.004} & \multicolumn{1}{r}{0.014} & \multicolumn{1}{r}{0.006} & \multicolumn{1}{r}{-0.016} & \multicolumn{1}{r}{-0.003} & \multicolumn{1}{r|}{-0.009} & \multicolumn{1}{r}{1.00} \\
    RTQV ($T_w=6$) & \multicolumn{1}{r}{-0.006} & \multicolumn{1}{r}{0.012} & \multicolumn{1}{r}{-0.005} & \multicolumn{1}{r}{0.004} & \multicolumn{1}{r}{0.013} & \multicolumn{1}{r}{0.008} & \multicolumn{1}{r}{0.005} & \multicolumn{1}{r}{-0.001} & \multicolumn{1}{r|}{0.000} & \multicolumn{1}{r}{1.00} \\
    RAC ($T_w=6$) & \multicolumn{1}{r}{0.005} & \multicolumn{1}{r}{-0.010} & \multicolumn{1}{r}{0.008} & \multicolumn{1}{r}{-0.013} & \multicolumn{1}{r}{-0.012} & \multicolumn{1}{r}{-0.008} & \multicolumn{1}{r}{-0.004} & \multicolumn{1}{r}{-0.011} & \multicolumn{1}{r|}{-0.002} & \multicolumn{1}{r}{1.00} \\
    VR ($T_w=6$) & \multicolumn{1}{r}{0.050} & \multicolumn{1}{r}{0.088} & \multicolumn{1}{r}{0.138} & \multicolumn{1}{r}{0.035} & \multicolumn{1}{r}{0.080} & \multicolumn{1}{r}{0.046} & \multicolumn{1}{r}{0.012} & \multicolumn{1}{r}{0.017} & \multicolumn{1}{r|}{0.038} & \multicolumn{1}{r}{1.00} \\
    RV ($T_w=6$) & \multicolumn{1}{r}{0.020} & \multicolumn{1}{r}{-0.009} & \multicolumn{1}{r}{0.026} & \multicolumn{1}{r}{0.029} & \multicolumn{1}{r}{0.006} & \multicolumn{1}{r}{0.016} & \multicolumn{1}{r}{0.012} & \multicolumn{1}{r}{-0.064} & \multicolumn{1}{r|}{0.031} & \multicolumn{1}{r}{1.00} \\
    MNRV2RV ($T_w=6$) & \multicolumn{1}{r}{-0.163} & \multicolumn{1}{r}{-0.190} & \multicolumn{1}{r}{-0.076} & \multicolumn{1}{r}{-0.080} & \multicolumn{1}{r}{-0.077} & \multicolumn{1}{r}{-0.114} & \multicolumn{1}{r}{-0.064} & \multicolumn{1}{r}{-0.112} & \multicolumn{1}{r|}{-0.243} & \multicolumn{1}{r}{1.00} \\
    MRV2RV ($T_w=6$) & \multicolumn{1}{r}{0.081} & \multicolumn{1}{r}{0.054} & \multicolumn{1}{r}{0.014} & \multicolumn{1}{r}{0.010} & \multicolumn{1}{r}{-0.031} & \multicolumn{1}{r}{0.035} & \multicolumn{1}{r}{0.025} & \multicolumn{1}{r}{0.068} & \multicolumn{1}{r|}{0.153} & \multicolumn{1}{r}{1.00} \\
    AMI ($T_w=12$) & \multicolumn{1}{r}{0.014} & \multicolumn{1}{r}{0.038} & \multicolumn{1}{r}{0.014} & \multicolumn{1}{r}{0.016} & \multicolumn{1}{r}{0.078} & \multicolumn{1}{r}{0.018} & \multicolumn{1}{r}{0.013} & \multicolumn{1}{r}{0.017} & \multicolumn{1}{r|}{0.069} & \multicolumn{1}{r}{1.00} \\
    $\text{Roll}$ ($T_w=12$) & \multicolumn{1}{r}{0.047} & \multicolumn{1}{r}{-0.052} & \multicolumn{1}{r}{0.097} & \multicolumn{1}{r}{0.041} & \multicolumn{1}{r}{-0.017} & \multicolumn{1}{r}{0.032} & \multicolumn{1}{r}{0.089} & \multicolumn{1}{r}{0.069} & \multicolumn{1}{r|}{0.010} & \multicolumn{1}{r}{1.00} \\
    $\text{Roll}^-$ ($T_w=12$) & \multicolumn{1}{r}{-0.063} & \multicolumn{1}{r}{0.033} & \multicolumn{1}{r}{-0.071} & \multicolumn{1}{r}{-0.033} & \multicolumn{1}{r}{0.016} & \multicolumn{1}{r}{0.017} & \multicolumn{1}{r}{-0.114} & \multicolumn{1}{r}{0.063} & \multicolumn{1}{r|}{-0.050} & \multicolumn{1}{r}{1.00} \\
    $\text{RollMod}$ ($T_w=12$) & \multicolumn{1}{r}{-0.068} & \multicolumn{1}{r}{0.041} & \multicolumn{1}{r}{-0.077} & \multicolumn{1}{r}{-0.058} & \multicolumn{1}{r}{-0.023} & \multicolumn{1}{r}{-0.041} & \multicolumn{1}{r}{-0.094} & \multicolumn{1}{r}{0.063} & \multicolumn{1}{r|}{-0.037} & \multicolumn{1}{r}{1.00} \\
    $\text{RollMod}^-$ ($T_w=12$) & \multicolumn{1}{r}{0.086} & \multicolumn{1}{r}{-0.024} & \multicolumn{1}{r}{0.047} & \multicolumn{1}{r}{0.051} & \multicolumn{1}{r}{0.043} & \multicolumn{1}{r}{-0.009} & \multicolumn{1}{r}{0.118} & \multicolumn{1}{r}{0.063} & \multicolumn{1}{r|}{0.086} & \multicolumn{1}{r}{1.00} \\
    TVV ($T_w=12$) & \multicolumn{1}{r}{0.010} & \multicolumn{1}{r}{-0.092} & \multicolumn{1}{r}{0.006} & \multicolumn{1}{r}{-0.037} & \multicolumn{1}{r}{0.032} & \multicolumn{1}{r}{0.100} & \multicolumn{1}{r}{0.008} & \multicolumn{1}{r}{-0.029} & \multicolumn{1}{r|}{0.001} & \multicolumn{1}{r}{1.00} \\
    TQV ($T_w=12$) & \multicolumn{1}{r}{-0.025} & \multicolumn{1}{r}{0.046} & \multicolumn{1}{r}{-0.015} & \multicolumn{1}{r}{0.031} & \multicolumn{1}{r}{-0.034} & \multicolumn{1}{r}{-0.095} & \multicolumn{1}{r}{-0.022} & \multicolumn{1}{r}{0.015} & \multicolumn{1}{r|}{-0.011} & \multicolumn{1}{r}{1.00} \\
    RTVV ($T_w=12$) & \multicolumn{1}{r}{0.009} & \multicolumn{1}{r}{0.022} & \multicolumn{1}{r}{0.005} & \multicolumn{1}{r}{-0.001} & \multicolumn{1}{r}{0.007} & \multicolumn{1}{r}{-0.007} & \multicolumn{1}{r}{-0.002} & \multicolumn{1}{r}{0.003} & \multicolumn{1}{r|}{0.013} & \multicolumn{1}{r}{1.00} \\
    RTQV ($T_w=12$) & \multicolumn{1}{r}{0.010} & \multicolumn{1}{r}{0.044} & \multicolumn{1}{r}{-0.001} & \multicolumn{1}{r}{0.008} & \multicolumn{1}{r}{-0.003} & \multicolumn{1}{r}{-0.003} & \multicolumn{1}{r}{0.016} & \multicolumn{1}{r}{0.011} & \multicolumn{1}{r|}{0.012} & \multicolumn{1}{r}{1.00} \\
    RAC ($T_w=12$) & \multicolumn{1}{r}{0.000} & \multicolumn{1}{r}{0.016} & \multicolumn{1}{r}{-0.005} & \multicolumn{1}{r}{0.008} & \multicolumn{1}{r}{0.003} & \multicolumn{1}{r}{0.012} & \multicolumn{1}{r}{0.002} & \multicolumn{1}{r}{0.013} & \multicolumn{1}{r|}{0.009} & \multicolumn{1}{r}{1.00} \\
    VR ($T_w=12$) & \multicolumn{1}{r}{-0.072} & \multicolumn{1}{r}{-0.205} & \multicolumn{1}{r}{-0.026} & \multicolumn{1}{r}{-0.029} & \multicolumn{1}{r}{-0.056} & \multicolumn{1}{r}{-0.063} & \multicolumn{1}{r}{-0.037} & \multicolumn{1}{r}{-0.022} & \multicolumn{1}{r|}{-0.070} & \multicolumn{1}{r}{1.00} \\
    RV ($T_w=12$) & \multicolumn{1}{r}{0.038} & \multicolumn{1}{r}{0.041} & \multicolumn{1}{r}{0.001} & \multicolumn{1}{r}{0.024} & \multicolumn{1}{r}{0.066} & \multicolumn{1}{r}{0.054} & \multicolumn{1}{r}{0.022} & \multicolumn{1}{r}{0.097} & \multicolumn{1}{r|}{0.042} & \multicolumn{1}{r}{1.00} \\
    MNRV2RV ($T_w=12$) & \multicolumn{1}{r}{0.183} & \multicolumn{1}{r}{0.149} & \multicolumn{1}{r}{0.083} & \multicolumn{1}{r}{0.073} & \multicolumn{1}{r}{-0.033} & \multicolumn{1}{r}{0.194} & \multicolumn{1}{r}{0.121} & \multicolumn{1}{r}{0.131} & \multicolumn{1}{r|}{0.211} & \multicolumn{1}{r}{1.00} \\
    MRV2RV ($T_w=12$) & \multicolumn{1}{r}{-0.045} & \multicolumn{1}{r}{-0.019} & \multicolumn{1}{r}{-0.007} & \multicolumn{1}{r}{0.016} & \multicolumn{1}{r}{0.091} & \multicolumn{1}{r}{-0.060} & \multicolumn{1}{r}{-0.028} & \multicolumn{1}{r}{-0.071} & \multicolumn{1}{r|}{-0.074} & \multicolumn{1}{r}{1.00} \\
    $\log(\sigma_{t-1})$ & \multicolumn{1}{r}{0.752} & \multicolumn{1}{r}{0.480} & \multicolumn{1}{r}{0.736} & \multicolumn{1}{r}{0.674} & \multicolumn{1}{r}{0.445} & \multicolumn{1}{r}{0.739} & \multicolumn{1}{r}{0.761} & \multicolumn{1}{r}{0.645} & \multicolumn{1}{r|}{0.545} & \multicolumn{1}{r}{1.00} \\
    \midrule
    total number of selected variables & \multicolumn{1}{r}{43} & \multicolumn{1}{r}{43} & \multicolumn{1}{r}{43} & \multicolumn{1}{r}{43} & \multicolumn{1}{r}{43} & \multicolumn{1}{r}{43} & \multicolumn{1}{r}{43} & \multicolumn{1}{r}{43} & \multicolumn{1}{r|}{43} &  \\
    \midrule
    \midrule
    \end{tabular}%
  \label{tab:model_estimation_results_IS_10stocks_sigma_MLE}%
\end{table}%

\subsection{Out-of-sample VaR forecast}
The coefficient estimates $\widehat{\bm{\beta}}$ obtained on the in-sample period are used to compute a one-step ahead VaR prediction in the out-of-sample period. Specifically, the VaR of each stock at a risk level $\alpha$ at time $t$ given $\bm{x}_{t-1}$ and $\widehat{\bm{\beta}}$ is obtained as 
%
\begin{equation}
   \widehat{\text{VaR}_t}(\alpha)  = \frac{\widehat{\sigma}_t(\bm{x}_{t-1},\widehat{\bm{\beta}})}{\widehat{k}_t(\bm{x}_{t-1},\widehat{\bm{\beta}})} \left(\left( 1 - \frac{\alpha - F_t(\widehat{u}_t) }{1 - F_t(\widehat{u}_t)}\right)^{-\widehat{k}_t(\bm{x}_{t-1},\widehat{\bm{\beta}})} - 1\right) + \widehat{u}_t.
\end{equation}
where $F_t(\widehat{u}_t)$ is the probability of exceeding the threshold $\widehat{u}_t$ and is fixed to 90\%.
The coverage rate of $\{\widehat{\text{VaR}_t}(\alpha)\}_{t=T_{is} + 1}^{T_{is} + T_{os}}$ for the out-of sample period is obtained as follows,
\begin{equation}
    \text{Coverage Rate} = \frac{ \sum_{t=T_{is} + 1}^{T_{is} + T_{os}} \mathbbm{1}\{ l_t \leq  \widehat{\text{VaR}_t}(\alpha) \}}{T_{os}}.
\end{equation}
Table~\ref{tab:VaR_coverage_rates} shows the coverage rate of $\{\widehat{\text{VaR}_t}(\alpha)\}_{t=T_{is} + 1}^{T_{is} + T_{os}}$ at the risk level $\alpha$ for various $\alpha\in[90\%,100\%)$. 
We resort to the Kolmogorov–Smirnov (K-S) test to test the goodness of fit of the predicted GPD over the out-of-sample period, i.e., we test whether $\{\widehat{F}(y_t | y_t >0 ) = \text{GPD}(y_t ;\bm{x}_{t-1},\widehat{\bm{\beta}})\}$ follows a standard uniform distribution. 
The p-values of the K-S tests in Table~\ref{tab:VaR_coverage_rates} indicate that we reject the regression model on three stocks out of nine at the 1\% significance level. 

\begin{table}[htbp]
  \centering
  \caption{Out-of-sample VaR Coverage Rates and p-values for the K-S Test.}
  \resizebox{0.9\textwidth}{!}{%
    \begin{tabular}{lrrrrrrrrrrrr|r}
    \toprule
    \toprule
     \diagbox{Stock Names}{VaR risk level} & 0.9 & 0.91  & 0.92  & 0.93  & 0.94  & 0.95  & 0.96  & 0.97  & 0.98  & 0.99  & 0.999 & 0.9999 & \multicolumn{1}{l}{K-S Test p-values} \\
    \midrule
    AXP   & 0.8999 & 0.9111 & 0.9224 & 0.9314 & 0.9414 & 0.9520 & 0.9631 & 0.9729 & 0.9827 & 0.9913 & 0.9986 & 0.9996 & 0.208 \\
    BA    & 0.9000 & 0.9098 & 0.9210 & 0.9309 & 0.9417 & 0.9522 & 0.9630 & 0.9738 & 0.9836 & 0.9930 & 0.9985 & 0.9996 & 0.0431 \\
    GE    & 0.8999 & 0.9171 & 0.9254 & 0.9376 & 0.9475 & 0.9622 & 0.9709 & 0.9811 & 0.9889 & 0.9943 & 0.9989 & 0.9999 & $2.68\times10^{-12}$ \\
    HD    & 0.9001 & 0.9097 & 0.9217 & 0.9339 & 0.9453 & 0.9551 & 0.9652 & 0.9756 & 0.9840 & 0.9934 & 0.9990 & 0.9998 & 0.0023 \\
    IBM   & 0.9003 & 0.9094 & 0.9204 & 0.9312 & 0.9414 & 0.9509 & 0.9614 & 0.9702 & 0.9809 & 0.9905 & 0.9980 & 0.9993 & 0.8954 \\
    JNJ   & 0.9000 & 0.9083 & 0.9195 & 0.9284 & 0.9382 & 0.9480 & 0.9571 & 0.9682 & 0.9786 & 0.9888 & 0.9977 & 0.9995 & 0.3322 \\
    JPM   & 0.9000 & 0.9094 & 0.9176 & 0.9294 & 0.9385 & 0.9500 & 0.9612 & 0.9711 & 0.9815 & 0.9904 & 0.9982 & 0.9997 & 0.3993 \\
    KO    & 0.8996 & 0.9109 & 0.9180 & 0.9273 & 0.9380 & 0.9489 & 0.9594 & 0.9709 & 0.9813 & 0.9909 & 0.9980 & 0.9994 & 0.3786 \\
    XOM   & 0.8988 & 0.9056 & 0.9136 & 0.9218 & 0.9303 & 0.9396 & 0.9502 & 0.9630 & 0.9742 & 0.9863 & 0.9972 & 0.9992 & $6.32\times10^{-8}$ \\
    \bottomrule
    \bottomrule
    \end{tabular}%
    }
  \label{tab:VaR_coverage_rates}%
\end{table}%

\section{Conclusion}
\label{sec:conclusion}
This paper proposes a novel extreme value regression framework to study the dynamics of high-frequency tail risk. The proposed model allows for both stationary and local unit-root predictors to capture the persistence of high-frequency extreme losses. We propose a two-step regularized approach to perform automatic variable selection, and establish the oracle property of the corresponding ALMLE in selecting stationary and local unit-root predictors. We use the proposed approach to investigate the predictive content of 42 liquidity and volatility indicators on the distribution of extreme losses for nine large liquid U.S. stocks. Our variable selection procedure reveals that the severity of tail risk is strongly associated to low price impact in periods of high volatility of liquidity and volatility of volatility. These findings can contribute to timely alert high-frequency traders of rising risk levels and facilitate improvements of their algorithmic trading practices for financial risk management. Moreover, it provides incentives for market markers to absorb liquidity demand in periods of instability. Finally, it suggests a set of predictors to regulators investigating trading activities that can help defining proper regulation guidelines and safeguard financial stability.

\bibliographystyle{apalike}
\bibliography{biblio}

\begin{thebibliography}{}

\bibitem[Balkema and De~Haan, 1974]{balkema1974residual}
Balkema, A.~A. and De~Haan, L. (1974).
\newblock Residual life time at great age.
\newblock {\em Annals of Probability}, 2(5):792--804.

\bibitem[Barndorff-Nielsen et~al., 2009]{barndorff2009realized}
Barndorff-Nielsen, O.~E., Hansen, P.~R., Lunde, A., and Shephard, N. (2009).
\newblock Realized kernels in practice: Trades and quotes.
\newblock {\em Econometrics Journal}, 12:C1--C32.

\bibitem[Bee et~al., 2019]{bee2019realized}
Bee, M., Dupuis, D.~J., and Trapin, L. (2019).
\newblock Realized peaks over threshold: A time-varying extreme value approach
  with high-frequency-based measures.
\newblock {\em Journal of Financial Econometrics}, 17(2):254--283.

\bibitem[Billingsley, 2013]{billingsley2013convergence}
Billingsley, P. (2013).
\newblock {\em Convergence of probability measures}.
\newblock John Wiley \& Sons.

\bibitem[Brogaard et~al., 2018]{brogaard2018}
Brogaard, J., Carrion, A., Moyaert, T., Riordan, R., Shkilko, A., and Sokolov,
  K. (2018).
\newblock High frequency trading and extreme price movements.
\newblock {\em Journal of Financial Economics}, 128(2):253--265.

\bibitem[Brownlees and Gallo, 2006]{brownlees2006financial}
Brownlees, C.~T. and Gallo, G.~M. (2006).
\newblock Financial econometric analysis at ultra-high frequency: Data handling
  concerns.
\newblock {\em Computational Statistics \& Data Analysis}, 51(4):2232--2245.

\bibitem[Chavez-Demoulin and Davison, 2012]{chavez2012modelling}
Chavez-Demoulin, V. and Davison, A. (2012).
\newblock Modelling time series extremes.
\newblock {\em REVSTAT}, 10(1):109--133.

\bibitem[Chavez-Demoulin et~al., 2016]{chavez2016extreme}
Chavez-Demoulin, V., Embrechts, P., and Hofert, M. (2016).
\newblock An extreme value approach for modeling operational risk losses
  depending on covariates.
\newblock {\em Journal of Risk and Insurance}, 83(3):735--776.

\bibitem[Chavez-Demoulin et~al., 2014]{chavez2014extreme}
Chavez-Demoulin, V., Embrechts, P., and Sardy, S. (2014).
\newblock Extreme-quantile tracking for financial time series.
\newblock {\em Journal of Econometrics}, 181(1):44--52.

\bibitem[Coles, 2001]{coles2001introduction}
Coles, S. (2001).
\newblock {\em An introduction to statistical modeling of extreme values}.
\newblock Springer.

\bibitem[Cs{\"o}rg{\H{o}}, 1968]{csorgHo1968strong}
Cs{\"o}rg{\H{o}}, M. (1968).
\newblock On the strong law of large numbers and the central limit theorem for
  martingales.
\newblock {\em Transactions of the American Mathematical Society},
  131(1):259--275.

\bibitem[Dionne et~al., 2009]{dionne2009intraday}
Dionne, G., Duchesne, P., and Pacurar, M. (2009).
\newblock Intraday value at risk ({IVaR}) using tick-by-tick data with
  application to the {T}oronto {S}tock {E}xchange.
\newblock {\em Journal of Empirical Finance}, 16(5):777--792.

\bibitem[Engle, 2001]{engle2001garch}
Engle, R. (2001).
\newblock {GARCH} 101: The use of {ARCH/GARCH} models in applied econometrics.
\newblock {\em Journal of Economic Perspectives}, 15(4):157--168.

\bibitem[Giot, 2005]{giot2005market}
Giot, P. (2005).
\newblock Market risk models for intraday data.
\newblock {\em The European Journal of Finance}, 11(4):309--324.

\bibitem[Goyenko et~al., 2009]{goyenko2009liquidity}
Goyenko, R.~Y., Holden, C.~W., and Trzcinka, C.~A. (2009).
\newblock Do liquidity measures measure liquidity?
\newblock {\em Journal of Financial Economics}, 92(2):153--181.

\bibitem[Grossman and Miller, 1988]{grossman1988liquidity}
Grossman, S.~J. and Miller, M.~H. (1988).
\newblock Liquidity and market structure.
\newblock {\em The Journal of Finance}, 43(3):617--633.

\bibitem[Hambuckers et~al., 2018]{hambuckers2018understanding}
Hambuckers, J., Groll, A., and Kneib, T. (2018).
\newblock Understanding the economic determinants of the severity of
  operational losses: A regularized generalized pareto regression approach.
\newblock {\em Journal of Applied Econometrics}, 33(6):898--935.

\bibitem[Hasbrouck and Schwartz, 1988]{hasbrouck1988liquidity}
Hasbrouck, J. and Schwartz, R.~A. (1988).
\newblock Liquidity and execution costs in equity markets.
\newblock {\em Journal of Portfolio Management}, 14(3):10.

\bibitem[Hendershott et~al., 2011]{hendershott2011}
Hendershott, T., Jones, C.~M., and Menkveld, A. (2011).
\newblock Does algorithmic trading improve liquidity?
\newblock {\em The Journal of Finance}, 66(1):1--33.

\bibitem[Hosking and Wallis, 1987]{hosking1987parameter}
Hosking, J.~R. and Wallis, J.~R. (1987).
\newblock Parameter and quantile estimation for the generalized {P}areto
  distribution.
\newblock {\em Technometrics}, 29(3):339--349.

\bibitem[James et~al., 2013]{james2013introduction}
James, G., Witten, D., Hastie, T., and Tibshirani, R. (2013).
\newblock {\em An introduction to statistical learning}.
\newblock Springer.

\bibitem[Kirilenko et~al., 2017]{kirilenko2017flash}
Kirilenko, A., Kyle, A.~S., Samadi, M., and Tuzun, T. (2017).
\newblock The flash crash: High-frequency trading in an electronic market.
\newblock {\em The Journal of Finance}, 72(3):967--998.

\bibitem[Kock, 2016]{kock2016consistent}
Kock, A.~B. (2016).
\newblock Consistent and conservative model selection with the adaptive lasso
  in stationary and nonstationary autoregressions.
\newblock {\em Econometric Theory}, 32(1):243--259.

\bibitem[Lee, 2016]{lee2016predictive}
Lee, J.~H. (2016).
\newblock Predictive quantile regression with persistent covariates: {IVX-QR}
  approach.
\newblock {\em Journal of Econometrics}, 192(1):105--118.

\bibitem[Lee et~al., 2022]{lee2021LASSO}
Lee, J.~H., Shi, Z., and Gao, Z. (2022).
\newblock On lasso for predictive regression.
\newblock {\em Journal of Econometrics}, 229(2):322--349.

\bibitem[Massacci, 2017]{massacci2017tail}
Massacci, D. (2017).
\newblock Tail risk dynamics in stock returns: Links to the macroeconomy and
  global markets connectedness.
\newblock {\em Management Science}, 63(9):3072--3089.

\bibitem[Medeiros and Mendes, 2016]{medeiros2016}
Medeiros, M.~C. and Mendes, E.~F. (2016).
\newblock $\mathcal{L}_1$-regularization of high-dimensional time-series models
  with non-{G}aussian and heteroskedastic errors.
\newblock {\em Journal of Econometrics}, 191(1):255--271.

\bibitem[Nieto and Ruiz, 2016]{nieto2016frontiers}
Nieto, M.~R. and Ruiz, E. (2016).
\newblock Frontiers in {VaR} forecasting and backtesting.
\newblock {\em International Journal of Forecasting}, 32(2):475--501.

\bibitem[Phillips, 1991]{phillips1991optimal}
Phillips, P.~C. (1991).
\newblock Optimal inference in cointegrated systems.
\newblock {\em Econometrica: Journal of the Econometric Society},
  59(2):283--306.

\bibitem[Phillips and Durlauf, 1986]{phillips1986multiple}
Phillips, P.~C. and Durlauf, S.~N. (1986).
\newblock Multiple time series regression with integrated processes.
\newblock {\em The Review of Economic Studies}, 53(4):473--495.

\bibitem[Phillips and Lee, 2013]{phillips2013predictive}
Phillips, P.~C. and Lee, J.~H. (2013).
\newblock Predictive regression under various degrees of persistence and robust
  long-horizon regression.
\newblock {\em Journal of Econometrics}, 177(2):250--264.

\bibitem[Pickands, 1975]{pickands1975statistical}
Pickands, J. (1975).
\newblock Statistical inference using extreme order statistics.
\newblock {\em Annals of Statistics}, 3(1):119--131.

\bibitem[Roll, 1984]{roll1984simple}
Roll, R. (1984).
\newblock A simple implicit measure of the effective bid-ask spread in an
  efficient market.
\newblock {\em The Journal of Finance}, 39(4):1127--1139.

\bibitem[Saikkonen, 1993]{saikkonen1993continuous}
Saikkonen, P. (1993).
\newblock Continuous weak convergence and stochastic equicontinuity results for
  integrated processes with an application to the estimation of a regression
  model.
\newblock {\em Econometric Theory}, 9(2):155--188.

\bibitem[Saikkonen, 1995]{saikkonen1995problems}
Saikkonen, P. (1995).
\newblock Problems with the asymptotic theory of maximum likelihood estimation
  in integrated and cointegrated systems.
\newblock {\em Econometric Theory}, 11(5):888--911.

\bibitem[Schwaab et~al., 2021]{schwaab2021modeling}
Schwaab, B., Lucas, A., and Zhang, X. (2021).
\newblock Modeling extreme events: time-varying extreme tail shape.
\newblock \url{https://doi.org/10.2866/252648}.
\newblock ECB Working Paper, No. 2524, ISBN 978-92-899-4524-0, European Central
  Bank (ECB), Frankfurt a. M.

\bibitem[Smith, 1985]{smith1985maximum}
Smith, R.~L. (1985).
\newblock Maximum likelihood estimation in a class of nonregular cases.
\newblock {\em Biometrika}, 72(1):67--90.

\bibitem[Tibshirani, 1996]{tibshirani1996regression}
Tibshirani, R. (1996).
\newblock Regression shrinkage and selection via the lasso.
\newblock {\em Journal of the Royal Statistical Society: Series B},
  58(1):267--288.

\bibitem[Tibshirani et~al., 2012]{tibshirani2012strong}
Tibshirani, R., Bien, J., Friedman, J., Hastie, T., Simon, N., Taylor, J., and
  Tibshirani, R.~J. (2012).
\newblock Strong rules for discarding predictors in lasso-type problems.
\newblock {\em Journal of the Royal Statistical Society: Series B},
  74(2):245--266.

\bibitem[Weller, 2017]{weller2018}
Weller, B.~M. (2017).
\newblock Does algorithmic trading reduce information acquisition?
\newblock {\em The Review of Financial Studies}, 31(6):2184--2226.

\bibitem[Zou, 2006]{zou2006adaptive}
Zou, H. (2006).
\newblock The adaptive lasso and its oracle properties.
\newblock {\em Journal of the American Statistical Association},
  101(476):1418--1429.

\end{thebibliography}
\newpage
\appendix
\section*{Appendices}
\section{Proofs}
\label{appendix:proofs}
\subsection{MLE - Gradient and Hessian matrix of the loglikelihood function $\mathcal{L}(\cdot)$}
\label{append:MLE}
\label{append:MLE_gradien_hessian}
The gradient function of  $\mathcal{L}(\bm{\beta}; \{ y_t \}, \{ \bm{z}_{t-1} \})$ is given by
\begin{align}
  & \frac{\partial \mathcal{L}(\bm{\beta}; \{ y_t \}, \{ \bm{z}_{t-1} \})}{\partial \bm{\beta}} 
  = 
  \sum_{t = 1}^T \bm{\psi}_t (\bm{\beta}),
  \\
  & \bm{\psi}_t (\bm{\beta}) 
  := \mathbbm{1}\{y_t>0\} \left( \,
   \bm{X}_{t-1} \begin{bmatrix}
   g_{k,t}(\bm{\beta})
  \\
  \vdots
  \\
  g_{k,t}(\bm{\beta})
  \\
 g_{\sigma,t}(\bm{\beta})
  \\
  g_{\sigma,t}(\bm{\beta})
  \\
  \vdots
  \\
   g_{\sigma,t}(\bm{\beta})
  \\
 g_{\sigma,t}(\bm{\beta})
  \end{bmatrix}
  +
  g_{\sigma,t}(\bm{\beta}) \sum_{i = 0}^{t-1}\begin{bmatrix}
  0
  \\
  \vdots
  \\
  0
  \\
 \beta_{2,p+1}^{i}
 \\
 \beta_{2,p+1}^{i} z_{1,t-1-i} 
  \\
  \vdots
  \\
 \beta_{2,p+1}^{i} z_{p,t-1-i} 
 \\
 \sum_{j=1}^{p}\beta_{2,j}\, i\,\beta_{2,p+1}^{i-1}\,z_{j,t-1-i}
  \end{bmatrix} 
  \right), 
 \end{align}
\begin{equation}
\begin{aligned}
  & \bm{X}_t := \text{diag}(\bm{x}_t),
  \\
  & \bm{x}_t : =[1, \bm{z}_t', 1, \bm{z}_t', \log(\sigma_t(\bm{\beta}))]' , 
\\
& g_{k,t}(\bm{\beta}) \resizebox{.9\hsize}{!}{$ := 
    \left( \frac{1}{k_t(\bm{\beta})}  - 2\right)\log\left(1 + k_t(\bm{\beta})\frac{y_t}{\sigma_t(\bm{\beta})}\right)  + \left( \frac{1}{k_t(\bm{\beta})} - 1 - 2 k_t(\bm{\beta})\right)\left(\frac{1}{1 + k_t(\bm{\beta})\frac{y_t}{\sigma_t(\bm{\beta})}} - 1 \right)$},
    \\
& g_{\sigma,t}(\bm{\beta})  :=   
    \frac{1}{ k_t(\bm{\beta})}\left( 1 - \frac{1}{1 + k_t(\bm{\beta})\frac{y_t}{\sigma_t(\bm{\beta})}}\right) - \frac{1}{1 + k_t(\bm{\beta})\frac{y_t}{\sigma_t(\bm{\beta})}},
    \\
& g_{A,i,t} := \resizebox{.9\hsize}{!}{$ \left[ 0, \ldots, 0 , \beta_{2,p+1}^{i}, \beta_{2,p+1}^{i} z_{1,t-1-i}, \ldots, \beta_{2,p+1}^{i} z_{p,t-1-i},
 \sum_{j=1}^{p}\beta_{2,j}\, i\,\beta_{2,p+1}^{i-1}\,z_{j,t-1-i}\right]', $}
\end{aligned}
  \label{eq:gradient_loglik}
\end{equation}
where $\text{diag}(\bm{x}_t)$ denotes the square diagonal matrix with the elements of vector $\bm{x}_t$ on the main diagonal.  
Using the gradient information in ~\eqref{eq:gradient_loglik}, we give the Hessian matrix of $\mathcal{L}(\bm{\beta}; \{ y_t \}, \{ \bm{z}_{t-1} \})$ below.
\begin{equation}
\label{eq:hessian_loglik}
\begin{aligned}
H_{\mathcal{L}}(\bm{\beta})     & := 
        \frac{\partial^2 \mathcal{L}(\bm{\beta}; \{ y_t \}, \{ \bm{z}_{t-1} \})}{\partial \bm{\beta}\partial \bm{\beta}' }
         = 
         \sum_{t = 1}^T \frac{\partial {\psi}_t (\bm{\beta})}{\partial \bm{\beta}'}
        \\
\frac{\partial {\psi}_t (\bm{\beta})}{\partial \bm{\beta}'}
        & = 
    \mathbbm{1}\{y_t>0\} \left( \,
    \bm{X}_{t-1} H_g
    + 
 \sum_{i = 0}^{t-1} \text{diag}\left( g_{A,i,t} \right) \left( \bm{I}_{(2p+3)\times 1}\bigotimes \frac{\partial g_{\sigma,t}(\bm{\beta})}{\partial \bm{\beta}'} \right)
   + 
   g_{\sigma,t}(\bm{\beta}) H_A \right),  
  \\
  H_g 
  & :=
  \begin{bmatrix}
   \frac{\partial g_{k,t}(\bm{\beta})}{\partial \bm{\beta}'}
  \\
  \vdots
  \\
  \frac{\partial g_{k,t}(\bm{\beta})}{\partial \bm{\beta}'}
  \\
 \frac{\partial g_{\sigma,t}(\bm{\beta})}{\partial \bm{\beta}'}
  \\
  \vdots
  \\
  \frac{\partial g_{\sigma,t}(\bm{\beta})}{\partial \bm{\beta}'}
  \end{bmatrix},
   \qquad 
   H_A
   := \resizebox{.6\hsize}{!}{$
   \sum_{i = 0}^{t-1} \begin{bmatrix}
  0 & \dots & 0 & 0 & 0 & \dots & 0 & 0
  \\
 0 & \dots & 0 & \vdots &  & \ddots & \vdots& \vdots
  \\
 0 & \dots & 0 & 0 & 0 & \dots & 0 & 0
  \\
 0 & \dots & 0 & 0 & 0 & \dots & 0 & i\,\beta_{2,p+1}^{i-1}
 \\
 0 & \dots & 0 &0 & 0 & \dots & 0 & i\,\beta_{2,p+1}^{i-1} z_{1,t-1-i} 
  \\
 0 & \dots & 0 & \vdots &  & \ddots & \vdots & \vdots
  \\
 0 & \dots & 0 & 0 & 0 & \dots & 0 &  i\,\beta_{2,p+1}^{i-1} z_{p,t-1-i} 
 \\
0 & \dots & 0 & i\,\beta_{2,p+1}^{i-1} & i\,\beta_{2,p+1}^{i-1}\,z_{1,t} & \dots & i\,\beta_{2,p+1}^{i-1}\,z_{p,t} & \sum_{j=1}^{p}\beta_{2,j}\, i\,(i-1)\,\beta_{2,p+1}^{i-2}\,z_{j,t-i-i}
  \end{bmatrix} 
  $}
  \end{aligned}
\end{equation}

\begin{equation}
\begin{aligned}
  \frac{\partial g_{k,t}(\bm{\beta})}{\partial \bm{\beta}'}
  & = \,\begin{bmatrix}
   \frac{\partial g_{k,t}(\bm{\beta})}{\partial \text{logit}(k_t)}
   & 
   \dots
   &
   \frac{\partial g_{k,t}(\bm{\beta})}{\partial \text{logit}(k_t)}
   &
   \frac{\partial g_{k,t}(\bm{\beta})}{\partial \log(\sigma_t)}
   &
   \dots
   &
   \frac{\partial g_{k,t}(\bm{\beta})}{\partial \log(\sigma_t)}
   \end{bmatrix}\,\text{diag}(\bm{X}_{t-1})
     \\
  \frac{\partial g_{\sigma,t}(\bm{\beta})}{\partial \bm{\beta}'}
  & = \,\begin{bmatrix}
   \frac{\partial g_{\sigma,t}(\bm{\beta})}{\partial \text{logit}(k_t)}
   & 
   \dots
   &
   \frac{\partial g_{\sigma,t}(\bm{\beta})}{\partial \text{logit}(k_t)}
   &
   \frac{\partial g_{\sigma,t}(\bm{\beta})}{\partial \log(\sigma_t)}
   &
   \dots
   &
   \frac{\partial g_{\sigma,t}(\bm{\beta})}{\partial \log(\sigma_t)}
   \end{bmatrix}\,\text{diag}(\bm{X}_{t-1})
  \end{aligned}
\end{equation}

\begin{equation}
\label{eq:H_k_sig}
\begin{aligned}
   \frac{\partial g_{k,t}(\bm{\beta})}{\partial \text{logit}(k_t)}
   & =  (2 - \frac{1}{k_t})\log(1 + k_t\frac{y_t}{\sigma_t})
   +
   (1 + k_t\frac{y_t}{\sigma_t})^{-1}( (1 - 2\,k_t)^2 - (\frac{1}{k_t} - 2\,k_t)(1- 2\,k_t) )
   \\
   & + 
   (1 + k_t\frac{y_t}{\sigma_t})^{-2}\frac{y_t}{\sigma_t}(2k_t^2 + k_t-1)
   \\
    \frac{\partial g_{k,t}(\bm{\beta})}{\partial \log(\sigma_t)}
   & = 
   k_t\frac{y_t}{\sigma_t}( (1-\frac{1}{k_t})(1 + k_t\frac{y_t}{\sigma_t})^{-1} + (\frac{1}{k_t} -1 -2\,k_t) (1 + k_t\frac{y_t}{\sigma_t})^{-2} )
   \\
   \frac{\partial g_{\sigma,t}(\bm{\beta})}{\partial \text{logit}(k_t)}
   & = (1 + k_t\frac{y_t}{\sigma_t})^{-2} \frac{y_t}{\sigma_t}(1+k_t)(1 - 2\,k_t)
   + 
   ((1 + k_t\frac{y_t}{\sigma_t})^{-1} -1)(1-2\,k_t)\frac{1}{k_t}
   \\
   \frac{\partial g_{\sigma,t}(\bm{\beta})}{\partial \log(\sigma_t)}
   & = -\frac{y_t}{\sigma_t} (1 + k_t\frac{y_t}{\sigma_t})^{-2}(1+k_t),
\end{aligned}
\end{equation}
where $\bigotimes$ denotes the Kronecker product operator, and we suppress the coefficients in $k_t:=k_t(\bm{\beta)}$, $\sigma_t:=\sigma_t(\bm{\beta)}$ and the conditional variables in $ {\psi}_t(\bm{\beta}) := {\psi} (\bm{\beta}; L_{t-1}, \bm{Z}_{t-1}, L_{t-2}, \bm{Z}_{t-2}, \ldots) $ for ease of the notations. We also denote specifically that $k_t^o:=k_t(\bm{\beta^o})$ and $\sigma_t^o:=\sigma_t(\bm{\beta}^o)$ to ease the notation.

\subsection{MLE - Proofs}
\label{append:MLE_proofs}
\begin{proposition}
\label{prop:assumpB2_continuity}
Under Assumptions~\ref{assump:prob_space}, \ref{assump:correct_model_speci}, \ref{assump:linear_processes_Zt}, \ref{assump:limiting_Zt}(1) and \ref{assump:narrow_Theta}, for any $\varepsilon >0$ and any interior $ \widetilde{\bm{\beta}}$ in $\Theta$, it holds that there exists an $\delta >0$ such that
\begin{equation}
 \sup_{ \norm{\bm{\beta} -  \widetilde{\bm{\beta}}} < \delta} \norm{\bm{\psi}_t(\bm{\beta}) - \bm{\psi}_t( \widetilde{\bm{\beta}} )} <\varepsilon,
\end{equation}
where $\norm{\cdot}$ is the Euclidean norm.
\end{proposition}
\begin{proof}
This proposition claims that $\bm{\psi}_t(\cdot)$ is uniformly continuous in $\bm{\beta}\in\Theta$. To prove this proposition, we show that $\bm{\psi}_t(\cdot)$ is continuous and tight in $\bm{\beta}\in\Theta$ below.
\\
First, viewing the formula of $\bm{\psi}_t(\cdot)$ in Appendix~\ref{append:MLE_gradien_hessian}, we know that $\bm{\psi}_t(\cdot)$ is a composition of continuous functions in $\bm{\beta}$ and hence is continuous in $\bm{\beta}$.
Secondly, by Assumptions~\ref{assump:prob_space}, \ref{assump:linear_processes_Zt}, \ref{assump:limiting_Zt}(1) and~\ref{assump:narrow_Theta}, we know that $\bm{\beta}$ and $\bm{z}_{t-1}$ are bounded in probability and thereby $k_t(\cdot)$ and $\sigma_t(\cdot)$ are bounded and their lower bounds are above zero in probability. It follows that we get that $g_{\sigma,t}(\cdot)$ and $g_{A,i,t}(\cdot)$ are also bounded in probability. Additionally, under Assumption~\ref{assump:correct_model_speci}, we know that $|y_t|$ is bounded in probability. By Jensen's inequality, we obtain that $g_{k,t}(\cdot)$ is bounded in probability. Therefore, we obtain that $\bm{\psi}_t(\cdot)$ is bounded in probability in $\bm{\beta}\in\Theta$ and thereby accomplish this proof. 
\end{proof}
\begin{proposition}
\label{prop:mean_of_score}
Under Assumptions~\ref{assump:prob_space}, \ref{assump:correct_model_speci}, \ref{assump:linear_processes_Zt}, \ref{assump:limiting_Zt}(1) and \ref{assump:narrow_Theta}, it holds that $\{\mathbb{E}[\bm{\psi}_t(\bm{\beta})], \bm{\beta}\in\Theta\}$ has unique zero at $\bm{\beta} = \bm{\beta}^{o*}$ and for any $\epsilon >0$,
\begin{equation}
\label{eq:score_truecoefs_convergence_inprob}
    \lim_{T\rightarrow\infty} P\left\{\norm{\frac{1}{T}\,\sum_{t=1}^{T}\; \bm{\psi}_t(\bm{\beta}^{o*}) }>\epsilon \right\} = 0,  
\end{equation} 
\end{proposition}
\begin{proof}
First, the equality~\eqref{eq:score_truecoefs_convergence_inprob} claims that $\frac{1}{T}\,\sum_{t=1}^{T}\; \bm{\psi}_t(\bm{\beta}^{o*})$ converges to zero in probability. Let us prove the equality~\eqref{eq:score_truecoefs_convergence_inprob} using the conditions in Theorem~1 of~\cite{csorgHo1968strong}.
\\
Under Assumptions~\ref{assump:prob_space} and \ref{assump:correct_model_speci}, we take the expectation with respect to the true conditional probability function of $y_t$ and obtain that 
\begin{equation}
    \mathbb{E}\left[ \bm{\psi}_t(\bm{\beta}^{o*})\middle| \mathcal{F}_{t-1} \right] = \bm{0}, \qquad \text{for}\; t=1,\ldots, T.
\end{equation}
We also have that 
\begin{equation}
\begin{aligned}
    \,\sum_{t=1}^{T}\; \frac{1}{t^2}\mathbb{E}\left[\bm{\psi}_t(\bm{\beta}^{o*}) \right]
    & \leq \sum_{t=1}^{T}\; \frac{1}{t^2}\, M\,\bm{1}_{2p+3}
    \\
    & < \infty,
\end{aligned}
\end{equation}
where $M\in\mathbb{R}$ is a large finite number and $\bm{1}_{2p+3}$ denotes a vector of $(2p+3)$ ones; the second last inequality is obtained by the tightness of $\bm{\psi}_t(\cdot)$ in Proposition~\ref{prop:assumpB2_continuity}; the last inequality is obtained by $\sum_{t=1}^{T}\; \frac{1}{t^2}<\infty$.
Therefore, by applying Theorem~1 of~\cite{csorgHo1968strong} we conclude that $\frac{1}{T}\,\sum_{t=1}^{T}\; \bm{\psi}_t(\bm{\beta}^{o*})$ converges to zero in probability.
\\
Secondly, we are going to show the uniqueness of $\bm{\beta}^{o*}$ in $\Theta$ such that $\mathbb{E}[\bm{\psi}_t(\bm{\beta}^{o*})] = 0$ by contradiction. Suppose there is a $\bm{\beta}\in\Theta$ and $\bm{\beta}\neq\bm{\beta}^{o*}$ such that 
\begin{equation}
  \mathbb{E}[\bm{\psi}_t(\bm{\beta})] = 0 .
  \label{eq:contradict_zero_beta}
\end{equation}
On the other hand, by the mean value theorem it holds that 
\begin{equation}
\bm{\psi}_t(\bm{\beta}) - \bm{\psi}_t(\bm{\beta}^{o*})
    =
    H_{\mathcal{L}}(\overline{\bm{\beta}})\left( \bm{\beta} - \bm{\beta}^{o*}\right), 
\end{equation}
where $\overline{\bm{\beta}}$ is between $\bm{\beta}$ and $\bm{\beta}^{o*}$. Since $H_{\mathcal{L}}(\cdot)$ is positive definite in $\Theta$ almost surely by Assumption~\ref{assump:narrow_Theta} and $\bm{\beta} \neq \bm{\beta}^{o*}$, we obtain that $\bm{\psi}_t(\bm{\beta}) \neq \bm{\psi}_t(\bm{\beta}^{o*})$ almost surely and contradict~\eqref{eq:contradict_zero_beta}. Hence, we conclude the uniqueness of $\bm{\beta}^{o*}$ in $\Theta$.
\end{proof}


\subsubsection{Proof of Theorem~\ref{thm:MLE_consistency}}
\begin{proof}
First, let us prove that $\widehat{\bm{\beta}}^{\text{mle}}$ is in $\Theta$.
By Proposition~\ref{prop:mean_of_score}, we can get that for any $\bm{\epsilon}\in\mathbb{R}^{2p+3}$ and $\bm{\epsilon} > \bm{0}$
, there exists $T_N$ such that for $T>T_N$ we have 
\begin{equation}
   \abs{ \frac{1}{T}\,\sum_{t=1}^{T}\; \bm{\psi}_t(\bm{\beta}^{o*})} < \bm{\epsilon}. 
\end{equation}
Under Assumptions~\ref{assump:magnitude_of_coefs} and \ref{assump:narrow_Theta}, we know that $\bm{\beta}^{o*}$ is an interior point in $\Theta$. Since $\bm{\beta}^{o*}$ is an interior point in $\Theta$ and $\frac{1}{T}\,\sum_{t=1}^{T}\; \bm{\psi}_t(\cdot)$ is uniformly continuous in $\Theta$ by Proposition~\ref{prop:assumpB2_continuity}, there exists $\delta > 0$ and
a ball $B(\bm{\beta}^{o*},\delta):= \{\bm{\beta}\in\Theta|\norm{\bm{\beta}^{o*} - \bm{\beta}} <\delta \}$ such that 
\begin{equation}
    \abs{ \frac{1}{T}\,\sum_{t=1}^{T}\; \bm{\psi}_t(\bm{\beta}^{o*}) -  \frac{1}{T}\,\sum_{t=1}^{T}\; \bm{\psi}_t(\bm{\beta})} < \frac{1}{2}\bm{\epsilon}. 
\end{equation}
Moreover, there exist $\bm{\beta}_1, \bm{\beta}_2 \in B(\bm{\beta}^{o*},\delta)$ such that if $ \abs{ \frac{1}{T}\,\sum_{t=1}^{T}\; \bm{\psi}_t(\bm{\beta}^{o*})} \neq \bm{0}$, then
\begin{align}
     & 
     \bm{0}
     < \frac{1}{T}\,\sum_{t=1}^{T}\; \bm{\psi}_t(\bm{\beta}^{o*}) -  \frac{1}{T}\,\sum_{t=1}^{T}\; \bm{\psi}_t(\bm{\beta}_1) 
     \leq \frac{1}{2} \abs{ \frac{1}{T}\,\sum_{t=1}^{T}\; \bm{\psi}_t(\bm{\beta}^{o*})}
     \\
     - & \frac{1}{2} \abs{ \frac{1}{T}\,\sum_{t=1}^{T}\; \bm{\psi}_t(\bm{\beta}^{o*})}
     \leq \frac{1}{T}\,\sum_{t=1}^{T}\; \bm{\psi}_t(\bm{\beta}^{o*}) -  \frac{1}{T}\,\sum_{t=1}^{T}\; \bm{\psi}_t(\bm{\beta}_2) 
     < \bm{0}
\end{align}
which results in
\begin{align}
          & \frac{1}{T}\,\sum_{t=1}^{T}\; \bm{\psi}_t(\bm{\beta}^{o*}) 
          - \frac{1}{2} \abs{ \frac{1}{T}\,\sum_{t=1}^{T}\; \bm{\psi}_t(\bm{\beta}^{o*})}
        \leq \frac{1}{T}\,\sum_{t=1}^{T}\; \bm{\psi}_t(\bm{\beta}_1) 
        < \frac{1}{T}\,\sum_{t=1}^{T}\; \bm{\psi}_t(\bm{\beta}^{o*}) 
        \\
        & \frac{1}{T}\,\sum_{t=1}^{T}\; \bm{\psi}_t(\bm{\beta}^{o*}) 
        < \frac{1}{T}\,\sum_{t=1}^{T}\; \bm{\psi}_t(\bm{\beta}_2)
        \leq \frac{1}{T}\,\sum_{t=1}^{T}\; \bm{\psi}_t(\bm{\beta}^{o*}) 
          + \frac{1}{2} \abs{ \frac{1}{T}\,\sum_{t=1}^{T}\; \bm{\psi}_t(\bm{\beta}^{o*})}.
\end{align}
Hence we find $\abs{ \frac{1}{T}\,\sum_{t=1}^{T}\; \bm{\psi}_t(\bm{\beta}_1) }   < \abs{ \frac{1}{T}\,\sum_{t=1}^{T}\; \bm{\psi}_t(\bm{\beta}^{o*})} < \bm{\epsilon}$ or $\abs{ \frac{1}{T}\,\sum_{t=1}^{T}\; \bm{\psi}_t(\bm{\beta}_2) }   < \abs{ \frac{1}{T}\,\sum_{t=1}^{T}\; \bm{\psi}_t(\bm{\beta}^{o*})} < \bm{\epsilon}$. Continue this process, and we can find a sequence of points in $\Theta$ has decreasing values of $\abs{ \frac{1}{T}\,\sum_{t=1}^{T}\; \bm{\psi}_t(\cdot)}$. There exists a subsequence of the resulted point sequence and the limit of the subsequence is $\widehat{\bm{\beta}}^{\text{mle}} $ in $\Theta$ with $\abs{ \frac{1}{T}\,\sum_{t=1}^{T}\; \bm{\psi}_t(\widehat{\bm{\beta}}^{\text{mle}})} = 0$. Since $\lim_{T\rightarrow\infty} \frac{1}{T}\,\sum_{t=1}^{T}\; \bm{\psi}_t(\bm{\beta}^{o*}) = \bm{0}$, then we obtain $ \lim_{T\rightarrow\infty}\widehat{\bm{\beta}}^{\text{mle}} = \bm{\beta}^{o*}$ and conclude this proof.

\end{proof}

\subsubsection{Proof of Theorem~\ref{thm:MLE_asym}}
\begin{proof}
Apply the Taylor expansion of $\sum_{t=1}^{T}\; \bm{\psi}_t(\widehat{\bm{\beta}}^{\text{mle}})$ at $\bm{\beta}^{o*}$ and the mean value theorem, we obtain
\begin{equation}
   \frac{1}{\sqrt{T}}\sum_{t=1}^{T}\; \bm{\psi}_t(\widehat{\bm{\beta}}^{\text{mle}}) = \frac{1}{\sqrt{T}} \sum_{t=1}^{T}\; \bm{\psi}_t( \bm{\beta}^{o*} ) + \frac{1}{T}H_{\mathcal{L}}(\overline{\bm{\beta}}) \,\sqrt{T}\left( \widehat{\bm{\beta}}^{\text{mle}} -  \bm{\beta}^{o*} \right),
\end{equation}
where $\overline{\bm{\beta}}$ is between $\widehat{\bm{\beta}}^{\text{mle}}$ and $\bm{\beta}^{o*}$.
Since $\sum_{t=1}^{T}\; \bm{\psi}_t(\widehat{\bm{\beta}}^{\text{mle}}) = 0$, the above expansion results in 
\begin{equation}
    \left( \frac{1}{T}H_{\mathcal{L}}(\overline{\bm{\beta}})\right)^{-1}\,\frac{1}{\sqrt{T}} \sum_{t=1}^{T}\; \bm{\psi}_t( \bm{\beta}^{o*} )  = \sqrt{T}\left( \widehat{\bm{\beta}}^{\text{mle}} -  \bm{\beta}^{o*} \right).
\end{equation}
$\frac{1}{T}H_{\mathcal{L}}(\cdot)$ is uniformly continuous in $\Theta$, which can be proved analogously to the proof of Proposition~\ref{prop:assumpB2_continuity} with knowing $y_t^2 \in O_p(1)$ thanks to $0<k_t(\bm{\beta}^{o*})<0.5$. By the continuous mapping theorem and knowing $\lim_{T\rightarrow\infty} \widehat{\bm{\beta}}^{\text{mle}} = \bm{\beta}^{o*}$ from Theorem~\ref{thm:MLE_consistency}, we obtain that
$$
\lim_{T\rightarrow\infty} \left( \frac{1}{T}H_{\mathcal{L}}(\overline{\bm{\beta}})\right)^{-1} = \left( \frac{1}{T}H_{\mathcal{L}}( \bm{\beta}^{o*} )\right)^{-1}.
$$
From Assumptions \ref{assump:limiting_score_Z} and \ref{assump:weakconverge_HL} we know that
\begin{equation}
\left\{
    \begin{aligned}
    & 
    \frac{1}{\sqrt{T}} \sum^T_{t=1}\, \bm{\psi}_t(\bm{\beta}^{o*})
  \overset{\mathcal{D}}{\sim}
  S_{\psi}
    \\
    & 
    \frac{1}{T}H_{\mathcal{L}}(\bm{\beta}^{o*}) \overset{\mathcal{D}}{\sim} \Omega_{H}
    \end{aligned}
    \right.
\end{equation}
and thus by Slutsky's theorem, we obtain that
\begin{equation}
    \sqrt{T}\left( \bm{\beta}^{o*} - \widehat{\bm{\beta}}^{\text{mle}} \right) \overset{\mathcal{D}}{\sim} \Omega_{H}^{-1} \; S_{\psi}.
\end{equation}
\end{proof}

\subsection{ALMLE - Proofs}
\label{append:ALMLE}
\subsubsection{Proof of Theorem~\ref{thm:ALMLE_nece_cond}} 
\begin{proof}
We prove this theorem by contradiction. In the proof below, we show that truly inactive predictors have a non-zero probability to get selected since $\lambda_{k, T}, \lambda_{\sigma, T}\in O(T^{\frac{1}{2}})$ if there is no $\widehat{\bm{\beta}}^{\text{al}}(\lambda_{k, T}, \lambda_{\sigma, T})$ with $\lambda_{k, T}, \lambda_{\sigma, T}\in O(T^{\frac{1}{2}})$ such that the condition~\eqref{eq:Thm3_DETcondition} is met.

We set $w_{k,i}, w_{\sigma,j}$, e.g. using the MLE in section~\ref{sec:MLE_Inference}, such that $\sqrt{T}(\frac{1}{w_{k,i}} - \beta_{1,i}^{*o})=O_p(1) $ and $\sqrt{T}(\frac{1}{w_{\sigma,j}} - \beta_{2,j}^{*o})=O_p(1) $, for $i=1,\ldots, p$, $j=1,\ldots, p+1$.

Let us go over the tuning parameter grid $\{(\lambda_{k, T}, \lambda_{\sigma, T}): \lambda_{k, T}\in S_{\lambda_{k,T}}, \lambda_{\sigma, T}\in S_{\lambda_{\sigma,T}},\}$. $\lambda_{k, T,\max}$ and $\lambda_{\sigma, T,\max}$ are chosen large enough such that $\widehat{\bm{\beta}}^{\text{al}}(\lambda_{k, T,\max}, \lambda_{\sigma, T,\max})= \bm{0}$, i.e., no predictors get selected. By the Karush-Kuhn-Tucker (KKT) optimality condition and with $(\lambda_{k, T,\max}, \lambda_{\sigma, T,\max})$, we have
\begin{equation}
\label{eq:KKT_ALMLE_maxtuning}
\left\{
    \begin{aligned}
   \abs{\frac{\partial \mathcal{L} }{\partial \widehat{\beta}_{1,i}^{\text{al}} }} 
    & \leq \lambda_{k, T,\max}\, w_{k,i} ,
    & (1,i)\in \mathcal{A}_{k}\cup \mathcal{A}_{k}^c ,
    \\
   \abs{\frac{\partial \mathcal{L} }{\partial \widehat{\beta}_{2,j}^{\text{al}} }} 
    & \leq \lambda_{\sigma, T,\max}\, w_{\sigma,j} ,
    & (2,j)\in \mathcal{A}_{\sigma} \cup \mathcal{A}_{\sigma}^c ,
    \\
    \end{aligned}
    \right.
\end{equation}
where we denote $\mathcal{A}_{k}: = \left\{ (1,i): i\geq 1, \beta_{1,i}^{*o} \neq 0\right\}$, $\mathcal{A}_{k}^c := \left\{ (1,i): i\geq 1,  \beta_{1,i}^{*o} = 0\right\}$, $\mathcal{A}_{\sigma} := \left\{ (2,j): j\geq 1,  \beta_{2,j}^{*o} \neq 0\right\}$, and $\mathcal{A}_{\sigma}^c := \left\{ (2,j): j\geq 1, \beta_{2,j}^{*o} = 0\right\}$. We rewrite $\frac{\partial \mathcal{L} }{\partial \widehat{\bm{\beta}}^{\text{al}}(\lambda_{k, T,\max}, \lambda_{\sigma, T,\max}) }$ using the mean value theorem because $\mathcal{L}(\cdot; \{ y_t \}, \{ \bm{z}_t^* \})$ is twice continuously differentiable, and substitute this rewriting into the above KKT condition and get
\begin{equation}
\label{eq:KKT_ALMLE_maxtuning_expansion}
\left\{
    \begin{aligned}
    \abs{
    \frac{\partial \mathcal{L} }{\partial \beta_{1,i}^{*o} } 
    - \sum_{i_g=0}^{p}\frac{\partial^2 \mathcal{L} }{\partial \widehat{\beta}_{1,i}^{\text{al}} \partial \widehat{\bm{\beta}}_{1,i_g}^{\text{al}'} }\overline{\bm{\beta}}_{1,i_g}^{(2\,n_{\lambda_k}\,n_{\lambda_{\sigma}})} 
    - \sum_{j_g=0}^{p+1}\frac{\partial^2 \mathcal{L} }{\partial \widehat{\beta}_{1,i}^{\text{al}} \partial \widehat{\bm{\beta}}_{2,j_g}^{\text{al}'}  }\overline{\bm{\beta}}_{2,j_g}^{(2\,n_{\lambda_k}\,n_{\lambda_{\sigma}})} 
    }   
    \leq & \lambda_{k, T,\max}\, w_{k,i} ,
    \\
    & (1,i)\in \mathcal{A}_{k} \cup \mathcal{A}_{k}^c ,
    \\
    \abs{
    \frac{\partial \mathcal{L} }{\partial \beta_{2,j}^{*o} } 
    - \sum_{i_g=0}^{p}\frac{\partial^2 \mathcal{L} }{\partial \widehat{\beta}_{2,j}^{\text{al}} \partial \widehat{\bm{\beta}}_{1,i_g}^{\text{al}'}  }\overline{\bm{\beta}}_{1,i_g}^{(2\,n_{\lambda_k}\,n_{\lambda_{\sigma}}-1)}  
    - \sum_{j_g=0}^{p+1}\frac{\partial^2 \mathcal{L} }{\partial \widehat{\beta}_{2,j}^{\text{al}} \partial \widehat{\bm{\beta}}_{2,j_g}^{\text{al}'}  }\overline{\bm{\beta}}_{2,j_g}^{(2\,n_{\lambda_k}\,n_{\lambda_{\sigma}}-1)}  
    }
    \leq & \lambda_{\sigma, T,\max}\, w_{\sigma,j} ,
    \\
    & (2,j)\in \mathcal{A}_{\sigma} \cup \mathcal{A}_{\sigma}^c,
    \\
    \end{aligned}
    \right.
\end{equation}
where $\overline{\bm{\beta}}_{1,i_g}^{(2\,n_{\lambda_k}\,n_{\lambda_{\sigma}})}, \overline{\bm{\beta}}_{1,i_g}^{(2\,n_{\lambda_k}\,n_{\lambda_{\sigma}} - 1)} \in \left[\bm{\beta}_{1,i_g}^{*o},\widehat{\bm{\beta}}^{\text{al}}_{1,i_g}(\lambda_{k, T,\max}, \lambda_{\sigma, T,\max})\right]$, and $\overline{\bm{\beta}}_{1,j_g}^{(2\,n_{\lambda_k}\,n_{\lambda_{\sigma}})}, \overline{\bm{\beta}}_{1,j_g}^{(2\,n_{\lambda_k}\,n_{\lambda_{\sigma}} - 1)} \in [\bm{\beta}_{2,j_g}^{*o},$ $\widehat{\bm{\beta}}^{\text{al}}_{2,j_g}(\lambda_{k, T,\max}, \lambda_{\sigma, T,\max})]$. In fact, we can find a pair of $(\lambda_{k, T,\max},\lambda_{\sigma, T,\max})$ in $O(T)$.

Lower the tuning parameters from $(\lambda_{k, T,\max},\lambda_{\sigma, T,\max})$, and $\widehat{\bm{\beta}}^{\text{al}}$ is able to screen variables. Truly inactively predictors always meet the inequality of the KKT condition with $\frac{\lambda_{k, T}}{\sqrt{T}}\rightarrow \infty$ and $\frac{\lambda_{\sigma, T}}{\sqrt{T}}\rightarrow \infty$, thereby not being selected. Coming to $\lambda_{k, T}, \lambda_{\sigma, T} \in O(T^{\frac{1}{2}})$, we get that $\lambda_{k, T}\, w_{k,i}=O_p(T^{\frac{1}{2}}) $ for $(1,i)\in \mathcal{A}_{k}$; $\lambda_{k, T}\, w_{k,i}=O_p(T) $ for $(1,i)\in \mathcal{A}_{k}^c$; $\lambda_{\sigma, T}\, w_{\sigma,j}=O_p(T^{\frac{1}{2}}) $ for $(2,j)\in \mathcal{A}_{\sigma}$; $\lambda_{\sigma, T}\, w_{\sigma,i}=O_p(T) $ for $(2,i)\in \mathcal{A}_{\sigma}^c$. If with any $\lambda_{k, T}\in O(T^{\frac{1}{2} })$ and $\lambda_{\sigma, T}\in O(T^{\frac{1}{2} })$, there is no $\widehat{\bm{\beta}}^{\text{al}}(\lambda_{k, T}, \lambda_{\sigma, T})$ such that the condition~\eqref{eq:Thm3_DETcondition} is met, then there exists $(1,i_a) \in \mathcal{A}_{k}$ or $(2,j_a) \in \mathcal{A}_{\sigma}$ such that 
\begin{equation}
\abs{\frac{\partial \mathcal{L} }{\partial \widehat{\beta}_{1,i_a}^{\text{al}} }} 
   \leq \lambda_{k, T}\, w_{k,i_a}
\end{equation}
or
\begin{equation}
\abs{\frac{\partial \mathcal{L} }{\partial \widehat{\beta}_{2,j_a}^{\text{al}} }} 
    \leq \lambda_{\sigma, T}\, w_{\sigma,j_a}
\end{equation}
but $\widehat{\beta}_{1,i_a}^{\text{al}} = 0$ or $\widehat{\beta}_{2,j_a}^{\text{al}} = 0$ correspondingly for any $\lambda_{k, T}, \lambda_{\sigma, T}\in O(T^{\frac{1}{2}})$. 

Under the condition~\eqref{eq:Thm3_DETcondition} is broken, if $\abs{\frac{\partial \mathcal{L} }{\partial \widehat{\beta}_{1,i_a}^{\text{al}} }} \leq \lambda_{k, T}\, w_{k,i_a}$ with $\widehat{\beta}_{1,i_a}^{\text{al}} = 0$ and any $\lambda_{k, T}, \lambda_{\sigma, T}\in O(T^{\frac{1}{2}})$, we get that $\lambda_{k, T}\, w_{k,i}=O_p(T^{\frac{1}{2} }) $ for $(1,i)\in \mathcal{A}_{k}$ and $\lambda_{k, T}\, w_{k,i}=O_p(T) $ for $(1,i)\in \mathcal{A}_{k}^c$. Then each truly inactively predictor but correlated with the $i_a$-th predictor of the shape model gains a non-zero probability to get selected due to
\begin{equation}
    P\left\{ \abs{\frac{\partial \mathcal{L} }{\partial \widehat{\beta}_{1,i}^{\text{al}} }} 
   \geq \lambda_{k, T}\, w_{k,i}, i\in  \mathcal{A}_{k}^c \middle| \widehat{\beta}_{1,i_a}^{\text{al}} = 0,\; \lambda_{k, T}\in O(T^{\frac{1}{2}}),\; \lambda_{\sigma, T}\in O(T^{\frac{1}{2}})  
    \right\} \neq 0.
\end{equation}
The above probability is obtained easily by rewriting $\abs{\frac{\partial \mathcal{L} }{\partial \widehat{\beta}_{1,i}^{\text{al}} }}$ using the mean value theorem as done in~\eqref{eq:KKT_ALMLE_maxtuning_expansion} and obtaining that $\abs{\frac{\partial \mathcal{L} }{\partial \widehat{\beta}_{1,i}^{\text{al}} }} \in O_p(T)$. For each $\widehat{\beta}_{1,i}^{\text{al}}\neq 0, i\in  \mathcal{A}_{k}^c$ with $\lambda_{k, T}\in O(T^{\frac{1}{2}})$, the bias of $\widehat{\beta}_{1,i}^{\text{al}}$ is $O_p(1)$. Lowering $\lambda_{k, T}$, it follows that $\widehat{\beta}_{1,i}^{\text{al}} = O_p(T^{-\gamma})$ meets the equality of the KKT condition with $\lambda_{k, T}\in O(T^{\frac{1}{2} - \gamma})$ and remains selected for $\gamma \in [0, \frac{1}{2}]$.

The same reasoning applies if $\abs{\frac{\partial \mathcal{L} }{\partial \widehat{\beta}_{2,j_a}^{\text{al}} }} \leq \lambda_{\sigma, T}\, w_{\sigma,j_a}$ with $\widehat{\beta}_{1,j_a}^{\text{al}} = 0$ and any $\lambda_{k, T}, \lambda_{\sigma, T}\in O(T^{\frac{1}{2}})$, we get that each truly inactively predictor but correlated with the $j_a$-th predictor of the scale model then gains a non-zero probability to get selected, i.e.,
\begin{equation}
    P\left\{ \abs{\frac{\partial \mathcal{L} }{\partial \widehat{\beta}_{2,j}^{\text{al}} }} 
    \geq \lambda_{\sigma, T}\, w_{\sigma,j}, j\in  \mathcal{A}_{\sigma}^c \middle| \widehat{\beta}_{2,j_a}^{\text{al}} = 0 \; \lambda_{k, T}\in O(T^{\frac{1}{2}}),\; \lambda_{\sigma, T}\in O(T^{\frac{1}{2}})  
    \right\} \neq 0.
\end{equation}
Therefore, the model selection consistency cannot be achieved if the condition~\eqref{eq:Thm3_DETcondition} is not met, and we finish the proof.
\end{proof}
\subsubsection{Proof of Theorem~\ref{thm:ALMLE_tal_oracle}}
\begin{proof}
In this proof, we first prove the selection consistency of $\widehat{\bm{\beta}}^{k,al}$ on truly active predictors of the shape model, i.e., 
\begin{equation}
\label{eq:prove_k_over_sele}
    \lim_{T\rightarrow \infty}\,P\left\{
    \widehat{\beta}^{k,al}_{1,i} \neq 0, \forall (1,i)\in\mathcal{A}_k
    \right\}=1
\end{equation}
which is also equivalent to $\lim_{T\rightarrow \infty}\,P\left\{
\mathcal{A}^{k,al}_{k,T} \supseteq \mathcal{A}_k\right\}=1$ with
$\mathcal{A}^{k,al}_{k,T}:=\{(1,i): \widehat{\bm{\beta}}^{k,al}_{1,i} \neq 0\}$. It follows that this over-selection possibility for the shape model due to $\widehat{\bm{\beta}}^{k,al}$ is proved to be curbed to zero by $\widehat{\bm{\beta}}^{tal}$, and the oracle property of $\widehat{\bm{\beta}}^{tal}$ is obtained lastly.

Firstly, we rewrite the objective function to obtain $\widehat{\bm{\beta}}^{k,al}$ as follows: according to Assumption~\ref{assump:kmle_tal}, there exists $\bm{\beta}^{k,o}$, and
\begin{equation}
\label{eq:Vk_nu_ALV2}
    V^{(k)}(\bm{\nu}^{(k)})
     =
    - \mathcal{L}\left(\bm{\beta}^{k,o} + \frac{\bm{\nu}^{(k)}}{\sqrt{T}}; \{ y_t \}, \{ \bm{z}_t^* \}\right) + \lambda_{k, T}\sum_{i=1}^{p} \widetilde{w}_{k,i}\abs{\beta_{1,i}^{k,o} + \frac{\nu_{1,i}^{(k)}}{\sqrt{T}}},
\end{equation}
where $\bm{\nu}^{(k)}:=[\nu_{1,0}^{(k)},\ldots,\nu_{1,p}^{(k)}, \nu_{2,0}^{(k)},0,\ldots,0]'\in\mathbb{R}^{p+2}\times\bm{0}_{(p+1)}$. We obtain that $\widehat{\bm{\nu}}^{(k)} = \argmin{\bm{\nu}^{(k)}\in\mathbb{R}^{p+2}\times\bm{0}_{(p+1)}} V^{(k)}(\bm{\nu}^{(k)})  - V^{(k)}(\bm{0})$ such that equivalently
\begin{equation}
    \widehat{\bm{\beta}}^{k,al} = \bm{\beta}^{k,o} + \frac{\widehat{\bm{\nu}}^{(k)}}{\sqrt{T}}.
\end{equation}
Specifically,
\begin{equation}
\label{eq:k_al_objective_difference_with_trueObj}
    \begin{aligned}
    & 
    V^{(k)}(\bm{\nu}^{(k)})  
    - 
    V^{(k)}(\bm{0})
    \\
    & =
    - \left(
    \mathcal{L}(\bm{\beta}^{k,o} + \frac{\bm{\nu}^{(k)}}{\sqrt{T}}; \{ y_t \}, \{ \bm{z}_t^* \})
    -
    \mathcal{L}\left(\bm{\beta}^{k,o}; \{ y_t \}, \{ \bm{z}_t^* \}\right) 
    \right)
    \\
    & +
    \frac{\lambda_{k, T}}{\sqrt{T}}\sum_{i=1}^{p} \widetilde{w}_{k,i}\,\sqrt{T}\left( \abs{\beta_{1,i}^{k,o} + \frac{\nu_{1,i}^{(k)}}{\sqrt{T}}} - \abs{\beta_{1,i}^{k,o}}
    \right).
    \end{aligned}
\end{equation}
According to Assumption~\ref{assump:kmle_tal} and Theorem~\ref{thm:MLE_asym}, we have that $\sqrt{T}(\widehat{\beta}^{k,mle}_{1,i} - \beta_{1,i}^{k,o}) = O_p(1)$ and $\sqrt{T}(\widehat{\beta}^{mle}_{1,i} -\beta_{1,i}^{*o}) = O_p(1)$. Considering that
\begin{equation}
 \lim_{T\rightarrow\infty}  \sqrt{T}\left( \abs{\beta_{1,i}^{k,o} + \frac{\nu_{1,i}^{(k)}}{\sqrt{T}}} - \abs{\beta_{1,i}^{k,o}}
    \right)
    =\left\{
    \begin{aligned}
    & 
    \nu_{1,i}^{(k)}\,\,\text{sgn}(\beta_{1,i}^{k,o}),\quad \text{if}\, \beta_{1,i}^{k,o}\neq 0,
    \\
    & 
    \abs{\nu_{1,i}^{(k)}},\quad \text{if}\, \beta_{1,i}^{k,o}= 0,
    \end{aligned}
    \right.
\end{equation}
and under Assumption~\ref{assump:NEW_tuning} we obtain that
\begin{equation}
\label{eq:kal_penality_stochastic_orders}
    \frac{\lambda_{k, T}}{\sqrt{T}} \widetilde{w}_{k,i} \sqrt{T}\left( \abs{\beta_{1,i}^{k,o} + \frac{\nu_{1,i}^{(k)}}{\sqrt{T}}} - \abs{\beta_{1,i}^{k,o}}
    \right)
    =\left\{
    \begin{aligned}
    & 
    O_p(T^{-\gamma_1})\, \nu_{1,i}^{(k)}\,\,\text{sgn}(\beta_{1,i}^{k,o}),\quad \text{if}\, \beta_{1,i}^{*o}\neq 0,
    \\
    & 
    O_p(T^{\frac{1}{2}-\gamma_1})\, \nu_{1,i}^{(k)}\,\,\text{sgn}(\beta_{1,i}^{k,o}),\quad \text{if}\, \beta_{1,i}^{*o} = 0\;\text{but}\;\beta_{1,i}^{k,o}\neq 0,
    \\
    & 
    O_p(T^{1 -\gamma_1})\, \abs{\nu_{1,i}^{(k)}},\quad \text{if}\, \beta_{1,i}^{*o} = 0\;\text{and}\;\beta_{1,i}^{k,o}= 0.
    \end{aligned}
    \right.
\end{equation}
Also, we know that
\begin{equation}
    - \left(
    \mathcal{L}(\bm{\beta}^{k,o} + \frac{\bm{\nu}^{(k)}}{\sqrt{T}}; \{ y_t \}, \{ \bm{z}_t^* \})
    -
    \mathcal{L}\left(\bm{\beta}^{k,o}; \{ y_t \}, \{ \bm{z}_t^* \}\right) 
    \right) = O_p(1)\,\max\{ (\nu_{1,i}^{(k)})^2\}
\end{equation}
by Taylor's expansion and the tightness of $\frac{1}{T}H_{\mathcal{L}}(\cdot) $ in $\Theta$ as shown in the proof of Theorem~\ref{thm:MLE_asym}.
Substitute~Eq.~\eqref{eq:kal_penality_stochastic_orders} back into Eq.~\eqref{eq:k_al_objective_difference_with_trueObj}, and by Slutsky's theorem we obtain that 
\begin{equation}
\label{eq:Vk_nuk_limit}
 \resizebox{.9\hsize}{!}{$
   V^{(k)}(\nu_{1,i}^{(k)})  
    - 
    V^{(k)}(\bm{0})
    =
    \left\{
  \begin{aligned}
  &
    O_p(T), & \text{if}\quad\frac{\nu_{1,i}^{(k)}}{\sqrt{T}} = O_p(1), \forall (1,i)\in\mathcal{A}^{k,al}_{k}\,\text{and}\,  \nu_{1,i}^{(k)} = O_p(1),\forall (1,i)\notin\mathcal{A}^{k,al}_{k};
  \\
  &
    O_p(T^{1-2\gamma}), & \text{if}\quad \nu_{1,i}^{(k)}  = O_p(T^{\frac{1}{2} - \gamma}), \forall (1,i)\in\mathcal{A}^{k,al}_{k}\,\text{with}\,\, 0<2\gamma<\gamma_1\,\text{and}\,  \nu_{1,i}^{(k)} = O_p(1),\forall (1,i)\notin\mathcal{A}^{k,al}_{k};
  \\
  &
    O_p(T^{1-\gamma_1}), & \text{if}\quad \nu_{1,i}^{(k)}  = O_p(T^{\frac{1}{2} - \gamma}), \forall (1,i)\in\mathcal{A}^{k,al}_{k}\,\text{with}\,\, \frac{\gamma_1}{2}\leq\gamma\leq\frac{1}{2}\,\text{and}\,  \nu_{1,i}^{(k)} = O_p(1),\forall (1,i)\notin\mathcal{A}^{k,al}_{k};
  \\
  &
    O_p(T^{1-2\gamma}), & \text{if}\quad \nu_{1,i}^{(k)}  = O_p(T^{\frac{1}{2} - \gamma}), \forall (1,i)\in\mathcal{A}^{k,al}_{k}\,\text{with}\,\, 0<\gamma<\gamma_1\,\text{and}\,  \nu_{1,i}^{(k)} = 0,\forall (1,i)\notin\mathcal{A}^{k,al}_{k};
  \\
  & 
 O_p(T^{1-\gamma-\gamma_1}), & \text{if}\quad \nu_{1,i}^{(k)}  = O_p(T^{\frac{1}{2} - \gamma}), \forall (1,i)\in\mathcal{A}^{k,al}_{k}\,\text{with}\,\, \gamma_1\leq\gamma\leq\frac{1}{2}\,\text{and}\,  \nu_{1,i}^{(k)} = 0,\forall (1,i)\notin\mathcal{A}^{k,al}_{k},
  \end{aligned}  
  \right.
  $}
\end{equation}
where $\mathcal{A}^{k,al}_{k}:=\{(1,i): {\beta}^{k,o}_{1,i} \neq 0\}$. 
Therefore, we have that 
\begin{equation}
    \widehat{\nu}_{1,i}^{(k)}  = O_p(1), \forall (1,i)\in\mathcal{A}^{k,al}_{k}\;\text{and}\,  \widehat{\nu}_{1,i}^{(k)} = 0,\forall (1,i)\notin\mathcal{A}^{k,al}_{k},
\end{equation}
minimizing $ V^{(k)}(\bm{\nu}^{(k)}) - V^{(k)}(\bm{0})$, and hence that
\begin{equation}
    \lim_{T\rightarrow\infty} P\left\{  \widehat{\beta}^{k,al}_{1,i} =\beta_{1,i}^{k,o} + \frac{\widehat{\nu}_{1,i}^{(k)}}{\sqrt{T}}\neq 0, \forall (1,i)\in \mathcal{A}^{k,al}_{k}\right\} = P\left\{ \beta_{1,i}^{k,o} \neq 0, \forall (1,i)\in \mathcal{A}^{k,al}_{k}\right\} = 1.
\end{equation}
Thus we obtain $ \lim_{T\rightarrow \infty}\,P\left\{
    \widehat{\beta}^{k,al}_{1,i} \neq 0, \forall (1,i)\in\mathcal{A}_k
    \right\}=1$.

Secondly, we show the asymptotic behaviour of $\widehat{\bm{\beta}}^{tal}$. 
Write $\bm{\beta}=\bm{\beta}^{o*} + \frac{\bm{\nu}}{\sqrt{T}}$ for $\bm{\beta}\in\Theta$ with $\bm{\nu}:=[\nu_{1,0},\ldots,\nu_{1,p}, \nu_{2,0},$ $\ldots,\nu_{2,p+1}]'\in\mathbb{R}^{2p+3}$, and rewrite the objective function to obtain $\widehat{\bm{\beta}}^{tal}$ as follows:
\begin{equation}
\label{eq:V_nu_ALV2}
    V(\bm{\nu})  
    = 
     - \mathcal{L}\left(\bm{\beta}^{o*} + \frac{\bm{\nu}}{\sqrt{T}}; \{ y_t \}, \{ \bm{z}_t^* \}\right)  
     + 
     \widehat{\lambda}_{k, T}\sum_{i=1}^{p} \widetilde{w}_{k,i}\abs{\beta_{1,i}^{o*} + \frac{\nu_{1,i}}{\sqrt{T}}}
     + 
     \lambda_{\sigma, T}\sum_{j=1}^{p+1} \widetilde{w}_{\sigma,j}\abs{\beta_{2,j}^{o*} + \frac{\nu_{2,j}}{\sqrt{T}}}.
\end{equation}
We obtain $\widehat{\bm{\nu}}:= \argmin{\bm{\nu}\in\mathbb{R}^{2p+3}} V(\bm{\nu}) - V(\bm{0})$ such that
\begin{equation}
    \widehat{\bm{\beta}}^{tal} = \bm{\beta}^{*o} + \frac{\widehat{\bm{\nu}}}{\sqrt{T}}.
\end{equation}
Specifically,
\begin{equation}
    \begin{aligned}
    & 
    V(\bm{\nu})  
    - 
    V(\bm{0})
    \\
    & =
    - \left(
    \mathcal{L}(\bm{\beta}^{o*} + \frac{\bm{\nu}}{\sqrt{T}}; \{ y_t \}, \{ \bm{z}_t^* \})
    -
    \mathcal{L}\left(\bm{\beta}^{o*}; \{ y_t \}, \{ \bm{z}_t^* \}\right) 
    \right)
    \\
    & +
    \frac{ \widehat{\lambda}_{k, T}}{\sqrt{T}}\sum_{i=1}^{p} \widetilde{w}_{k,i}\,\sqrt{T}\left( \abs{\beta_{1,i}^{o*} + \frac{\nu_{1,i}}{\sqrt{T}}} - \abs{\beta_{1,i}^{o*}}
    \right)
     +
    \frac{\lambda_{\sigma, T}}{\sqrt{T}}\sum_{j=1}^{p+1} \widetilde{w}_{\sigma,j}\,\sqrt{T}\left( \abs{\beta_{2,j}^{o*} + \frac{\nu_{2,j}}{\sqrt{T}}} - \abs{\beta_{2,j}^{o*}}
    \right).
    \end{aligned}
\end{equation}
Analogously to~\eqref{eq:Vk_nuk_limit}, we can also get that
\begin{equation}
\label{eq:V_nu_limit}
   \lim_{T\rightarrow\infty} V(\bm{\nu})  
    - 
    V(\bm{0})
    =
    \left\{
  \begin{aligned}
  &
  O_p(1),
  &
  \text{if}\; \nu_{1,i}=O_p(1), \nu_{2,j}=O_p(1),\; \text{for}\;\forall (1,i)\in\mathcal{A}_{k}, \forall (2,j)\in\mathcal{A}_{\sigma},\;
  \\
  & 
  & \text{and}\,\nu_{1,i}=0, \nu_{2,j}=0,\; \text{for}\;\forall (1,i)\notin\mathcal{A}_{k}, \forall (2,j)\notin\mathcal{A}_{\sigma},
  \\
  &
  \infty,\qquad & \text{otherwise of}\quad \bm{\nu}=O_p(\bm{1}). 
  \end{aligned}  
  \right.
\end{equation}
Therefore, minimizing $V(\bm{\nu})  -  V(\bm{0})$ is equivalent to minimizing $  - \left(
    \mathcal{L}(\bm{\beta}^{o*} + \frac{\bm{\nu}}{\sqrt{T}}; \{ y_t \}, \{ \bm{z}_t^* \})
    -
    \mathcal{L}\left(\bm{\beta}^{o*}; \{ y_t \}, \{ \bm{z}_t^* \}\right) 
    \right)$ with $\nu_{1,i}=0, \nu_{2,j}=0$ for any $ (1,i)\notin\mathcal{A}_{k}, (2,j)\notin\mathcal{A}_{\sigma}$, which leads to
\begin{equation}
        \sqrt{T}\left(\widehat{\bm{\beta}}^{tal} - \bm{\beta}^{o*} \right)= \widehat{\bm{\nu}} = \argmin{\bm{\nu}\in \mathbb{R}^{2p+3}\cap\{\nu_{1,i}=0.\nu_{2,j} =0, (1,i)\notin\mathcal{A}_{k}, (2,j)\notin\mathcal{A}_{\sigma}\} } 
    V(\bm{\nu})  -  V(\bm{0})
\end{equation}
as $T\rightarrow\infty$ and thereby we obtain the asymptotic behavior of $\widehat{\bm{\nu}}$ by Theorem~\ref{thm:MLE_asym} with the restriction of $\lim_{T\rightarrow\infty}\widehat{\nu}_{1,i}=0, \lim_{T\rightarrow\infty}\widehat{\nu}_{2,j}=0$ for any $ (1,i)\notin\mathcal{A}_{k}, (2,j)\notin\mathcal{A}_{\sigma}$.
Moreover, 
\begin{equation}
\resizebox{\hsize}{!}{$
\begin{aligned}
        &
        \lim_{T\rightarrow\infty}P\left\{ \mathcal{A}_{T}^{tal} = \mathcal{A}\right\}
        \\
        & =
        \lim_{T\rightarrow\infty} P\left\{ \{
        \widehat{\beta}^{tal}_{1,i} =\beta_{1,i}^{*o} + \frac{\widehat{\nu}_{1,i}}{\sqrt{T}}\neq 0, \widehat{\beta}^{tal}_{2,j} =\beta_{2,j}^{*o} + \frac{\widehat{\nu}_{2,j}}{\sqrt{T}}\neq 0, \forall (1,i)\in \mathcal{A}_{k}, \forall (2,j)\in \mathcal{A}_{\sigma} \}
        \cap\{
        \widehat{\beta}^{tal}_{1,i} = 0, \widehat{\beta}^{tal}_{2,j} = 0, \forall (1,i)\notin \mathcal{A}_{k},\forall (2,j)\notin \mathcal{A}_{\sigma} \}
        \right\} 
        \\
        & =
       P\left\{ \{
        \beta_{1,i}^{*o}\neq 0, \beta_{2,j}^{*o}\neq 0, \forall (1,i)\in \mathcal{A}_{k}, \forall (2,j)\in \mathcal{A}_{\sigma} \}
        \cap\{
        \beta_{1,i}^{*o} = 0, \beta_{2,j}^{*o} = 0, \forall (1,i)\notin \mathcal{A}_{k},\forall (2,j)\notin \mathcal{A}_{\sigma} \}
        \right\} = 1 
\end{aligned}
$}
\end{equation}
which concludes that $\widehat{\bm{\beta}}^{tal} $ is model selection consistent. Together with the limiting distribution of  $\sqrt{T}\left(\widehat{\bm{\beta}}^{tal} - \bm{\beta}^{o*} \right)$, we conclude that $\widehat{\bm{\beta}}^{tal}$ has the oracle property.
\end{proof}

\end{document}